\documentclass[conference]{IEEEtran}
\IEEEoverridecommandlockouts
\ifCLASSINFOpdf
\else
\fi
\usepackage{amsthm}

\newtheorem{theorem}{Theorem}
\newtheorem{lemma}{Lemma}

\usepackage[utf8]{inputenc}

\usepackage{amsthm} 
\usepackage{amsfonts}
\usepackage{amsmath}
\usepackage{algorithmicx}
\usepackage{graphicx}
\usepackage{xcolor}
\usepackage{mathtools}
\usepackage{booktabs}

\usepackage[boxed,ruled,vlined, linesnumbered]{algorithm2e}
\SetVlineSkip{0.3pt}
\SetNlSty{bfseries}{\color{black}}{}
\SetCommentSty{normalfont}
\SetKwComment{Comment}{$\triangleright$\ }{}
\newcommand{\AComment}[1]{\Comment*[r]{#1}}
\newcommand{\ACommentL}[1]{\Comment*[l]{#1}}

\usepackage{mathpartir}
\usepackage{color}
\usepackage{listings}
\usepackage{tikz}
\usetikzlibrary{arrows}
\usetikzlibrary{automata}
\usetikzlibrary{positioning}
\usepackage{stmaryrd}
\usepackage{comment}
\usepackage{multirow}
\usepackage{wrapfig}
\usepackage{subcaption}

\usepackage{pgfplots}
\pgfplotsset{width=10cm,compat=1.9}
\usepgfplotslibrary{dateplot}

\def\st{\phantom{X}}

\def\Set{\mathsf{Set}}

\def\<{\langle}
\def\>{\rangle}

\def\S{\mathcal{S}}

\def\llet{\textbf{let}}
\def\lin{\textbf{in}}

\def\Implements{\textbf{Implements}}
\def\Vars{\textbf{vars}}
\def\Uses{\textbf{uses}}

\def\request{\textbf{request}}
\def\call{\textbf{call}}
\def\response{\textbf{response}}

\def\true{\mathsf{true}}
\def\false{\mathsf{false}}

\def\TOB{\mathsf{TotalOrderBroadcast}}

\def\for_all{\textsf{for all}}

\def\rb{\mathit{rb}}

\def\abeb{\mathit{abeb}}
\def\brd{\mathit{brd}}
\def\complained{\mathit{complained}}

\def\tob{\mathit{tob}}
\def\apl{\mathit{apl}}

\def\lem{\mathit{le}}

\def\process{\mathit{process}}
\def\return{\mathit{return}}

\def\interBroadcast{\mathit{inter\text{-}broadcast}}
\def\Inter{\mathit{Inter}}

\def\batchSize{\mathit{batch\text{-}size}}
\def\execute{\mathit{execute}}

\def\complain{\mathit{complain}}
\def\newLeader{\mathit{new\text{-}leader}}

\def\join{\mathit{join}}
\def\joined{\mathit{joined}}
\def\leave{\mathit{leave}}
\def\left{\mathit{left}}
\def\RequestJoin{\mathit{RequestJoin}}
\def\RequestLeave{\mathit{RequestLeave}}
\def\CurrState{\mathit{CurrState}}

\def\Transaction{\mathit{Trans}}
\def\reconfigure{\mathit{reconfigure}}
\def\Reconfig{\mathit{Reconfig}}

\def\send{\mathit{send}}
\def\deliver{\mathit{deliver}}
\def\broadcast{\mathit{broadcast}}
\def\Recs{\mathit{Recs}}

\def\Agg{\mathit{Agg}}
\def\echo{\mathit{Echo}}
\def\ready{\mathit{Ready}}

\def\sendRecs{\mathit{send\text{-}recs}}

\def\certs{\mathit{certs}}
\def\preCerts{\mathit{p\text{-}certs}}

\def\operations{\mathit{operations}}
\def\preOps{\mathit{p\text{-}ops}}
\def\ops{\mathit{ops}}

\def\mym{\mathit{my}\text{-}\!\mathit{m}}

\def\nextLeader{\mathit{next\text{-}leader}}
\def\timer{\mathit{timer}}

\def\clientTimer{\mathit{client\text{-}timer}}
\def\Local{\mathit{Local}}
\def\Complaint{\mathit{Complaint}}
\def\LComplaint{\mathit{LComplaint}}
\def\RComplaint{\mathit{RComplaint}}
\def\cs{\mathit{cs}}
\def\cn{\mathit{cn}}
\def\rcn{\mathit{rcn}}
\def\ts{\mathit{ts}}

\def\C{\mathit{C}}
\def\leader{\mathit{leader}}
\def\recs{\mathit{recs}}
\def\st{\mathit{state}}

\def\echoed{\mathit{echoed}}
\def\readied{\mathit{readied}}
\def\delivered{\mathit{delivered}}
\def\valid{\mathit{valid}}
\def\highValid{\mathit{high\text{-}valid}}
\def\Valid{\mathit{Valid}}

\def\rc{\mathit{rc}}

\definecolor{forestgreen}{rgb}{0.0, 0.27, 0.13}
\definecolor{tropicalrainforest}{rgb}{0.0, 0.46, 0.37}

\newcommand{\mllater}[1]{}
\newcommand{\todo}[1]{\textcolor{purple}{{[Todo:~#1]}}}

\usepackage{hyperref}
\hypersetup{
    linktoc=all,
    colorlinks=false,       
    linkcolor=black,          
    linkbordercolor=black,
    citecolor=black,        
    filecolor=black,      
    urlcolor=black,           
    pdfborderstyle={/S/U/W 0}
}

\usepackage{tabularx}

\usepackage{caption}
\usepackage{graphicx}

\DeclareMathAlphabet{\mathpzc}{OT1}{pzc}{m}{it}

\DeclareUnicodeCharacter{00AC}{\ensuremath{\neg}}                     
\DeclareUnicodeCharacter{00B1}{\ensuremath{\pm}}                      
\DeclareUnicodeCharacter{00B2}{\ensuremath{^2}}                       
\DeclareUnicodeCharacter{00B3}{\ensuremath{^3}}                       
\DeclareUnicodeCharacter{00B7}{\ensuremath{\cdot}}                    
\DeclareUnicodeCharacter{00B9}{\ensuremath{^1}}                       
\DeclareUnicodeCharacter{00D7}{\ensuremath{\times}}                   
\DeclareUnicodeCharacter{02B0}{\ensuremath{^h}}                       
\DeclareUnicodeCharacter{02B3}{\ensuremath{^r}}                       
\DeclareUnicodeCharacter{02E2}{\ensuremath{^s}}                       
\DeclareUnicodeCharacter{0302}{\ensuremath{\hat{\phantom{x}}}}        
\DeclareUnicodeCharacter{0332}{\ensuremath{\underline{\phantom{x}}}}  
\DeclareUnicodeCharacter{0308}{\ensuremath{\ddot{\phantom{x}}}}       
\DeclareUnicodeCharacter{0393}{\ensuremath{\Gamma}}                   
\DeclareUnicodeCharacter{0394}{\ensuremath{\Delta}}                   
\DeclareUnicodeCharacter{0398}{\ensuremath{\Theta}}                   
\DeclareUnicodeCharacter{039B}{\ensuremath{\Lambda}}                  
\DeclareUnicodeCharacter{039E}{\ensuremath{\Xi}}                      
\DeclareUnicodeCharacter{03A0}{\ensuremath{\Pi}}                      
\DeclareUnicodeCharacter{03A3}{\ensuremath{\Sigma}}                   

\DeclareUnicodeCharacter{03A4}{\ensuremath{\mathrm{T}}}  

\DeclareUnicodeCharacter{03A5}{\ensuremath{\Upsilon}}   
\DeclareUnicodeCharacter{03A6}{\ensuremath{\Phi}}       
\DeclareUnicodeCharacter{03A8}{\ensuremath{\Psi}}       
\DeclareUnicodeCharacter{03A9}{\ensuremath{\Omega}}     
\DeclareUnicodeCharacter{03B1}{\ensuremath{\alpha}}     
\DeclareUnicodeCharacter{03B2}{\ensuremath{\beta}}      
\DeclareUnicodeCharacter{03B3}{\ensuremath{\gamma}}     
\DeclareUnicodeCharacter{03B4}{\ensuremath{\delta}}     
\DeclareUnicodeCharacter{03B5}{\ensuremath{\epsilon}}   
\DeclareUnicodeCharacter{03B6}{\ensuremath{\zeta}}      
\DeclareUnicodeCharacter{03B7}{\ensuremath{\eta}}       
\DeclareUnicodeCharacter{03B8}{\ensuremath{\theta}}     
\DeclareUnicodeCharacter{03B9}{\ensuremath{\iota}}      
\DeclareUnicodeCharacter{03BA}{\ensuremath{\kappa}}     
\DeclareUnicodeCharacter{03BB}{\ensuremath{\lambda}}    
\DeclareUnicodeCharacter{03BC}{\ensuremath{\mu}}        
\DeclareUnicodeCharacter{03BD}{\ensuremath{\nu}}        
\DeclareUnicodeCharacter{03BE}{\ensuremath{\xi}}        
\DeclareUnicodeCharacter{03C0}{\ensuremath{\pi}}        
\DeclareUnicodeCharacter{03C1}{\ensuremath{\rho}}       
\DeclareUnicodeCharacter{03C2}{\ensuremath{\varsigma}}  
\DeclareUnicodeCharacter{03C3}{\ensuremath{\sigma}}     
\DeclareUnicodeCharacter{03C4}{\ensuremath{\tau}}       
\DeclareUnicodeCharacter{03C5}{\ensuremath{\upsilon}}   
\DeclareUnicodeCharacter{03C6}{\ensuremath{\varphi}}    
\DeclareUnicodeCharacter{03C7}{\ensuremath{\chi}}       
\DeclareUnicodeCharacter{03C8}{\ensuremath{\psi}}       
\DeclareUnicodeCharacter{03C9}{\ensuremath{\omega}}     
\DeclareUnicodeCharacter{03D1}{\ensuremath{\vartheta}}  

\DeclareUnicodeCharacter{0432}{\ensuremath{\mathrm{PDF\;TODO}}}  
\DeclareUnicodeCharacter{0435}{\ensuremath{\mathrm{PDF\;TODO}}}  
\DeclareUnicodeCharacter{0438}{\ensuremath{\mathrm{PDF\;TODO}}}  
\DeclareUnicodeCharacter{043C}{\ensuremath{\mathrm{PDF\;TODO}}}  
\DeclareUnicodeCharacter{0440}{\ensuremath{\mathrm{PDF\;TODO}}}  
\DeclareUnicodeCharacter{0442}{\ensuremath{\mathrm{PDF\;TODO}}}  
\DeclareUnicodeCharacter{041F}{\ensuremath{\mathrm{PDF\;TODO}}}  
\DeclareUnicodeCharacter{0BA8}{\ensuremath{\mathrm{PDF\;TODO}}}  
\DeclareUnicodeCharacter{0BBF}{\ensuremath{\mathrm{PDF\;TODO}}}  
\DeclareUnicodeCharacter{1100}{\ensuremath{\mathrm{PDF\;TODO}}}  
\DeclareUnicodeCharacter{11F9}{\ensuremath{\mathrm{PDF\;TODO}}}  

\DeclareUnicodeCharacter{1D48}{\ensuremath{^d}}  
\DeclareUnicodeCharacter{1D50}{\ensuremath{^m}}  
\DeclareUnicodeCharacter{1D56}{\ensuremath{^p}}  
\DeclareUnicodeCharacter{1D62}{\ensuremath{_i}}  
\DeclareUnicodeCharacter{1D63}{\ensuremath{_r}}  
\DeclareUnicodeCharacter{1D64}{\ensuremath{_u}}  
\DeclareUnicodeCharacter{1D9C}{\ensuremath{^c}}

\DeclareUnicodeCharacter{2032}{\ensuremath{\text{\textasciiacute}}}  

\DeclareUnicodeCharacter{2070}{\ensuremath{^0}}  
\DeclareUnicodeCharacter{2074}{\ensuremath{^4}}  
\DeclareUnicodeCharacter{2075}{\ensuremath{^5}}  
\DeclareUnicodeCharacter{2076}{\ensuremath{^6}}  
\DeclareUnicodeCharacter{2077}{\ensuremath{^7}}  
\DeclareUnicodeCharacter{2078}{\ensuremath{^8}}  
\DeclareUnicodeCharacter{2079}{\ensuremath{^9}}  
\DeclareUnicodeCharacter{207A}{\ensuremath{^+}}  
\DeclareUnicodeCharacter{207B}{\ensuremath{^-}}  
\DeclareUnicodeCharacter{207F}{\ensuremath{^n}}  
\DeclareUnicodeCharacter{2080}{\ensuremath{_0}}  
\DeclareUnicodeCharacter{2081}{\ensuremath{_1}}  
\DeclareUnicodeCharacter{2082}{\ensuremath{_2}}  
\DeclareUnicodeCharacter{2083}{\ensuremath{_3}}  
\DeclareUnicodeCharacter{2084}{\ensuremath{_4}}  
\DeclareUnicodeCharacter{2085}{\ensuremath{_5}}  
\DeclareUnicodeCharacter{2086}{\ensuremath{_6}}  
\DeclareUnicodeCharacter{2087}{\ensuremath{_7}}  
\DeclareUnicodeCharacter{2088}{\ensuremath{_8}}  
\DeclareUnicodeCharacter{2089}{\ensuremath{_9}}  
\DeclareUnicodeCharacter{208A}{\ensuremath{_+}}  
\DeclareUnicodeCharacter{208B}{\ensuremath{_-}}  
\DeclareUnicodeCharacter{2091}{\ensuremath{_e}}  
\DeclareUnicodeCharacter{2095}{\ensuremath{_h}}  
\DeclareUnicodeCharacter{2096}{\ensuremath{_k}}  
\DeclareUnicodeCharacter{2097}{\ensuremath{_l}}  
\DeclareUnicodeCharacter{2098}{\ensuremath{_m}}  
\DeclareUnicodeCharacter{2099}{\ensuremath{_n}}  
\DeclareUnicodeCharacter{209A}{\ensuremath{_p}}  
\DeclareUnicodeCharacter{209B}{\ensuremath{_s}}  
\DeclareUnicodeCharacter{209C}{\ensuremath{_t}}  

\DeclareUnicodeCharacter{2113}{\ensuremath{\ell}} 

\DeclareUnicodeCharacter{2115}{\ensuremath{\mathbbm{N}}}  
\DeclareUnicodeCharacter{211A}{\ensuremath{\mathbbm{Q}}}  
\DeclareUnicodeCharacter{211D}{\ensuremath{\mathbbm{R}}}  
\DeclareUnicodeCharacter{2124}{\ensuremath{\mathbbm{Z}}}  

\DeclareUnicodeCharacter{2135}{\ensuremath{\aleph}}           
\DeclareUnicodeCharacter{2190}{\ensuremath{\leftarrow}}       
\DeclareUnicodeCharacter{2192}{\ensuremath{\rightarrow}}      
\DeclareUnicodeCharacter{2194}{\ensuremath{\leftrightarrow}}  

\DeclareUnicodeCharacter{21A0}{\ensuremath{\twoheadrightarrow}}  

\DeclareUnicodeCharacter{21A6}{\ensuremath{\mapsto}}  
\DeclareUnicodeCharacter{21A6}{\ensuremath{\mapsto}}  

\DeclareUnicodeCharacter{21BE}{\ensuremath{\restriction}}  

\DeclareUnicodeCharacter{21D0}{\ensuremath{\Leftarrow}}       
\DeclareUnicodeCharacter{21D2}{\ensuremath{\Rightarrow}}      
\DeclareUnicodeCharacter{21D4}{\ensuremath{\Leftrightarrow}}  

\DeclareUnicodeCharacter{2200}{\ensuremath{\forall}}      
\DeclareUnicodeCharacter{2203}{\ensuremath{\exists}}      
\DeclareUnicodeCharacter{2204}{\ensuremath{\not\exists}}  
\DeclareUnicodeCharacter{2205}{\ensuremath{\emptyset}}    
\DeclareUnicodeCharacter{2208}{\ensuremath{\in}}          
\DeclareUnicodeCharacter{2209}{\ensuremath{\not\in}}      
\DeclareUnicodeCharacter{220B}{\ensuremath{\ni}}          
\DeclareUnicodeCharacter{220C}{\ensuremath{\not\ni}}      

\DeclareUnicodeCharacter{220E}{{\tiny \ensuremath{\blacksquare}}} 

\DeclareUnicodeCharacter{220F}{{\tiny \ensuremath{\prod}}}  
\DeclareUnicodeCharacter{2211}{\ensuremath{\sum}}           
\DeclareUnicodeCharacter{2218}{\ensuremath{\circ}}          
\DeclareUnicodeCharacter{2219}{\ensuremath{\bullet}}        
\DeclareUnicodeCharacter{221E}{\ensuremath{\infty}}         
\DeclareUnicodeCharacter{2223}{\ensuremath{\mid}}           
\DeclareUnicodeCharacter{2225}{\ensuremath{\parallel}}      
\DeclareUnicodeCharacter{2227}{\ensuremath{\wedge}}         
\DeclareUnicodeCharacter{2228}{\ensuremath{\vee}}           

\DeclareUnicodeCharacter{2229}{\ensuremath{\cap}}           

\DeclareUnicodeCharacter{222A}{\ensuremath{\cup}}           
\DeclareUnicodeCharacter{222B}{\ensuremath{\int}}           
\DeclareUnicodeCharacter{2234}{\ensuremath{\therefore}}     

\DeclareUnicodeCharacter{2237}{\ensuremath{\dblcolon}} 

\newcommand{\dotminus}{\mathop{\mbox{$-^{\hspace{-.5em}\cdot}\,$}}}
\DeclareUnicodeCharacter{2238}{\ensuremath{\dotminus}}  

\DeclareUnicodeCharacter{223C}{\ensuremath{\sim}}       
\DeclareUnicodeCharacter{2243}{\ensuremath{\simeq}}     
\DeclareUnicodeCharacter{2245}{\ensuremath{\cong}}      
\DeclareUnicodeCharacter{2248}{\ensuremath{\approx}}    
\DeclareUnicodeCharacter{2250}{\ensuremath{\doteq}}     

\DeclareUnicodeCharacter{2254}{\ensuremath{\coloneqq}}  

\newcommand{\eqq}{\stackrel{\text{\tiny ?}}{=}}
\DeclareUnicodeCharacter{225F}{\ensuremath{\eqq}}  

\DeclareUnicodeCharacter{2260}{\ensuremath{\neq}}         
\DeclareUnicodeCharacter{2261}{\ensuremath{\equiv}}       
\DeclareUnicodeCharacter{2262}{\ensuremath{\not\equiv}}   
\DeclareUnicodeCharacter{2264}{\ensuremath{\le}}          
\DeclareUnicodeCharacter{2265}{\ensuremath{\ge}}          
\DeclareUnicodeCharacter{226E}{\ensuremath{\nless}}       
\DeclareUnicodeCharacter{226F}{\ensuremath{\ngtr}}        
\DeclareUnicodeCharacter{2270}{\ensuremath{\nleq}}        
\DeclareUnicodeCharacter{2271}{\ensuremath{\ngeq}}        
\DeclareUnicodeCharacter{227A}{\ensuremath{\prec}}        
\DeclareUnicodeCharacter{2280}{\ensuremath{\nprec}}        
\DeclareUnicodeCharacter{227C}{\ensuremath{\preceq}}      
\DeclareUnicodeCharacter{2282}{\ensuremath{\subset}}      
\DeclareUnicodeCharacter{2284}{\ensuremath{\not\subset}} 
\DeclareUnicodeCharacter{2283}{\ensuremath{\supset}}      
\DeclareUnicodeCharacter{2286}{\ensuremath{\subseteq}}    
\DeclareUnicodeCharacter{228E}{\ensuremath{\uplus}}       
\DeclareUnicodeCharacter{2291}{\ensuremath{\sqsubseteq}}  
\DeclareUnicodeCharacter{2293}{\ensuremath{\sqcap}}       
\DeclareUnicodeCharacter{2294}{\ensuremath{\sqcup}}       
\DeclareUnicodeCharacter{2295}{\ensuremath{\oplus}}       
\DeclareUnicodeCharacter{22A2}{\ensuremath{\vdash}}       
\DeclareUnicodeCharacter{22A4}{\ensuremath{\top}}         
\DeclareUnicodeCharacter{22A5}{\ensuremath{\bot}}         
\DeclareUnicodeCharacter{22C0}{\ensuremath{\bigwedge}}    
\DeclareUnicodeCharacter{22C2}{\ensuremath{\bigcap}}      
\DeclareUnicodeCharacter{22C3}{\ensuremath{\bigcup}}      
\DeclareUnicodeCharacter{22EE}{\ensuremath{\vdots}}       
\DeclareUnicodeCharacter{22EF}{\ensuremath{\cdots}}       

\DeclareUnicodeCharacter{22F0}{\ensuremath{\iddots}} 

\DeclareUnicodeCharacter{230A}{\ensuremath{\lfloor}}  
\DeclareUnicodeCharacter{230B}{\ensuremath{\rfloor}}  
\DeclareUnicodeCharacter{2308}{\ensuremath{\lceil}}   
\DeclareUnicodeCharacter{2309}{\ensuremath{\rceil}}   
\DeclareUnicodeCharacter{2329}{\ensuremath{\langle}}  
\DeclareUnicodeCharacter{232A}{\ensuremath{\rangle}}  

\DeclareUnicodeCharacter{23CE}{\ensuremath{\Return}}  

\DeclareUnicodeCharacter{2460}{\ensuremath{\text{\ding{172}}}} 
\DeclareUnicodeCharacter{2461}{\ensuremath{\text{\ding{173}}}} 
\DeclareUnicodeCharacter{2462}{\ensuremath{\text{\ding{174}}}} 
\DeclareUnicodeCharacter{2463}{\ensuremath{\text{\ding{175}}}} 
\DeclareUnicodeCharacter{2464}{\ensuremath{\text{\ding{176}}}} 
\DeclareUnicodeCharacter{2465}{\ensuremath{\text{\ding{177}}}} 
\DeclareUnicodeCharacter{2466}{\ensuremath{\text{\ding{178}}}} 
\DeclareUnicodeCharacter{2467}{\ensuremath{\text{\ding{179}}}} 
\DeclareUnicodeCharacter{2468}{\ensuremath{\text{\ding{180}}}} 

\DeclareUnicodeCharacter{25C5}{\ensuremath{\triangleleft}}  

\DeclareUnicodeCharacter{2660}{\ensuremath{\spadesuit}} 
\DeclareUnicodeCharacter{266D}{\ensuremath{\flat}}      
\DeclareUnicodeCharacter{266F}{\ensuremath{\sharp}}     

\DeclareUnicodeCharacter{2713}{\ensuremath{\checkmark}}  

\DeclareUnicodeCharacter{27E6}{\ensuremath{\llbracket}}  
\DeclareUnicodeCharacter{27E7}{\ensuremath{\rrbracket}}  

\DeclareUnicodeCharacter{27E8}{\ensuremath{\langle}}          
\DeclareUnicodeCharacter{27E9}{\ensuremath{\rangle}}          
\DeclareUnicodeCharacter{27F6}{\ensuremath{\longrightarrow}}  

\DeclareUnicodeCharacter{2983}{\ensuremath{\mathrm{PDF\;TODO}}}  
\DeclareUnicodeCharacter{2984}{\ensuremath{\mathrm{PDF\;TODO}}}  
\DeclareUnicodeCharacter{2985}{\ensuremath{\mathrm{PDF\;TODO}}}  
\DeclareUnicodeCharacter{2986}{\ensuremath{\mathrm{PDF\;TODO}}}  

\DeclareUnicodeCharacter{2987}{\ensuremath{\llparenthesis}}  
\DeclareUnicodeCharacter{2988}{\ensuremath{\rrparenthesis}}  

\DeclareUnicodeCharacter{2216}{\ensuremath{\setminus}}

\DeclareUnicodeCharacter{2A06}{\ensuremath{\bigcup}}  
\DeclareUnicodeCharacter{2C7C}{\ensuremath{_j}}       

\DeclareUnicodeCharacter{D7B0}{\ensuremath{\mathrm{PDF\;TODO}}} 

\DeclareUnicodeCharacter{1D7D9}{\ensuremath{\mathbbm{1}}}

\DeclareUnicodeCharacter{1D4D0}{\ensuremath{\mathcal{A}}} 
\DeclareUnicodeCharacter{1D4D1}{\ensuremath{\mathcal{B}}}  
\DeclareUnicodeCharacter{1D4D2}{\ensuremath{\mathcal{C}}}  
\DeclareUnicodeCharacter{1D4D3}{\ensuremath{\mathcal{D}}}  
\DeclareUnicodeCharacter{1D4D4}{\ensuremath{\mathcal{E}}}  
\DeclareUnicodeCharacter{1D4D5}{\ensuremath{\mathcal{F}}}  
\DeclareUnicodeCharacter{1D4D6}{\ensuremath{\mathcal{G}}}  
\DeclareUnicodeCharacter{1D4D7}{\ensuremath{\mathcal{H}}}  
\DeclareUnicodeCharacter{1D4D8}{\ensuremath{\mathcal{I}}}  
\DeclareUnicodeCharacter{1D4D9}{\ensuremath{\mathcal{J}}}  
\DeclareUnicodeCharacter{1D4DA}{\ensuremath{\mathcal{K}}}  
\DeclareUnicodeCharacter{1D4DB}{\ensuremath{\mathcal{L}}}  
\DeclareUnicodeCharacter{1D4DC}{\ensuremath{\mathcal{M}}}  
\DeclareUnicodeCharacter{1D4DD}{\ensuremath{\mathcal{N}}}  
\DeclareUnicodeCharacter{1D4DE}{\ensuremath{\mathcal{O}}}  
\DeclareUnicodeCharacter{1D4DF}{\ensuremath{\mathcal{P}}}

\DeclareUnicodeCharacter{1D4E0}{\ensuremath{\mathcal{Q}}}  
\DeclareUnicodeCharacter{1D4E1}{\ensuremath{\mathcal{R}}}  
\DeclareUnicodeCharacter{1D4E2}{\ensuremath{\mathcal{S}}}  
\DeclareUnicodeCharacter{1D4E3}{\ensuremath{\mathcal{T}}}  
\DeclareUnicodeCharacter{1D4E4}{\ensuremath{\mathcal{U}}}  
\DeclareUnicodeCharacter{1D4E5}{\ensuremath{\mathcal{V}}}  
\DeclareUnicodeCharacter{1D4E6}{\ensuremath{\mathcal{W}}}  
\DeclareUnicodeCharacter{1D4E7}{\ensuremath{\mathcal{X}}}  
\DeclareUnicodeCharacter{1D4E8}{\ensuremath{\mathcal{Y}}}  
\DeclareUnicodeCharacter{1D4E9}{\ensuremath{\mathcal{Z}}}  

\DeclareUnicodeCharacter{25C2}{\ensuremath{\triangleleft}}  
\DeclareUnicodeCharacter{25B9}{\ensuremath{\triangleright}}
\DeclareUnicodeCharacter{2286}{\ensuremath{\subseteq}}    
\DeclareUnicodeCharacter{2217}{\em}

\DeclareUnicodeCharacter{1D62}{\ensuremath{_i}}  
\DeclareUnicodeCharacter{2C7C}{\ensuremath{_j}}
\DeclareUnicodeCharacter{2096}{\ensuremath{_k}}

\DeclareUnicodeCharacter{276A}{\mbox{$($}}  
\DeclareUnicodeCharacter{276B}{\mbox{$)$}}

\DeclareUnicodeCharacter{261E}{\todo}      

\DeclareUnicodeCharacter{1D539}{\mathbb{B}}      
\DeclareUnicodeCharacter{2102}{\mathbb{C}}      
\DeclareUnicodeCharacter{2098}{_m}      
\DeclareUnicodeCharacter{2115}{\mathbb{N}}

\def\ie{\textit{i.e.}}

\def\self{\textbf{\textit{self}}}
\def\i{\textbf{\textit{i}}}

\newcounter{MyAlgoLine}
\newcommand{\alglinenoRestore}{\setcounter{AlgoLine}{\value{MyAlgoLine}}}
\newcommand{\alglinenoStore}{\setcounter{MyAlgoLine}{\value{AlgoLine}}}

\def\algsize{\footnotesize}

\IfFileExists{./predefs_me.tex}{
\input{predefs_me}
}{}

\begin{document}

\setcounter{page}{1}

\title{Hamava: Fault-tolerant Reconfigurable Geo-Replication on Heterogeneous Clusters}

\author{
\IEEEauthorblockN{
Tejas Mane\IEEEauthorrefmark{1},
Xiao Li\IEEEauthorrefmark{1},
Mohammad Sadoghi\IEEEauthorrefmark{2},
Mohsen Lesani\IEEEauthorrefmark{3}
}
\IEEEauthorblockA{\IEEEauthorrefmark{1}University of California, Riverside,  
\IEEEauthorrefmark{2}University of California, Davis,
\IEEEauthorrefmark{3}University of California, Santa Cruz
}
}

\maketitle

\begin{abstract}
Fault-tolerant replicated database systems consume significantly 
less energy than the compute-intensive proof-of-work blockchain.
Thus, they are promising technologies for the building blocks that assemble global financial infrastructure.
To facilitate global scaling, clustered replication protocols are
essential in orchestrating nodes into clusters based on proximity.
However, existing approaches often assume a homogeneous and fixed model in which the number of nodes across clusters is the same and fixed, and often limited to a fail-stop fault model. 
This paper presents heterogeneous and reconfigurable clustered replication for the general environment with arbitrary failures.
In particular, we present \textsc{Hamava}, a fault-tolerant reconfigurable geo-replication that allows dynamic membership: replicas are allowed to join and leave clusters.
We formally state and prove the safety and liveness properties of the protocol. Furthermore, our replication protocol is consensus-agnostic, meaning each cluster can utilize any local replication mechanism. 
In our comprehensive evaluation, we instantiate our replication with both HotStuff and BFT-SMaRt.
Experiments on geo-distributed deployments on Google Cloud demonstrates that members of clusters can be reconfigured without significantly affecting  transaction processing,
and
that heterogeneity of clusters may significantly improve throughput.
\end{abstract}

\vspace{-1mm}

\section{Introduction} 
Blockchains such as
Bitcoin \cite{nakamoto2008peer}
and Ethereum \cite{wood2014ethereum}
maintain a global replicated ledger
on untrusted hosts.
However, 
they suffer from a few drawbacks, including
high energy consumption,
partitions \cite{tran2020stealthier,gervais2016security,saad2023three},
and stake and vote centralization
\cite{wahrstatter2024blockchain}.
Byzantine replicated systems 
such as
PBFT \cite{castro1999practical} and its numerous following variants \cite{miller2016honey,yin2019hotstuff,spiegelman2022bullshark,androulaki2018hyperledger, stathakopoulou2019mir,amiri2019caper,gupta2019proof,ruan2021blockchains,ruan2021lineagechain,hs1,bedrock}
can maintain consistent replications in the presence of malicious nodes.
More interestingly,
these techniques
avoid energy-intensive proof-of-work hashing.
Therefore, they are an appealing technology
to serve as the global financial infrastructure.
Thus, several projects, such as
Hyperledger \cite{androulaki2018hyperledger},
Solida \cite{abraham2016solida}, 
Tendermint \cite{buchman2016tendermint},
Casper \cite{buterin2017casper}, 
Algorand \cite{gilad2017algorand}, 
OmniLedger \cite{kokoris2018omniledger} 
RapidChain \cite{zamani2018rapidchain},
and 
\cite{cachin2023quorum}
deployed Byzantine replication protocols to manage blockchains.

However,
Byzantine replication protocols need to be improved on two fronts: 
\emph{scale} and \emph{dynamic} membership.
They often require rounds of message-passing between nodes;
therefore, they tend not to scale to many or distant nodes.
Further, 
their membership is often fixed.
In fact, 
the resulting blockchains are called permissioned
since the nodes are fixed and initially known.

To scale Byzantine replication across the globe,
projects such as
Steward \cite{amir2008steward}
and
ResilientDB \cite{gupta13resilientdb,gupta2021rcc}
and Narwhal \cite{danezis2022narwhal}
try to 
use global communication judiciously,
and
decrease global
in favor of
local communication.
They allow neighboring nodes to form clusters. This enables each cluster to order transactions locally while reducing the need for inter-cluster communication to reach an agreement on the global order.
Since the communication among members of a cluster is local,
clusters can maintain high throughput and low latency.
Further, 
coordination is divided between clusters,
and they can order transactions in parallel.

However, existing clustered replication protocols are homogeneous and fixed.
The number of nodes is the same across clusters.
Further, nodes cannot join or leave clusters.
A global financial system needs to be \emph{heterogeneous}: 
different regions might have different numbers of active nodes.
More importantly, decentralization promised \emph{dynamic} 
membership: active nodes should be able to churn.
This is the property that proof-of-work blockchains, such as 
Bitcoin, observe, allowing any incentivized nodes to join and keep the system running.
Reconfiguration has been studied for non-clustered replication 
\cite{duan2022foundations, jehl2015smartmerge, duan2014bchain, lamport2010reconfiguring, lamport1998part, guerraoui2010next}
but it remains an open problem for clustered replication.
Can we have the best of both worlds?
\emph{Can we have the energy efficiency, equity, and scalability of clustered Byzantine replication, and the dynamic membership of proof-of-work?
Can we have reconfigurable clustered replication?}

Reconfiguring a clustered replication system
\emph{without compromising security} is a challenging task.
If the reconfigurations are not propagated uniformly to all clusters,
correct replicas (\ie, processes) might accept invalid messages or miss valid ones.
Inconsistent views of membership may lead to violation of both safety and liveness.
Byzantine replicated systems can often tolerate one-third of replicas to be Byzantine.
Thus, if a message is received from more than one-third of replicas,
at least one correct replica must have sent it;
therefore, the message can be trusted.
Consider a cluster $C_{old}$ that is not informed of new additions to another cluster $C_{new}$.
The cluster $C_{old}$'s record of one-third is less than the actual one-third for $C_{new}$.
Therefore, the Byzantine replicas
in $C_{new}$ can form a group that is larger than the old one-third,
and 
can make $C_{old}$ accept a invalid message.
On the flip side,
$C_{new}$ might miss messages from $C_{old}$.
Since $C_{old}$ thinks that $C_{new}$ is smaller,
in order to communicate a message,
$C_{old}$ might send a message to 
an insufficient number of replicas in $C_{new}$.
Thus, the Byzantine replicas in $C_{new}$ can censor the message for 
other replicas in that cluster, hindering liveness.
Uniform propagation of reconfigurations is particularly challenging
when the leader 
simultaneously changes.

In this paper, we present \textsc{Hamava}, a reconfigurable clustered replication that can tolerate arbitrary faults.
It allows replicas to be divided into multiple \emph{heterogeneous} clusters, and further allows \emph{dynamic} membership for clusters: replicas can join and leave a cluster. This clustered design further reduces the cost of inter-cluster communication by allowing each cluster to independently reach an agreement on its membership and order transactions locally and only propagate the local decisions globally.
Reconfigurations are processed efficiently in parallel to transactions. 
Since the reconfigurations received in a round $r$ take effect for 
the next round $r+1$,
they do not need to be ordered in round $r$. 
Thus, instead of processing them in sequence through the consensus that orders transactions,
they are \emph{aggregated} into a set, and processed together. 
We present the reconfiguration protocol and
formally state and prove its safety and liveness.

Reconfiguration of heterogenous clusters introduces nuances 
that affect ordering and executing transactions.
In particular, the inter-cluster communication primitive and 
the remote leader fault detection mechanism must have up-to-date 
knowledge of the size of the local and remote clusters in order 
to ensure safety and liveness.

\textsc{Hamava} is a meta-protocol that is
\emph{agnostic} to the local replication protocol.
We implement \textsc{Hamava}
for HotStuff \cite{yin2019hotstuff}
in C++,
and 
for BFT-SMaRt \cite{bessani2014state}
in Java.
We deployed our systems
on geo-distributed clusters in multiple regions of 
Google Cloud.
The experimental results show that 
heterogeneous geo-distributed deployments
significantly improve throughput,
can be reconfigured without affecting transaction processing,
and 
can gracefully tolerate Byzantine failures.

In short, this paper makes the following contributions: 
 \begin{itemize}
   \item 
   We present \textsc{Hamava}, a \emph{reconfigurable clustered replication protocol} 
   that allows replicas to dynamically join and leave clusters safely and efficiently.
   \item \textsc{Hamava} entails \emph{Heterogeneous clustered replication} to support clusters with varying sizes.
   Thus, \textsc{Hamava} includes a novel inter-cluster communication primitive and remote leader replacement mechanism.

 \item  \emph{Formal specification and proof of safety and liveness properties} of \textsc{Hamava}
 including dynamic reconfiguration in Byzantine 
 heterogeneous environments.
    \item  \emph{Implementation and experimental results} that demonstrate that \textsc{Hamava} is agnostic to the the local consensus (e.g., HotStuff and BFT-SMaRt).
   Thorough experiments with the resulting systems (\textsc{Ava-Hotstuff} and \textsc{Ava-Bftsmart}) show that \textsc{Hamava} supports efficient and fault-tolerant global replication and reconfiguration,
   and that heterogeneity improves performance.
\end{itemize}

\vspace{-1mm}

\section{Overview}
\label{sec:overview}
We describe the system and threat model, and illustrate the protocol with diagrams and representative executions.


\textbf{System and Threat Model. \ }
A replicated system consists of a set $P$ of replicas
that are partitioned into 
clusters $C = \{N_{1}, .., C_{N}\}$.
Clients can send requests to any replica to execute 
operations of two different types: transactions and reconfigurations.
The state is replicated at each replica.
A replica can be correct or Byzantine. 
A Byzantine replica can fail arbitrarily including but not limited to crash failures, sending conflicting messages, dropping messages, and impersonating other Byzantine replicas.
We assume that 
at any time
in each cluster, at most one-third of replicas can be Byzantine,
\ie, at most $f$ out of $3f+1$ replicas can be Byzantine.
(This paper does not consider problems orthogonal to Byzantine fault tolerance such as access control or sybil resistance \cite{rajabi2023feasibility,tran2011optimal,zhang2014you}.)
We further assume that each replica can be identified by its public key, and 
that replicas are computationally bound, and cannot subvert standard cryptographic primitives.
Thus, replicas can communicate with authenticated links.
We consider a partially synchronous network \cite{dwork1988consensus}:
after an unknown global stabilization time,
messages between any pair of correct replicas 
will be eventually delivered within a bounded delay.
Replicas communicate with authenticated perfect links $\apl$,
and authenticated best-effort broadcast $\abeb$ 
which simply abstracts $\apl$ to send a message to all replicas.
Each message $m^\sigma$ delivered from an authenticated link comes with a signature $\sigma$ of the sender. 

\textbf{Overview. \ }
The protocol proceeds in consecutive rounds $r$.
In order to avoid global communication in favor of local communication, replicas are divided into clusters.
Each round has three stages.
In the first stage, clusters process requests in parallel; the replicas of each cluster agree on the transactions and their order, and further the set of reconfigurations for that round.
In the second stage, clusters communicate these operations with each other.
Finally in the third stage, they execute all the operations in the decided orders.
Each cluster has a leader that coordinates the replication of both transactions and reconfigurations.
Each replica knows the $\leader$ of its cluster and its associated timestamp $\ts$.
(Leaders of each cluster are elected together with a monotonically increasing timestamp which replicas use to decide which leader is more recent.)
A leader might continue to serve for multiple rounds.
Several leaders might change until a correct leader
properly replicates transactions and reconfigurations of the round.


\begin{figure*}[t!]
\centering
   \begin{subfigure}{0.4\linewidth}
   \centering
   \includegraphics[width=\linewidth]{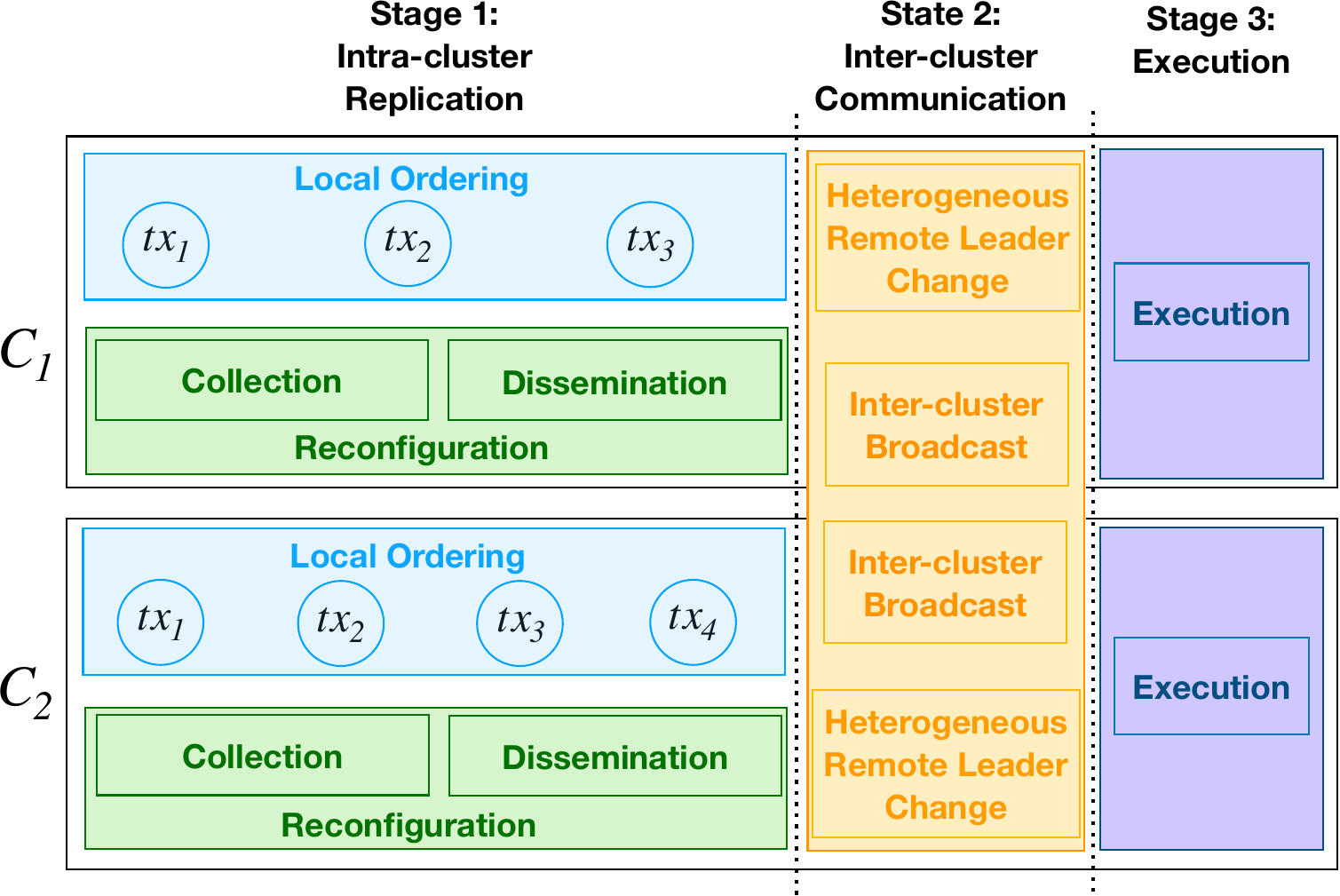}
   \caption{Overview of Stages and Sub-protocols}
   \label{fig:modules-overview}
   \end{subfigure}
   \begin{subfigure}{0.2\linewidth}
   \centering
      \includegraphics[width=\linewidth]{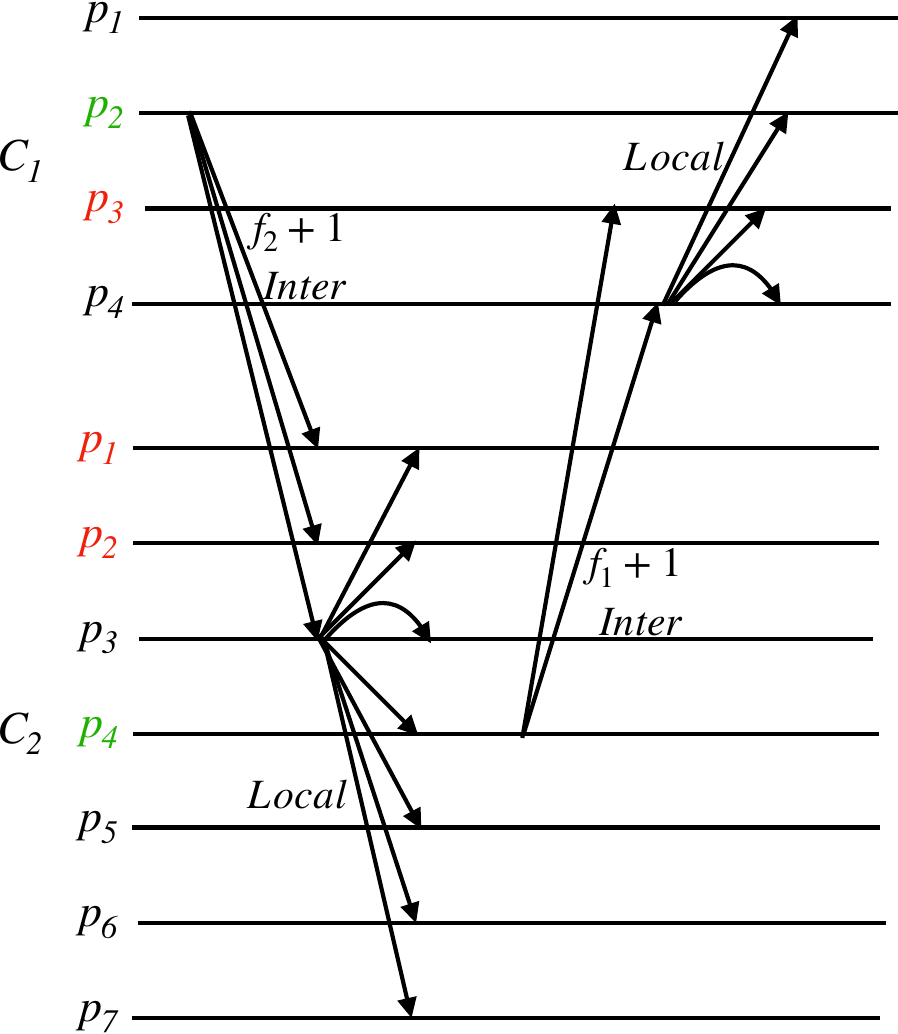}
      \caption{Inter-cluster Broadcast}      
      \label{fig:inter-comm-a}
   \end{subfigure}
   \begin{subfigure}{0.38\linewidth}
   \centering
      \includegraphics[width=\linewidth]{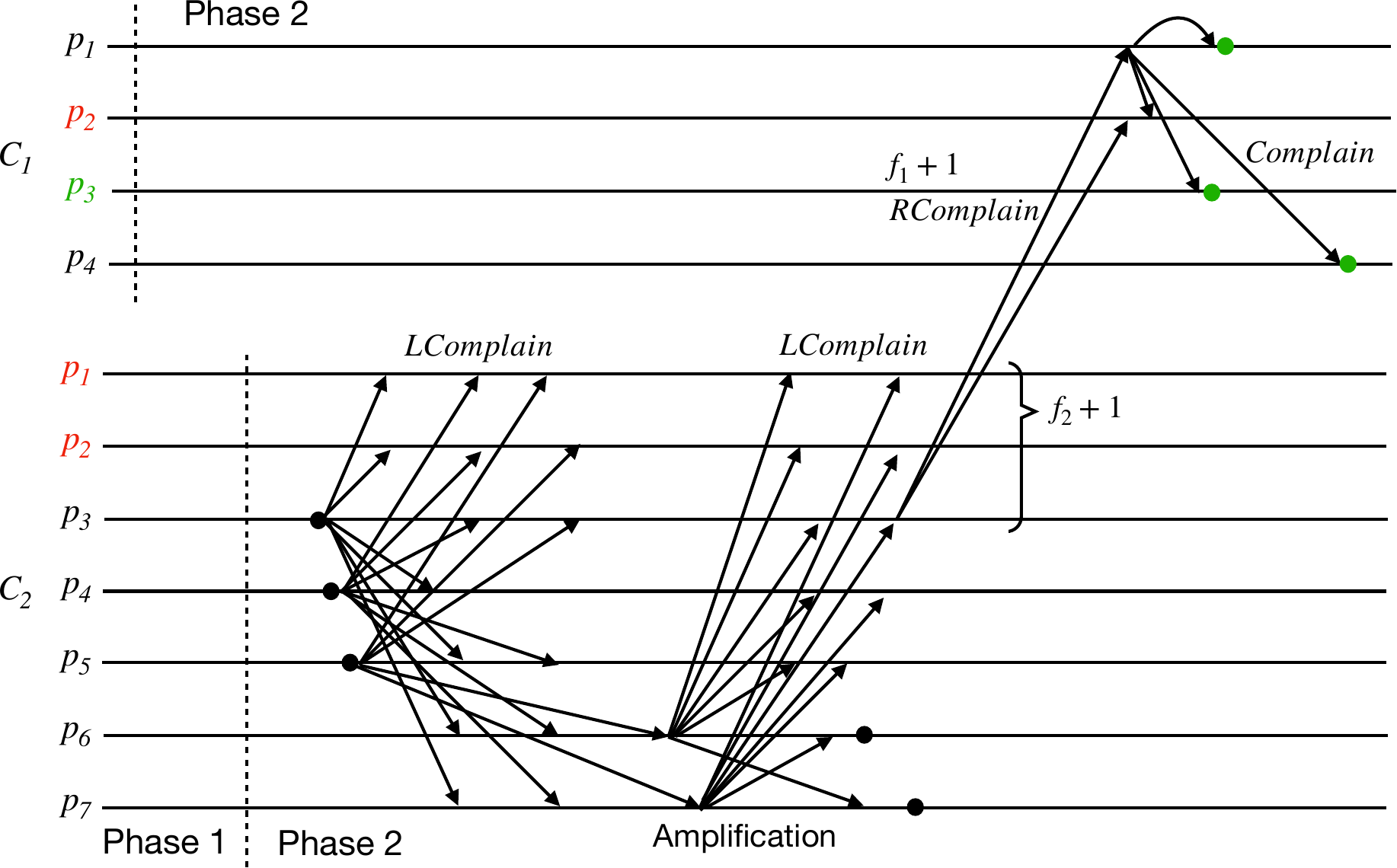}
      \caption{Remote Leader Change}
      \label{fig:inter-comm-b}
   \end{subfigure}
%
\caption{\textsc{Hamava} Reconfigurable Clustered Replication Protocol}
   \vspace{-3mm}
   \label{fig:inter-comm}
\end{figure*}

\textbf{Stages. \ }
The \textsc{Hamava} protocol has three stages. 
\autoref{fig:modules-overview} shows an overview of the stages and sub-protocols in a round.
The replicas are split into clusters.
The figure shows two clusters $C_1$ and $C_2$ that 
make progress
from left to right through the stages.

\textit{State 1: Intra-cluster Replication.}
The first stage is intra-cluster replication
where each cluster 
coordinates replication locally and independently of other clusters.
The first stage has two parts that are executed in parallel:
local ordering, and reconfiguration.
The local ordering protocol orders a batch of transactions uniformly across the replicas of the cluster.
The protocol is 
agnostic to
the local ordering sub-protocol; 
any consensus protocol can be used.
The second part, the reconfiguration protocol, collects and uniformly disseminates the reconfiguration requests across 
the cluster,
even if the leader is Byzantine or changes simultaneously.

\textit{State 2: Inter-cluster Communication.}
After the first stage finishes intra-cluster replication,
the second stage performs inter-cluster communication:
the leader of each cluster broadcasts to other clusters
the transactions and reconfigurations that it has locally replicated.
Each cluster waits to receive these messages from every other cluster.
If a remote leader is Byzantine, it may refrain from sending these messages.
Therefore, 
to ensure progress, 
if the replicas of a cluster don't receive the message from a remote cluster,
they trigger the remote leader change protocol to eventually change the leader of that remote cluster.

\textit{Stage 3: Execution.}
Finally, in the third stage,
each replica 
orders the transactions and reconfigurations that it has received from all clusters by a predefined order for the clusters,
executes them in order,
and issues responses.
This predefined order yields a total-order for operations across replicas.
replicas converge to the same state at the end of the round.

\subsection{Overview of \textsc{Hamava} Inter-cluster Communication}
Let us consider the inter-cluster communication stage (\ie, stage 2).
\autoref{fig:inter-comm} shows example executions of this stage for two heterogeneous clusters $C_1$ and $C_2$ with $4$ and $7$ replicas respectively.
In \autoref{fig:inter-comm-a},
the leaders of $C_1$ and $C_2$ are the green replicas $p_2$ and $p_4$ respectively.
The Byzantine replicas of $C_2$ and $C_2$ are the red replicas $\{ p_{3} \}$ and $\{ p_{1}, p_{2} \}$ respectively.
We note that in both clusters, the number of Byzantine replicas is less than one-third of the size of the cluster: $f < |C| / 3$ that is
$f_{1} = 1 < 4 / 3$ and $f_{2} = 2 < 7 / 3$.

\textit{Inter-cluster Broadcast. \ }
\autoref{fig:inter-comm-a} shows an execution of the inter-cluster broadcast protocol.
Each cluster has already locally replicated operations (including transactions and reconfiguration requests);
each operation is paired with a certificate of consensus which is approval signatures from a quorum of replicas in that cluster.
The leader of each cluster sends its operations together with their certificates to other clusters 
as inter-cluster messages $\Inter$.
To ensure that at least one correct replica in the remote cluster receives the message,
the leader sends the message to $f+1$ replicas in the remote cluster.
In our heterogeneous clusters example,
the leader $p_{2}$ of $C_{1}$ sends the message to $2+1 = 3$ replicas in $C_{2}$,
and the leader $p_{4}$ of $C_{2}$ sends the message to $1+1 = 2$ replicas in $C_{1}$.
In the remote cluster,
the correct replica that receives the $\Inter$ message
then broadcasts the operations as $\Local$ messages to replicas in its own cluster.
Thus, if the leaders are correct,
all correct replicas eventually receive operations from all clusters.

\begin{table}[] 
   \setlength{\tabcolsep}{1.5pt}
   \vspace{-2mm}
   \footnotesize
   \centering
	\begin{tabular}{@{}ccll@{}c}
		\toprule
		\multirow{2}{*}{Protocol} & \multirow{2}{*}{$D$} & \multicolumn{2}{c}{Communication} & \multirow{2}{*}{$DC$} \\
		\cmidrule(lr){3-4} 
		& & Local & Global & \\ \midrule

		Ava-HotStuff
		& z & $O(8zn)$ & $O(fz^2)$ & $\checkmark$ \\ 
		Ava-BftSmart
		& z & $O(2zn^2)$ & $O(fz^2)$ & $\checkmark$\\ 
      GeoBFT                                 & z & $O(4n^2)$ & $O(fz^2)$ & $\checkmark$ \\ 
		Steward                                & 1 & $O(2zn^2)$ & $O(z^2)$ & $\times$ \\ 
		PBFT                                      & 1 & \multicolumn{2}{c}{$O(2(zn)^2)$} & $\times$ \\
		ZYZZYVA                               & 1 & \multicolumn{2}{c}{$O(zn)$} & $\times$ \\ 
		\bottomrule

	\end{tabular}
\caption{
\textcolor{black}{
Best-case complexity of the protocols, where $D$ is decisions, $z$ is the \# of clusters, $n$ is the maximum \# of nodes per cluster, and $f$ is the \# of faulty nodes in a cluster for the given $n$, and $DC$ is Decentralized.}}
\label{tab:complexity}
\vspace{-5mm}
\end{table}

Clustering reduces 
the number of rounds and message complexity for
global communication.
We just considered the inter-cluster broadcast of stage 2 above.
Let us compare the complexity of classical (\ie, not clustered) replication such as PBFT
with clustered replication.
Consider $n_{1} = |C_{1}|$, $n_{2} = |C_{2}|$, and the total number $n = n_{1} + n_{2}$ replicas.
To process a single transaction,
replication
requires $2$ global rounds with message complexity $O((n_{1} + n_{2})^{2})$.
To process $2$ transactions in parallel in $C_{1}$ and $C_{2}$,
clustered replication
executes stage 1 with
$2$ local rounds with message complexity $O(n_{1}^{2} + n^{2}_{2})$,
and then
stage 2 with
$1$ global round with message complexity $(f_{1} + 1) + (f_{2} + 1) = O(n_{1} + n_{2})$,
and finally, $1$ local round with message complexity $(f_{1} + 1)  \times  n_{1} + (f_{2} + 1)  \times  n_{2} = O(n_{1}^{2} + n_{2}^{2})$.
Therefore, global communication is reduced
from $2$ rounds of complexity $O((n_{1} + n_{2})^{2})$
to
$1$ round of complexity $O(n_{1} + n_{2})$.

\textcolor{black}{
We present a comparison of the complexity of \textsc{Hamava} and previous works in 
\autoref{tab:complexity}.
}

\textit{Remote Leader Change. \ }
A Byzantine leader may behave properly in the local cluster, but skip sending
$\Inter$ messages to other clusters.
Let us now consider how 
the replicas of a cluster can instigate
the change of the leader of a remote cluster if they don't receive the expected message from it.
\autoref{fig:inter-comm-b} shows an execution of the remote leader change protocol.
The current leader $p_{2}$ of the cluster $C_{1}$ is Byzantine, 
and 
will be changed to the correct replica $p_{3}$.
In cluster $C_{2}$, 
the replicas $p_{3}$, $p_{4}$ and $p_{5}$ have not received the operations of $C_{1}$,
and their timers expire;
thus, they broadcast a local complaint $\LComplaint$ in $C_{2}$ about $C_{1}$.
The replicas $p_{6}$ and $p_7$ in $C_{2}$ have not already complained,
but receive $f_{2} + 1 = 3$ complaints from the three replicas above.
Since at least $1$ out of $3$ is from a correct replica,
they amplify the complaint by broadcasting an $\LComplaint$ message locally.
A replica accepts the local complain
only when it receive it from $2  \times  f_{2} + 1 = 5$ replicas.
It can be shown that 
this prevents a coalition of Byzantine replicas from forcing a leader change,
and
ensures that all local correct replicas eventually deliver the complaint.
When the first $f_{2} + 1 = 3$ replicas accept the local complaint,
they send a remote complaint $\RComplaint$.
To make sure that the message is sent to the remote cluster,
it is sent by $f_{2}+1$ replicas which contain at least one correct replica.
In the first three replicas,
$p_{3}$ is correct and sends the remote complaint.
The complaint should reach at least one correct replica in $C_{1}$;
thus, $p_{3}$ sends it to $f_{1}+1 = 2$ replicas in $C_{1}$.
The replica $p_{1}$ in $C_{1}$ is correct, and receives the remote complaint.
It accepts the complain if it carries $2  \times  f_{2} + 1$ signatures from $C_{2}$.
It then broadcasts a $\Complaint$ message locally in $C_{1}$.
When the correct replicas receive the local complaint
(at green circles),
they move to the next leader $p_{3}$.
The protocol should deal with complaint replay attacks, and multiple simultaneous change requests, that we will describe in the next section.

\subsection{Overview of \textsc{Hamava} Reconfiguration}
Let's now consider reconfiguration.
Reconfiguration
not only allows replicas to join and leave
but also supports rebalancing the system to maintain 
the proximity of replicas in a cluster, and 
similarity of performance across clusters.

\begin{figure*}
\vspace{-8mm}
\captionsetup[subfigure]{justification=centering}
\centering
\begin{subfigure}{0.17\textwidth}
\centering
   \includegraphics[width=\linewidth]{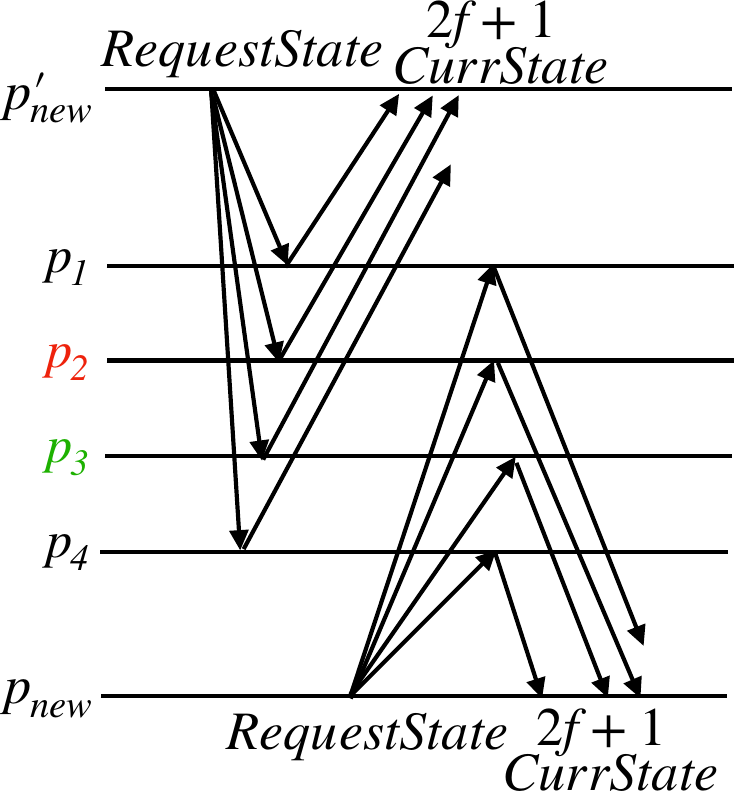}
   \caption{Collection}
   \label{fig:reconfiguration-a}
\end{subfigure}
\hfil
\begin{subfigure}{0.55\textwidth}
\centering
      \includegraphics[width=\linewidth]{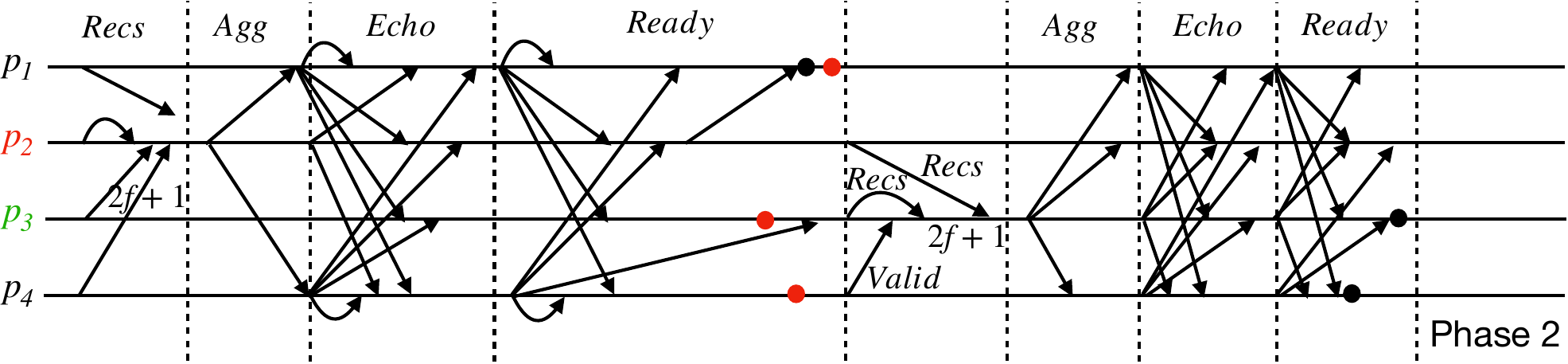}
   \caption{Dissemination}
   \label{fig:reconfiguration-b}
\end{subfigure}
\vspace{-2mm}
\caption{Reconfiguration: Collection and Dissemination}
\vspace{-5mm}
\label{fig:reconfiguration}
\end{figure*}

%



\textit{Attacks. \ }
The reconfiguration requests should be uniformly propagated across clusters, \ie, 
the configurations that every pair of correct replicas
(possibly from different clusters)
execute
in a round should be the same.
When they are not,
the following Byzantine attacks may arise.
Consider two clusters $C_{1}$ and $C_{2}$ with $4$ and $7$ replicas, and 
the failure thresholds $f_1 = 1$ and $f_2  = 2$ respectively.
Assume that $3$ new replicas join $C_1$ and one of them is Byzantine.
The updated $C_1$ now has $7$ replicas, 
and the failure threshold is $f_1' = 2$.
However, assume that the correct replicas in $C_2$ are unaware of the newly joined replicas in $C_1$; 
they keep the stale failure threshold $f_{1} = 1$, and will accept any operations with $2  \times  f_1 + 1 =3$ signatures.
If $C_1$ has a Byzantine leader,
it can forge a certificate for a set of operations $\ops_1$:
it can get a signature from only one correct replica for $\ops_1$.
Then, it can also have signatures from itself and the other Byzantine replica,
to have a total of $3$ signatures.
It can then make the replicas in $C_2$ accept $\ops_1$ with the forged certificate.
However, it can lead the correct replicas in $C_1$ to eventually replicate a different set of operations.
Thus, the correct replicas in $C_1$ and $C_2$ diverge.

Let us now consider another attack in the same setting.
Since the correct leader of $C_2$ has a stale failure threshold $f_1 = 1$,
it sends $f_{1} + 1 = 2$ inter-cluster broadcast messages to $C_{1}$.
The receiver replicas in $C_1$ can be both Byzantine, and may drop the message.
Then, the timers of the correct replicas in $C_1$ will eventually trigger, 
and they complain about the leader of $C_2$.
The remote leader change eventually replaces the correct leader in $C_2$.
Unfortunately, the Byzantine replicas in $C_{1}$ can repeat changing the leader until 
a Byzantine replica is in control in $C_{2}$.

Let's consider the reconfiguration protocol.
Replicas can request join and leave reconfigurations in stage 1 (Intra-cluster replication).
Clusters communicate only in stage 2 (Inter-cluster communication).
Clusters communicate only in stage 2 (Inter-cluster communication).
Thus, if the reconfigurations requested in a cluster in stage 1 are processed as they are requested,
remote clusters will have an inconsistent view of membership for the local cluster.
We explained above that these inconsistencies are unsafe.
Therefore, 
the reconfigurations requested in a round
are locally collected and disseminated in stage 1,
are remotely communicated in stage 2,
and applied in stage 3 to
uniformly update membership for the next round.
Thus, in each round, they can be collected as a set,
and
the order that they are processed in is immaterial.
Therefore, collecting them can be taken off the critical path
that 
orders transactions.
Thus, as \autoref{fig:modules-overview} shows,
reconfigurations are collected and disseminated as a separate workflow in parallel to transaction processing.
\autoref{fig:reconfiguration} shows example executions for both parts of the reconfiguration protocol, collection and dissemination, which we will describe next.

\textit{Collection. \ }
In \autoref{fig:reconfiguration-a},
two replicas $p_{new}$ and $p_{new}'$ request to join the cluster.
Each broadcasts a $\RequestJoin$ message.
When a correct replica delivers a $\RequestJoin$ message,
it adds the join request to its set of reconfiguration requests,
and
responds back by a $\mathit{Ack}$ message.
A joining replica periodically keeps sending $\RequestJoin$ messages
until it receives the $\mathit{Ack}$ message with the same configuration from a quorum of $2  \times  f + 1 = 3$ replicas.
It stops then as it learns that 
Byzantine replicas cannot censor the request.

\textit{Dissemination. \ }
The same set of reconfigurations should be uniformly disseminated to all correct replicas in the cluster.
Otherwise, as we discussed in the introduction,
an inconsistent view of members can 
lead to accepting fake, or discarding genuine messages. 
We describe an execution where the leader is Byzantine and is changed;
nonetheless, the same set of reconfigurations are uniformly delivered to all correct replicas.

In \autoref{fig:reconfiguration-b},
each correct replica sends
the set of reconfiguration requests that it has collected
as $\Recs$ messages
to the leader replica $p_{2}$.
When the leader $p_{2}$ receives messages from a quorum,
it aggregates the received sets of reconfigurations,
and the accompanying signatures,
and then starts disseminating them.
Since there is a correct replica in the intersection of every pair of quorums,
the leader does not miss the requests.
In \autoref{fig:reconfiguration},
the quorum $\{p_1, p_2, p_3\}$ that $p_{new}'$ receives the state from, and the quorum $\{p_2, p_3, p_4\}$ that the leader $p_{2}$ receives requests from intersect in the correct replica $p_{3}$.

The leader broadcasts the aggregation of the reconfiguration requests that it collected.
Upon delivery from the leader,
a correct replica checks whether the received reconfigurations are valid:
they should be accompanied by at least a quorum of signatures for $\Recs$ messages.
As we saw in the collection part,
a requesting replica makes a quorum of replicas store the reconfiguration request.
Therefore,
the leader cannot drop requested reconfigurations:
if the leader drops a request, and hence, any signature from the quorum of replicas that stored it,
then the remaining replicas will be smaller than a quorum,
and
the leader cannot collect a quorum of signatures.
In \autoref{fig:reconfiguration-b}, 
although the leader $p_{2}$ is Byzantine,
it has to send the complete aggregated set.
However, it only sends it to a subset of replicas $\{ p_{1}, p_{4} \}$.
The correct replicas $p_{1}$ and $p_{4}$ that receive a message from the leader echo it.
The Byzantine replica $p_{2}$ echos to them but not $p_{3}$.
Thus, $p_{1}$ and $p_{4}$ receive a quorum of $3$ $\echo$ messages, 
and broadcast a $\ready$ message.
The Byzantine replica $p_{2}$ sends a $\ready$ message to only $p_{1}$.
Thus, only $p_{1}$ receives a quorum of $3$ $\ready$ messages,
and delivers the reconfigurations (at the black circle).

The correct replicas $p_{3}$ and $p_{4}$ don't receive enough $\ready$ messages, 
eventually complain about the leader $p_{2}$, and 
change the leader to the correct replica $p_{3}$ (at the red circles).
To preserve uniformity, 
the new leader $p_{3}$ should retrieve the set of reconfigurations that $p_{1}$ previously delivered.
We will describe later in \autoref{sec:reconfiguration},
how the leader
retrieves that set, and
makes the remaining correct replicas
$p_{3}$ and $p_{4}$ eventually deliver the same set (at black circles).

We note that the classical Byzantine reliable broadcast (BRB)
and Byzantine consensus
would be inadequate for reconfiguration dissemination.
Firstly,
in contrast to BRB that guarantees termination only when the sender is correct,
the reconfigurations are expected to be eventually delivered in each round
even if the initial leader is Byzantine.
Thus, to ensure termination, 
the leader might be changed during dissemination.
The challenge is to 
keep uniformity across leaders.
Further, 
in contrast to BRB 
where a message from one designated sender is broadcast,
and in contrast to consensus 
where a proposal from one replica is decided,
this protocol should 
aggregate and
broadcast a collection of reconfigurations from a quorum of replicas.

In this section, we saw an overview of the stages (\autoref{fig:modules-overview} and the accompanying description).
Next, 
we first consider 
the two more important sub-protocols:
inter-cluster communication (\autoref{sec:inter-cluster-comm}),
and reconfiguration (\autoref{sec:reconfiguration}).
(We detail the
the stages
in the extended report~\cite{reconfig}.

\vspace{-1mm}

\section{Inter-cluster Communication}
\label{sec:inter-cluster-comm}

We will now present 
the inter-cluster broadcast protocol that propagates operations between clusters,
and
the heterogeneous remote leader change protocol 
that detects and changes Byzantine leaders 
for remote clusters.

\textit{State. \ }
Each replica keeps the set of replicas $C_j$ for each cluster.
(We use the index $\i$ only for the current cluster, and the index $j$ for clusters in general.)
The set $C_\i$ keeps track of membership within the current cluster $\i$,
and is used for intra-cluster communication.
The sets $C_j$ that keep track of the members of remote clusters $j$
are used for inter-cluster broadcast.
Accordingly, a replica
has the failure threshold $f_j$ for each cluster $C_j$
as one-third of the size of $C_j$.
Each replica also keeps the current round $r$.
Further, it stores the operations $\operations_{j}$ that it receives from each cluster $C_j$.
Each replica keeps a set of certificates $\certs$ for its local operations $\operations_\i$.
A certificate for an operation 
contains at least $2 \times f_\i + 1$ signatures,
and is sent to other clusters together with the operation.
The protocol uses authenticated perfect links $\apl$,
and authenticated best-effort broadcast $\abeb$ 
(that were described in \autoref{sec:overview}).
Each message $m^\sigma$ delivered from an authenticated link comes with a signature $\sigma$ of the sender. (We elide the signature when it is not needed in a context.)

\textbf{Inter-cluster Broadcast. \ }
At the end of stage 1, the local ordering stage,
the leader calls the function 
$\interBroadcast$
(\autoref{alg:inter-cluster})
to start
the second stage.
Each cluster broadcasts 
its locally ordered operations
to remote clusters.
As \autoref{fig:inter-comm-a} shows,
this function sends out 
the batch of operations $\ops$ of the local cluster together with their certificates $\certs$
as $\Inter$ messages
to other clusters 
(at \autoref{algI:sharing_handler}-\ref{algI:inter_send}).
For each remote cluster $j$, 
the $\Inter$ messages are sent to $f_{j} + 1$ distinct replicas.
Therefore, at least one correct replica at cluster $C_{j}$ eventually receives the $\Inter$ message 
(at \autoref{algI:inter_received}).
It checks that the certificates are valid:
a certificate for an operation from cluster $C_{j'}$ is valid if it contains at least $2 \times f_{j'} + 1$ signatures from the cluster $C_{j'}$.
The receiving replica then broadcasts the operations as $\Local$ messages to other replicas in its own cluster
(at \autoref{algI:local_send}).
Upon receiving a $\Local$ message 
containing operations $\ops$ from a remote cluster $j'$
with valid certificates
(at \autoref{algI:local_handler}), 
the replica 
maps $j'$ to $\ops$ 
in its $\operations$ map.
It also stops a $\timer$ that watches the leader of cluster $j$.
(We will consider remote leader change in the next paragraph).
When operations from all clusters are received,
the replica calls 
the function $\execute$
to enter stage 3, the ordering and execution stage 
(at \autoref{algI:execution_request}).

\begin{algorithm}[t]
	\algsize
	\caption{Inter-cluster Broadcast}
	\label{alg:inter-cluster}
	\DontPrintSemicolon
	\SetKwBlock{When}{when received}{end}
	\SetKwBlock{Upon}{upon}{end}
	\SetKwBlock{Function}{function}{end}
	\SetKwBlock{Foreach}{foreach}{end}

	\Vars : \;
	\ \ \ $\C_{j} : \Set[P]$ \AComment{Replicas of each cluster $C_{j}$}
	\ \ \ $\i$ \AComment{The number of the current cluster}
	\ \ \ $f_{j} : N$ \AComment{Failure threshold for $C_{j}$}
	\ \ \ $r$ \AComment{The current round}   
	\ \ \ $\operations_{j} \leftarrow \emptyset$  \AComment{Operations from each cluster $C_j$}
	\ \ \ $\certs$
	\AComment{Certificates for $\operations_\i$ of $C_\i$}
	
	\Uses $:$ \;
	\ \ \ $\apl: \mathsf{Authenticated Point 2 Point Link}$ \;
	\ \ \ $\abeb : \mathsf{Authenticated Best Effort Broadcast}$ in $C_\i$ \;

	\Function($\interBroadcast ❪ r, \ops, \certs ❫$\label{algI:sharing_handler}){
		\Foreach($C_{j}, j \neq \i$) {
			\Foreach($p \in P$ where $P \subseteq C_{j}  \land |P| = f_{j} + 1$) {
				$\apl \ \request \ \send(p, \Inter(r, \i, \ops, \certs))$ \label{algI:inter_send} \;
			}
		}
	}
	
	\vspace{0.7ex}    
	\Upon($\apl \ \response \ \deliver❪p, \Inter ❪ r', j, \ops, \Sigma ❫❫$ where 
	$r' = r$
	$\land$
	$\Sigma$ is valid ❪\ie, $\Sigma$ has at least $2 \times f_{j} + 1$ signatures from $C_j$ for each $\mathit{op} \in \ops$❫ \label{algI:inter_received}){
		$\abeb \ \request \ \broadcast(\Local(r, j, \ops, \Sigma))$ \label{algI:local_send}\;
	}

	\vspace{0.7ex} 
	\Upon($\abeb \ \response \ \deliver❪p, \Local ❪ r', j, \ops, \Sigma ❫❫$ where 
	$r' = r$
	$\land$ 
	$\Sigma$ $\text{is}$ $\text{valid}$
	\label{algI:local_handler}){
		$\operations_{j} \leftarrow \ops$ \label{algI:update_received}\;
		stop $\timer_{j}$ \label{algI:stop_timer} \;
		\If{$|\mathsf{dom}(\operations)| = N$ \ \ $\triangleright \, N$ \textnormal{is \# of clusters}\label{algI:receive_all_clusters}}{
			$\call \ \execute(\operations)$ \label{algI:execution_request}\;
		}
	}

\end{algorithm}

\textbf{Heterogeneous Remote Leader Change. \ }
Each replica 
waits until it receives operations from other clusters.
Thus, if the leader of a cluster is Byzantine,
and avoids sending operations to other clusters, 
it can stall progress.
Consider a system where cluster $C_j$ has a Byzantine leader $l$. 
For example in \autoref{fig:inter-comm-b}, the leader $p_2$ of $C_1$ is Byzantine.
It acts as a correct leader internally in $C_j$ for the local ordering stage.
The correct replicas in $C_j$ cannot identify $l$ as a Byzantine leader to replace it.
But $l$ does not follow the protocol to send its operations to a remote cluster $C_{j'}$.
Thus, replicas of the cluster $C_{j'}$  cannot proceed to the ordering and execution stage.

\textit{Intuition. \ }
We describe how the local cluster can trigger leader change in a remote cluster.
Each replica keeps a timer $\timer_j$ for 
the leader of each cluster $C_j$.
It resets the timers for all clusters at the beginning of each round.
When a local replica does not receive the operations of a remote cluster, and the timer expires,
it broadcasts a complaint in its local cluster.
When enough local replicas complain,
the complaint is eventually accepted locally.
A subset of local replicas that accept a complaint
send complaints to remote replicas 
which in turn broadcast it in the remote cluster.
Once remote replicas receive the remote complaint,
they change the remote leader.

A remote replica accepts a remote complaint
only if it comes with a quorum of signatures from the complaining cluster.
This prevents any coalition of Byzantine replicas in the complaining cluster
to force a remote leader change.
But a Byzantine replica in the remote cluster can keep a valid complaint and its accompanying signatures, 
and launch a replay attack:
it can resend the valid complaint to repeatedly change the leader.
To prevent this attack, 
the complaining cluster maintains a complaint number $\cn_j$ for each remote cluster $C_j$,
which is incremented on every remote complaint sent to $C_j$.
A remote replica 
maintains the number of complaints received $\rcn_{j'}$ from each other cluster $C_{j'}$,
and only accepts a complaint with the next expected number, 
then increments it.
Thus, a remote replica accepts each remote complaint only once.

\begin{algorithm}[t]
	\algsize
	\caption{Heterogeneous Remote Leader Change}
	\label{alg:het-remote-leader-change}
	\DontPrintSemicolon
	\SetKwBlock{When}{when received}{end}
	\SetKwBlock{Upon}{upon}{end}
	\SetKwBlock{Function}{function}{end}
	\SetKwBlock{Foreach}{foreach}{end}

	\Vars $:$   \;
	\ \ \ $\self$ \AComment{The current replica}   
	\ \ \ $\timer_{j} \leftarrow \Delta$
	\AComment{A timer for each cluster $C_{j}$} 
	\ \ \ $\cn_{j} \leftarrow rcn_{j} \leftarrow 0$
	\AComment{\# of complaints sent to \& received from $C_{j}$}
	\ \ \ $\cs_{j} \leftarrow \emptyset$ \AComment{Complaint signatures for each cluster $C_{j}$}
	\ \ \ $\complained_{j} \leftarrow \false$ \AComment{If complained about each cluster $C_{j}$}

	\Upon($\timer_{j}$ for remote $C_{j}$ expires \label{algV:timer_trigger}){
		$\abeb \ \request \ \broadcast(\LComplaint( j, \cn_{j}, r ))$ \label{algV:Drvc_send} \;
		$\complained_{j} \leftarrow \true$ \;
	}
	
	\vspace{0.7ex} 
	\Upon($\abeb \ \response \ \deliver ❪ p,$ $\LComplaint ❪ j,$ $c,$ $r' ❫^\sigma❫$ where 
	$r' = r$ $\land$ 
	$c = \cn_{j} \land \operations_{j} = \perp$\label{algV:Drvc_deliver}){
		$\cs_{j} \leftarrow \cs_{j} \cup \{ \sigma \}$ \label{algV:add_complain}\;
		
		\If{$|\cs_{j}| \geq f_\i + 1 \land \neg\complained_{j}$\label{algV:f_complain}}{
			$\complained_{j} \leftarrow \true$ \;
			$\abeb \ \request \ \broadcast(\LComplaint(j, c, r))$ \label{algV:amplify_Drvc} \;
		}
		
		\If{$|\cs_{j}| \geq 2 \times f_\i + 1$\label{algV:2f_complain}}{
			$\llet \ S ≔ \mbox{first } f_\i + 1 \mbox{ replicas of } C_\i \ \lin$ \label{algV:let-sender-set}\;
			\If{$\self \in S$\label{algV:in-sender-set}}{
				$\apl \ \request \ \send(p, \RComplaint( \cn_{j},$ $i,$ $\cs_{j},$ $r ))$, 
				for each $p \in S'$ in a set $S'$ such that $S'\subseteq C_{j}$ $\land$
				$|S'| = f_{j} + 1$ \label{algV:Rvc_send} \;
			}
			$\cn_{j} \leftarrow \cn_{j} + 1$ \label{algV:reset1} \;
			$\cs_{j} \leftarrow \emptyset$; \ \ \ 
			$\complained_{j} \leftarrow \false$; \label{algV:reset3} \ \ \ 
			reset $\timer_{j}$
		}
	}
	
	\vspace{0.7ex} 
	\Upon($\apl \ \response \ \deliver ❪p, \RComplaint ❪ c,$ $j',$ $\Sigma, r❫ ❫$
	where
	$r = r'$ $\land$ 
	$c = \rcn_{j'}$ $\land$ 
	$\Sigma$ contains $2 \times f_{j'} + 1$ signatures from $C_{j'}$\label{algV:Rvc_deliver}){
		$\abeb \ \request \ \broadcast(\Complaint(c, j', \Sigma))$ \label{algV:rb_Rvc} \;
	}
	
	\vspace{0.7ex} 
	\Upon($\abeb \ \response \ \deliver ❪p, \Complaint ❪ c,$ $j',$ $\Sigma ❫❫$ 
	where 
	$c = \rcn_{j'}$ $\land$   
	$\Sigma$ contains $2 \times f_{j'} + 1$ signatures from $C_{j'}$\label{algV:Rvc_rb_deliver}){
		$\rcn_{j'} \leftarrow \rcn_{j'} + 1$ \label{algV:rcn-inc} \;
		\If{$\Delta - timer_\i > \epsilon$}{
			$\lem \ \request \ \nextLeader$ \label{algV:local_complain} \;
		}
	}
	
\end{algorithm}

\textit{Protocol. \ }
As \autoref{alg:het-remote-leader-change} presents,
if a local replica finds that the timer $\timer_j$ for a remote cluster $C_j$ is expired
(at \autoref{algV:timer_trigger}),
it broadcasts a local complaint $\LComplaint$ message about $C_j$
to replicas in its own local cluster
(at \autoref{algV:Drvc_send}).
In \autoref{fig:inter-comm-b}, the replicas $\{p_3, p_4, p_5\}$ in $C_2$ send $\LComplaint$ messages.
The message includes the current complaint number $\cn_j$.
Once a local replica receives a local complaint for a remote cluster $C_j$
with the expected complaint number $\cn_j$,
and it has not received operations from that cluster
(at \autoref{algV:Drvc_deliver}),
it records the accompanying signature $\sigma$ in 
the set of complaint signatures $\cs_j$
(at \autoref{algV:add_complain}).
If the replica receives $f_\i + 1$ complaint signatures,
since at least one is from a correct replica,
the replica amplifies the complaint locally
if it has not already complained
(at \autoref{algV:f_complain}-\ref{algV:amplify_Drvc}).
In \autoref{fig:inter-comm-b}, the replicas $\{p_{6}, p_{7}\}$ in $C_{2}$ amplify the $\LComplaint$ message.

Once a replica receives $2 \times f_\i + 1$ complaint signatures
(at \autoref{algV:2f_complain}),
it accepts the local complaint.
Since there is at least one correct replica in the senders,
Byzantine replicas cannot force a leader change.
Further, 
since the complaint is received from $2 \times f_\i + 1$ replicas,
it can be shown that all correct replicas in the local cluster eventually deliver the complaint.
The complaint should reach at least one correct replica in the remote cluster $C_{j}$.
Therefore, 
the remote complaint message $\RComplaint$ should be sent to at least $f_j + 1$ remote replicas.
Further, 
at least one \emph{correct} replica should send these messages.
Therefore, at least $f_\i + 1$ replicas should send it.
The first $f_\i + 1$ replicas of the local cluster (by a predefined order)
send the complaint
(at \autoref{algV:let-sender-set});
we call them the sender set.
In \autoref{fig:inter-comm-b}, the sender set is $\{ p_{1}, p_{2}, p_{3} \}$.
The two replicas $p_{1}$ and $p_{2}$ are Byzantine but $p_{3}$ is correct and sends the message.
If the current replica is in the sender set,
it sends a remote complaint $\RComplaint$ message
to a subset of $C_{j}$ of size $f_{j} + 1$
(at \autoref{algV:in-sender-set}-\ref{algV:Rvc_send}).
The remote complaint message 
includes the complaint number $\cn_{j}$ 
and
the collected signatures $\cs_{j}$.
Finally,
the local replica increments the complaint number,
and resets the state for the next complaint
(at \autoref{algV:reset1}-\ref{algV:reset3}).

Once a replica receives the remote complaint message 
(at \autoref{algV:Rvc_deliver}),
if the message has the next expected complaint number $\rcn_{j'}$,
and it carries $2 \times f_{j'} + 1$ signatures from the complaining cluster $C_{j'}$,
it broadcasts a $\Complaint$ message in its own cluster
(at \autoref{algV:rb_Rvc}).
When a replica receives the complaint message from its local cluster
(at \autoref{algV:Rvc_rb_deliver}), 
it performs similar checks
to accept it.
It then increments the received complaint number $\rcn_{j'}$ for the complaining cluster $C_{j'}$,
and 
unless the leader is recently changed,
it requests the local leader election module $\lem$ to move to the next leader
(at \autoref{algV:rcn-inc}-\ref{algV:local_complain}).
(We will consider the local leader election module $\lem$ in \autoref{sec:protocol-phases-app}.)
If the leader is changed recently 
(\ie, only a small amount of time $\epsilon$ is passed since the $\timer_\i$ is reset to $\Delta$),
the protocol avoids requesting to change the leader again 
so that the new leader is not disrupted.
In particular, this happens when 
multiple remote clusters complain about 
the same leader at almost the same time.

\vspace{-1mm}

\section{Reconfiguration}
\label{sec:reconfiguration}

A replica $p$ can issue a $\join$ or $\leave$ request to join or leave.
Later, it receives a $\joined$ or $\left$ response 
(when the reconfiguration is executed in stage 3).
As we showed in \autoref{fig:modules-overview} and
briefly described in the overview \autoref{sec:overview},
reconfiguration requests are collected, and then disseminated locally in stage 1.
We now consider these two steps.

\begin{algorithm}[t]
   \algsize
    \caption{Collection}
    \label{alg:reconf1} 
   \DontPrintSemicolon
   \SetKwBlock{When}{when received}{end}
   \SetKwBlock{Upon}{upon}{end}
   \SetKwBlock{Function}{function}{end}

   $\request : \join, \leave$ \;
   $\response : \joined$, $\left$ \;

   \Vars $:$ \;
	\ \ \ $\recs \leftarrow \emptyset$  \AComment{Set of reconfigurations}   
   \ \ \ $\clientTimer \leftarrow \Delta$ \;

	\Upon(\request \ $\join$ \label{algJ:join_handler}){
	    $\abeb \ \request \ \broadcast(\RequestJoin(r))$ \label{algJ:request_states}  \; 
   }

    \vspace{0.7ex} 	
    \Upon(\request \ $\leave$ \label{algJ:leave_handler}){
	    $\abeb \ \request \ \broadcast(\RequestLeave(r))$ \label{algJ:request_leave} \;
	}

   \Upon($\clientTimer$ expires \label{algJ:join_timer_expires}){
      \If{\textnormal{requested} $\join$}{         
         $\abeb \ \request \ \broadcast(\RequestJoin(r))$ \label{algJ:repeat_join}\;
      }
      \ElseIf{\textnormal{requested} $\leave$}{
         $\abeb \ \request \ \broadcast(\RequestLeave(r))$ \label{algJ:request_leave-repeat}
      }
      reset $\clientTimer$ to a longer period \;
   }

   \vspace{0.7ex} 
   \Upon($\abeb \ \response \ \deliver ❪p, \RequestJoin^\sigma❪r'❫ ❫$ where
   $r = r'$ \label{algJ:requestS_delivered}){	
      $\recs \leftarrow \recs \cup \{ \join(p)^{\sigma} \}$ \label{algJ:update_tentative} \;
      $\apl \ \request \ \send(p, \mathit{Ack}(C_\i, r))$ \label{algJ:ack} \;
   }

   \vspace{0.7ex} 
    \Upon($\abeb \ \response \ \deliver ❪p, \RequestLeave^\sigma ❪r'❫ ❫$ where
   $r = r'$\label{algJ:requestL_delivered}){
	    $\recs \leftarrow \recs  \cup  \{ \leave(p)^{\sigma} \}$ \label{algJ:update_tentative2}\;
	    $\apl \ \request \ \send(p, \mathit{Ack}(C_\i, r))$ \label{algJ:ack2} \;
	}

   \vspace{0.7ex}  
   \Upon($\apl \ \response \ \deliver ❪\overline{p}, \mathit{Ack}❪C',$ $r' ❫❫$ where $|\{\overline{p}\}|  \geq  2  \times  f_\i + 1$ where
   $r = r'$ \label{algJ:stop-client-timer}){
	    stop $\clientTimer$ \;
   }

   \alglinenoStore

\end{algorithm}
\begin{algorithm}[t]
   \alglinenoRestore
   \algsize
    \caption{Dissemination}
    \label{alg:reconf1-dissem}
   \DontPrintSemicolon
   \SetKwBlock{When}{when received}{end}
   \SetKwBlock{Upon}{upon}{end}
   \SetKwBlock{Function}{function}{end}

   \Uses $:$ \;
   \ \ \ $\brd : \mathsf{Byzantine Reliable Dissemination}$ in $C_\i$
   
    \vspace{0.7ex} 
    \Function($\sendRecs$\label{algJ:send_collect_handler}){
      \ACommentL{Called by each replica before the end of stage 1.}      
      $\brd \ \request \ \broadcast (\Recs(r, \recs))$ \label{algJ:send_collect}
         
   }
   
   \vspace{0.7ex}  
	\Upon($\brd \ \response \ \deliver ❪ \overline{\Recs❪r', \recs❫}, \Sigma ❫^{\Sigma'}$ where $r'$ $=$ $r$ $ \land $ $\Sigma$ \mbox{ and } $\Sigma'$ \mbox{ are valid.}\label{algJ:deliver-collection}){
         append $\Reconfig(  \cup  \, \overline{\recs} )$ to $\operations_\i$ \label{algJ:insert}\;
         add $\Sigma$, $\Sigma'$ to $\certs$ \label{algJ:insert2} \;
   }

   \vspace{0.7ex}     
   \Upon(\mbox{$\brd \ \response \ \complain(p)$} \label{algJ:brd-complain-received}) {   
        $\call \ \complain(p)$ \label{algTOB:brd-call-complain-fun}\;
   }

   \alglinenoStore

\end{algorithm}

\textbf{Collection. \ }
As \autoref{alg:reconf1} presents,
when a client process (or replica) $p$ receives a $\join$ request (at \autoref{algJ:join_handler}),
it broadcasts $\RequestJoin$ messages in the local cluster
(at \autoref{algJ:request_states}).
In \autoref{fig:reconfiguration-a},
two replicas $p_{new}$ and $p_{new}'$ request to join.
Similarly, when a correct replica $p$ receives a $\leave$ request, 
it sends out $\RequestLeave$ messages.
The client uses the $\clientTimer$ to track progress while it waits for a response.
If the timer expires
(at \autoref{algJ:join_timer_expires}), 
it resends the messages, and resets the timer to a larger period.
When a correct replica delivers the $\RequestJoin$ message from $p$
(at \autoref{algJ:requestS_delivered}), 
it adds the reconfiguration request $\join(p)$ 
to its set of collected reconfigurations $\recs$,
and sends back an $\mathit{Ack}$ message
(at \autoref{algJ:update_tentative}-\ref{algJ:ack}).
The steps are similar for the $\RequestLeave$.
When the requesting replica receives $\mathit{Ack}$ messages with 
the same 
cluster members, and round
from a quorum 
(at \autoref{algJ:stop-client-timer}),
it learns that the request cannot be censored
by Byzantine replicas;
therefore,
it stops the timer.
In \autoref{fig:reconfiguration-a}, 
the two joining replicas stop the timer
when they receive $\mathit{Ack}$ from $3$ replicas.

\textbf{Dissemination. \ }
Before completing the first stage,
a correct replica calls $\sendRecs$
(\autoref{alg:reconf1-dissem} at \autoref{algJ:send_collect_handler})
that sends a $\Recs$ message 
containing the set of reconfiguration requests $\recs$ that it has collected
to the Byzantine Reliable Dissemination (BRD) module
(at \autoref{algJ:send_collect}).

BRD collects messages and disseminates them.
It eventually issues a response with
a set of collected reconfigurations $\overline{\recs}$
(at \autoref{algJ:deliver-collection}).
The delivery is accompanied by two certificates.
The certificate $\Sigma$ attests that $\overline{\recs}$
are collected from at least a quorum of replicas.
In the collection part,
a reconfiguration request was stored in at least a quorum of replicas.
If $\Sigma$ is valid,
then BRD has collected reconfigurations from at least a quorum of replicas.
Since 
there is a correct replica in the intersection of two quorums,
a Byzantine leader cannot censor the reconfiguration request.
The certificate $\Sigma'$ attests that a quorum of replicas voted to deliver the set;
therefore, correct replicas will eventually deliver the same set.
If the certificates are valid,
the receiving replica
appends the union of $\overline{\recs}$
to $\operations_\i$,
and the certificates to $\certs$
(at \autoref{algJ:deliver-collection}-\ref{algJ:insert2}).
The BRD module may complain if the leader does not lead delivery in a timely manner
(at \autoref{algJ:brd-complain-received}-\ref{algTOB:brd-call-complain-fun}).
The complaint is forwarded to the local leader election module $\lem$.

\textbf{Byzantine Reliable Dissemination. \ }
\label{sec:brd}
In this section, 
we present the Byzantine Reliable Dissemination (BRD) protocol
that we just used.
We present it as a general reusable module,
that is of independent interest.

\begin{algorithm}[t]
   \algsize
    
   \caption{BRD (1/2)}
   \label{alg:reconf2} 
	\DontPrintSemicolon
	\SetKwBlock{When}{when received}{end}
	\SetKwBlock{Upon}{upon}{end}  
   \SetKwBlock{Function}{function}{end}

   $\request : \broadcast (m), \ \newLeader(p, \ts)$ \;
   $\response : \deliver(\{ \overline{m} \}, \Sigma ), \ \complain(p)$ \;

   \Uses: \ \;
   \ \ \ $\apl: \mathsf{Authenticated Point 2 Point Link}$ \;
   \ \ \ $\abeb : \mathsf{Authenticated Best Effort Broadcast}$ \;

   \Vars: \ \;
   \ \ \ $(\leader, \ts) \leftarrow (p_{0}, 0)$ \;     
   \ \ \ $\mym \leftarrow \bot$ \;
   \ \ \ $\echoed, \readied, \delivered \leftarrow \false$ 
   \AComment{Tracking reliable delivery}  
   \ \ \ $\valid, \highValid
   \leftarrow \bot$ 
   \AComment{Validated set of requests}

   \ \ \ $q, M, \Sigma \leftarrow  \emptyset $ \AComment{Collected senders, messages, and signatures}
   \ \ \ $\timer \leftarrow \Delta$ \; 

   \Upon($\request \ \broadcast ❪m❫$\label{algJ:broadcast_init}){
      $\mym \leftarrow m$ \label{algJ:save_m} \;
      $\apl \ \request \ \send (\leader, \langle m, \ts \rangle )$ \label{algJ:send_m_leader} \;
      reset $\timer$ \label{algLC:set-timer} \;      
	}
   
   \Upon($\apl \ \response \ \deliver ❪p, \langle m, t \rangle ^\sigma❫$ $\mbox{where}$ $\self$ $=$ $\leader$ $ \land $ $t = \ts$\label{algJ:leader_converge-single}){
        $q \leftarrow q  \cup  \{p\} $ \label{algJ:add-sender}\;
        $M \leftarrow M  \cup  \{m\}$ \label{algJ:merge} \;
        $\Sigma \leftarrow \Sigma  \cup  \{ \sigma \}$ \label{algJ:cert_for_proposal} \;        
   }

   \Upon($|q|  \geq  2  \times  f + 1  \land  \highValid = \bot$\label{algJ:leader_converge}){
      $\abeb \ \request \ \broadcast(\Agg(M, \Sigma, \ts))$ \label{algJ:leader_rb}     
   }
   
   \Upon($\abeb \ \response \ \deliver ❪p, \Agg ❪ M, \Sigma, t ❫❫$ where $p=\leader$ $ \land $ $ t=\ts$ 
   $ \land $ $ \neg  \echoed$
   $ \land $ $\Sigma$ attests $M$
   $❪$\ie , $\Sigma$ has either 
   at least $2  \times  f + 1$ signatures for $M$,   
   at least $2  \times  f + 1$ $\echo❪M❫$ messages,
   or $f + 1$ $\ready❪M❫$ messages$❫$
   \label{algJ:proposal_delivered}){  
        $\echoed \leftarrow \true$ \label{algJ:record-echoed}\;
        $\abeb \ \request \ \broadcast( \echo( M, \ts ))$ \label{algJ:send_echo}\;
    }

    \vspace{0.7ex}     
    \Upon($\abeb  \ \response \ \deliver ❪\overline{p}, \echo❪ M, t❫ ^{\overline{\sigma}}❫$ where $|\{\overline{p}\}|  \geq  2 \times f + 1$ $ \land $ $t=\ts$ $ \land $ $ \neg \readied$ \label{algJ:echo_delivered}){
        $\readied \leftarrow \true$ \label{algJ:update_readied} \;
        $\abeb \ \request \ \broadcast(\ready(M, \ts))$ \label{algJ:send_ready}\;        
        $\valid \leftarrow  \langle  M, \overline{\sigma}, \ts  \rangle $ \label{algJ:update_rconfigcandidate1} \;        
    }
    \alglinenoStore
\end{algorithm}


\textit{Module. \ }
BRD accepts 
a $\broadcast (m)$ request from each replica.
It then collects and disseminates messages $m$.
It issues a response $\deliver( M, \Sigma )^{\Sigma'}$
where
$M$ is a set of messages,
and
$\Sigma$ and $\Sigma'$ are two sets of signatures.
The certificate $\Sigma$ attests that $M$ is a set of messages from a quorum of replicas,
and 
the certificate $\Sigma'$ attests that $M$ is the only delivered set, 
and every correct replica will eventually deliver it.
In our reconfiguration protocol, these certificate are sent to other clusters
as a proof of these properties for the dissemination in the current cluster. 
Further,
the component may issue a $\complain(p)$ event to complain about the current leader $p$,
and
accepts a $\newLeader(p, \ts)$ request to set a new leader $p$ with a timestamp $\ts$.
Leaders are elected with monotonically increasing timestamps.
BRD guarantees 
the following properties.
Integrity: 
A correct replica may only deliver messages 
from at least a quorum of replicas.
No duplication: 
Every correct replica delivers at most one set of messages.
Uniformity: 
No two correct replicas deliver different set of messages.
Termination:
If all correct replicas broadcast messages, then
every correct replica eventually delivers a set of messages.
Totality:
If a correct replica delivers a set of messages,
then all correct replicas deliver a set of messages.
Validity: 
If a correct replica delivers a set of messages containing $m$ from a correct sender $p$, 
then $m$ was broadcast by $p$.

\textit{Protocol. \ }
As \autoref{alg:reconf2} presents,
when a replica broadcasts a message
(at \autoref{algJ:broadcast_init}),
it stores it and sends it to the leader
(at \autoref{algJ:save_m}-\ref{algJ:send_m_leader}).
It also resets the timer to watch the leader
(at \autoref{algLC:set-timer}).
The leader
adds messages and the accompanying signatures
that it receives (at \autoref{algJ:leader_converge-single})
to the set of messages $M$ and signatures $\Sigma$
(at \autoref{algJ:add-sender}-\ref{algJ:cert_for_proposal}).
Once it collects
messages from a quorum
(at \autoref{algJ:leader_converge}),
it broadcasts 
an aggregation message $\Agg$ containing $M$ and $\Sigma$
(at \autoref{algJ:leader_rb}).
Massages carry the timestamp $\ts$ of the current $\leader$ as well;
any message with a stale timestamp is ignored.
Upon delivery of the aggregation
(at \autoref{algJ:proposal_delivered}),
a correct replica accepts it if $M$ is attested by accompanying signatures $\Sigma$.
The signatures $\Sigma$ 
attest $M$
if
they include at least a quorum of signatures 
for the messages $M$.
The signatures serve as a proof that
the $\leader$ has genuinely collected
messages from at least a quorum.
Thus,
the leader cannot drop the reconfiguration request of a replica that has reached out to at least a quorum.
For example in \autoref{fig:reconfiguration},
the quorum that $p_{new}'$ stored the request at, and 
the quorum that the leader $p_{2}$ receives requests from 
intersect in the correct replica $p_{3}$.
Even if leader $p_{2}$ is Byzantine, and 
sends the aggregated set to only a subset of replicas $\{p_{1}, p_{4}\}$,
it cannot drop reconfigurations from the aggregated set.

If the accepting replica hasn't sent the $\echo$ message,
it records (in the variable $\echoed$) that it is sending it,
and
broadcasts the $\echo$ message
(at \autoref{algJ:record-echoed}-\ref{algJ:send_echo}).
In \autoref{fig:reconfiguration-b}, the correct replicas $p_{1}$ and $p_{4}$ that 
receive
an attested set of messages from the leader echo it.
Upon delivery of an $\echo$ message from a quorum,
if the receiving replica has not sent $\ready$ messages
(at \autoref{algJ:echo_delivered}),
it records (in the variable $\readied$) that it is sending it,
and
then broadcasts a $\ready$ message
(at \autoref{algJ:update_readied}-\ref{algJ:send_ready}).
In \autoref{fig:reconfiguration-b}, 
replicas $p_{1}$ and $p_{4}$ receive a quorum of $3$ $\echo$ messages, 
and broadcast $\ready$.

\begin{algorithm}[t]
   \algsize
   \alglinenoRestore

   \caption{BRD (2/2)}
   \label{alg:reconf3} 
	\DontPrintSemicolon
	\SetKwBlock{When}{when received}{end}
	\SetKwBlock{Upon}{upon}{end}  
   \SetKwBlock{Function}{function}{end}
    
    \Upon($\abeb \ \response \ \deliver ❪\overline{p}, \ready ❪ M, t ❫^{\overline{\sigma}}❫$ where $|\{ \overline{p} \}|  \geq  f + 1$ $ \land $ $t=\ts$ $ \land $ $ \neg \readied$\label{algJ:f_ready_delivered}){
        $\readied \leftarrow \true$ \label{algJ:ready_set_true} \;
        $\abeb \ \request \ \broadcast(\ready(M, \ts))$ \label{algJ:amplify_ready}\;
        $\valid \leftarrow  \langle  M, \overline{\sigma}, \ts  \rangle $ \label{algJ:update_rconfigcandidate2}\;
    }
    
    \vspace{0.7ex} 
    \Upon($\abeb \ \response \ \deliver❪\overline{p}, \ready ❪M, t❫^{\overline{\sigma}}❫$ where $|\{\overline{p}\}|  \geq  2 \times f + 1$ $ \land $ $t=\ts$ $ \land $ $ \neg \delivered$\label{algJ:2f_ready_delivered}){
            $\delivered \leftarrow \true$\label{algJ:delivered_true} \;
            $\response$ $\deliver(M, \Sigma)^{\overline{\sigma}}$\label{algJ:issue_deliver} \;
            stop $\timer$ \label{algJ:stoptimer}
    }

    \vspace{0.7ex} 
    \Upon($\timer$ expires \label{algLC:timout}) {
      $\response \ \complain(\leader)$ \label{algLC:brc-complain}\;
   }

    \vspace{0.7ex} 
    \Upon($\request \ \newLeader ❪ p, t ❫$\label{algJ:brc-new-leader}){
      $(\leader, \ts) \leftarrow (p, t)$ \label{algLC:recod-new-leader} \;
      $\echoed, \readied \leftarrow \false$ \;      
      $\valid, \highValid \leftarrow \bot$ \;
      $q, M, \Sigma \leftarrow  \emptyset $ \;
      reset $\timer$ \label{algLC:set-timer2} \;
      \If {$\valid  \neq  \bot$}{
         $\apl \ \request \ \send (\leader, \Valid(\valid))$\label{algTOB:newleadercollect_send}
      }
      \Else {\If {$\mym  \neq  \bot$}{$\apl \ \request \ \send (\leader,  \langle \mym, \ts \rangle )$\label{algJ:send_collect-p}} }
   }

   \vspace{0.7ex} 
    \Upon($\apl \ \response \, \deliver❪p, \Valid❪ M, \Sigma, t ❫❫$ where $\self$ $=$ $\leader$ $ \land $ $\Sigma$ attests $M$ $❪$\ie, $\Sigma$ has at least $2  \times  f + 1$ $\echo❪M❫$ messages or $f + 1$ $\ready❪M❫$ messages$❫$\label{algTOB:newleadercollect_received}){
         $\llet \  \langle \_, \_, ht \rangle ≔ \highValid \ \lin$ \label{algTOB:high-valid1} \;
         \lIf{$t > ht$}{$\highValid \leftarrow  \langle M, \Sigma, t \rangle $ \label{algTOB:high-valid2}}
        $q ←q  \cup  \{p\} $ \label{algJ:add-sender2}\;         
    }

   \Upon($|q|  \geq  2  \times  f + 1  \land  \highValid  \neq  \bot$\label{algJ:leader_converge2}){
      $\llet \  \langle M, \Sigma, \_ \rangle ≔ \highValid \ \lin$ \;
      $\abeb \ \request \ \broadcast(\Agg(M, \Sigma, \ts))$ \label{algTOB:adopt_candidate}\;
   }    
    \alglinenoStore
\end{algorithm}

If the leader changes during the broadcast,
some correct replica might have delivered the aggregated messages
while others may have not.
Thus, 
to preserve the uniformity of delivered messages across replicas,
the new leader should retrieve the previously delivered messages,
and 
rebroadcast them.
Thus, when a replica accepts a sufficiently echoed set,
it stores it together with
its accompanying signatures,
as $\valid$
(at \autoref{algJ:update_rconfigcandidate1}),
and later forwards it to a new leader.
In \autoref{fig:reconfiguration-b}, 
$p_{1}$ and $p_{4}$ record a $\valid$ set at the end of the $\echo$ step.

When a replica receives at least $f + 1$ $\ready$ messages
(at \autoref{algJ:f_ready_delivered}),
at least one of them is correct and has received at least a quorum of $\echo$ messages.
Therefore, the replica trusts the $\ready$ message and amplified it:
it records that it is sending it,
and
broadcasts a $\ready$ message
(at \autoref{algJ:ready_set_true}-\ref{algJ:amplify_ready}).
It also records the received messages $M$ and 
signatures of the received $\ready$ messages
as $\valid$
(at \autoref{algJ:update_rconfigcandidate2}),
and later forwards it to a new leader.

Finally, when a replica receives a quorum of $\ready$ messages,
and it has not delivered the aggregated messages yet
(at \autoref{algJ:2f_ready_delivered}), 
it records (in the variable $\delivered$) that it is delivering,
delivers the aggregated messages $M$,
and stops the timer
(at \autoref{algJ:delivered_true}-\ref{algJ:stoptimer}).
If a replica does not deliver the aggregated messages before the timer times out,
it complains about the current leader
(at \autoref{algLC:timout}-\ref{algLC:brc-complain}).
In \autoref{fig:reconfiguration-b}, the correct replica $p_{1}$
receives a quorum of $3$ $\ready$ messages, and delivers the reconfigurations (at the black circle).
However, the other correct replicas don't receive enough $\ready$ messages, complain about the leader, and 
eventually change the leader to the correct replica $p_{3}$ (at the red circles).

To preserve uniformity, 
the new leader should retrieve the set of reconfigurations 
that have been previously delivered.
When a replica is informed of a new leader
(at \autoref{algJ:brc-new-leader}),
it records the new leader and timestamp,
resets the state and 
the timer
(at \autoref{algLC:recod-new-leader}-\ref{algLC:set-timer2}),
and then sends a message to the new leader to inform him about the current state of dissemination.
If a $\valid$ set of messages is recorded during the execution with the previous leaders, 
the replica sends it to the new leader
(at \autoref{algTOB:newleadercollect_send}).
Otherwise, 
it sends the message that it originally broadcast (\autoref{algJ:broadcast_init}-\ref{algJ:send_m_leader}) to the current leader
(at \autoref{algJ:send_collect-p}).
In \autoref{fig:reconfiguration-b}, 
the two replicas $p_{2}$ and $p_{3}$ 
send to the new leader
the set of reconfigurations that they had collected and sent to the previous leader.
However, 
$p_{4}$ 
has a $\valid$ set of reconfigurations and
sends them
to the leader.

Let $l$ be the latest leader with the timestamp $\ts$ that has guided the system to delivery of a set $M$ at a correct replica.
Consider the next leader $l'$ with the timestamp $\ts'$.
To preserve uniformity, 
$l'$ should adopt $M$.
In order to find $M$,
$l'$ waits to receive messages from a quorum of replicas,
and then picks the $\valid$ set with the largest timestamp.
Let us explain why.
The set $M$ was delivered only after a quorum of $\ready$ messages was received.
At least $f+1$ of the senders are correct.
A correct replica sends a $\ready$ message 
only after receiving $2  \times  f + 1$ $\echo$ messages, or $f+1$ $\ready$ messages.
In both of those cases,
the receiving replica stores $M$ with $ts$ as $valid$.
Thus, 
at least $f+1$ correct
replicas $P$ have stored $M$ with $ts$ as $valid$.
Therefore,
if $l'$ 
receives messages from a quorum ($2  \times  f + 1$) of replicas,
and retrieves any $valid$ sets,
then 
$M$ with the largest timestamp $ts$
is retrieved from at least one replica in $P$.
The leader $l'$ adopts and broadcasts $M$.
Even if it does not lead to any new delivery of $M$,
any $\valid$ set that is stored under his leadership
will have the same set $M$ with now the larger timestamp $\ts'$.

When the leader receives a $\valid$ set
(at \autoref{algTOB:newleadercollect_received}),
it checks that the accompanying signatures attest its validity:
there are at least
$2  \times  f + 1$ signatures of $\echo$ messages, or
$f+1$ signatures of $\ready$ messages.
The leader keeps the $\valid$ set with the highest timestamp
as $\highValid$
(at \autoref{algTOB:high-valid1}-\ref{algTOB:high-valid2}).
Finally, when the leader has collected messages from a quorum,
if it has received a $\valid$ set
(at \autoref{algJ:leader_converge2}),
it broadcasts $\highValid$
(at \autoref{algTOB:adopt_candidate}).
Otherwise, similar to the first leader (at \autoref{algJ:leader_converge}),
it broadcasts the aggregated messages.
In \autoref{fig:reconfiguration-b}, 
the new correct leader $p_{3}$ waits for $3$ messages,
adopts the $\valid$ set that $p_{4}$ sends,
goes through the $\echo$ and $\ready$ steps,
and makes the remaining correct replicas
$p_{3}$ and $p_{4}$ deliver the same set 
(at black circles).


\textbf{Correctness. \ }
We now state the correctness properties of the end-to-end protocol. All the proofs including the properties of sub-protocols are available in the extended report~\cite{reconfig} \autoref{sec:correctness-app}.

\label{sec:correctness}

%
\begin{theorem}[Validity]
\label{lem:general-validity-main}
\label{lem:general-validity}
Every operation that a correct process requests
is eventually executed by a correct process.

\end{theorem}
%
%
\begin{theorem}[Agreement]
\label{lem:general-termination-main}
\label{lem:general-termination}
If a correct process executes an operation in a round
then every correct process executes that operation in the same round.

\end{theorem}
%
\begin{theorem}[Total-order]
\label{lem:general-total-order-main}
\label{lem:general-total-order}
For every pair of operations $o$ and $o'$,
if a correct process executes only $o$, or executes $o$ before $o'$,
then every correct process executes $o'$ only after $o$.

\end{theorem}

\vspace{-1mm}

\section{Experimental Results}

\textbf{Implementation. \ }
The clustered replication protocol is parametric with respect to the local replication protocol.
We instantiated it for
both HotStuff \cite{yin2019hotstuff} 
and BFTSmart \cite{bessani2014state}
as the local replication protocol to implement 
replicated systems that 
we call \textsc{Hamava} (and \textsc{Ava} for short).
We refer to the two as 
\textsc{Ava-Hotstuff} (A.H) and \textsc{Ava-Bftsmart} (A.B).
We will release all the code and workloads as open source software.

\textbf{Questions.}
We perform experiments to answer the following questions:
   (E0-E2):
   How does clustered replication impact performance?
    (E3) What is the impact of introducing heterogeneity in the clusters on the performance of clustered replication? 
    (E4) What is the impact of failures on performance?
    We are especially interested in leader failures.
   (E5) What is the impact of reconfiguration requests on performance?

\textbf{Platform. \ }
We used Google cloud compute to deploy instances acting as servers and clients in our system. 
Each instance runs Ubuntu Server 22.04 LTS, and has a 2 core processor with 16GB of main memory.
We deploy our framework globally on nodes across 3 Google compute regions, namely US (us-west1-b), Asia (asia-south1-c) and Europe (europe-west3-b). The inter-region network latency is
presented in table \ref{tab:LatencyRegions}. 
For both \textsc{Ava-Hotstuff} and \textsc{Ava-Bftsmart},
we choose the YCSB benchmark with a $85\%$ read and $15\%$ write ratio. 
We deployed one client per cluster with multiple threads that issued its requests 
with the Zipfian distribution one after the other without any delays.
We batched transactions (to batches of size 100) in each round.
We issued operations of size 1KB. 
All experiments were run for 3 minutes and the results were taken from the last minute.

\begin{table}[t!]
\vspace{-2mm}
\centering
\begin{tabular}{|llll|} 
\hline
\multicolumn{1}{|l|}{ms} & \multicolumn{1}{l|}{US} & \multicolumn{1}{l|}{EU} & \multicolumn{1}{l|}{Asia }  \\ \hline
US                                 & 0                       & 148                   & 214     \\
EU                                 & 148                   & 0                       & 134     \\
Asia                               & 214                   & 134                  & 0          \\            \hline                      
\end{tabular}
\caption{Inter-region round-trip latency for three regions: US (us-west1-b), EU (europe-west3-c), Asia (asia-south1-c).}
\vspace{-5mm}
\label{tab:LatencyRegions}
\end{table} 

\begin{figure*}
\vspace{-8mm}
\centering
\includegraphics[width=0.255\linewidth]{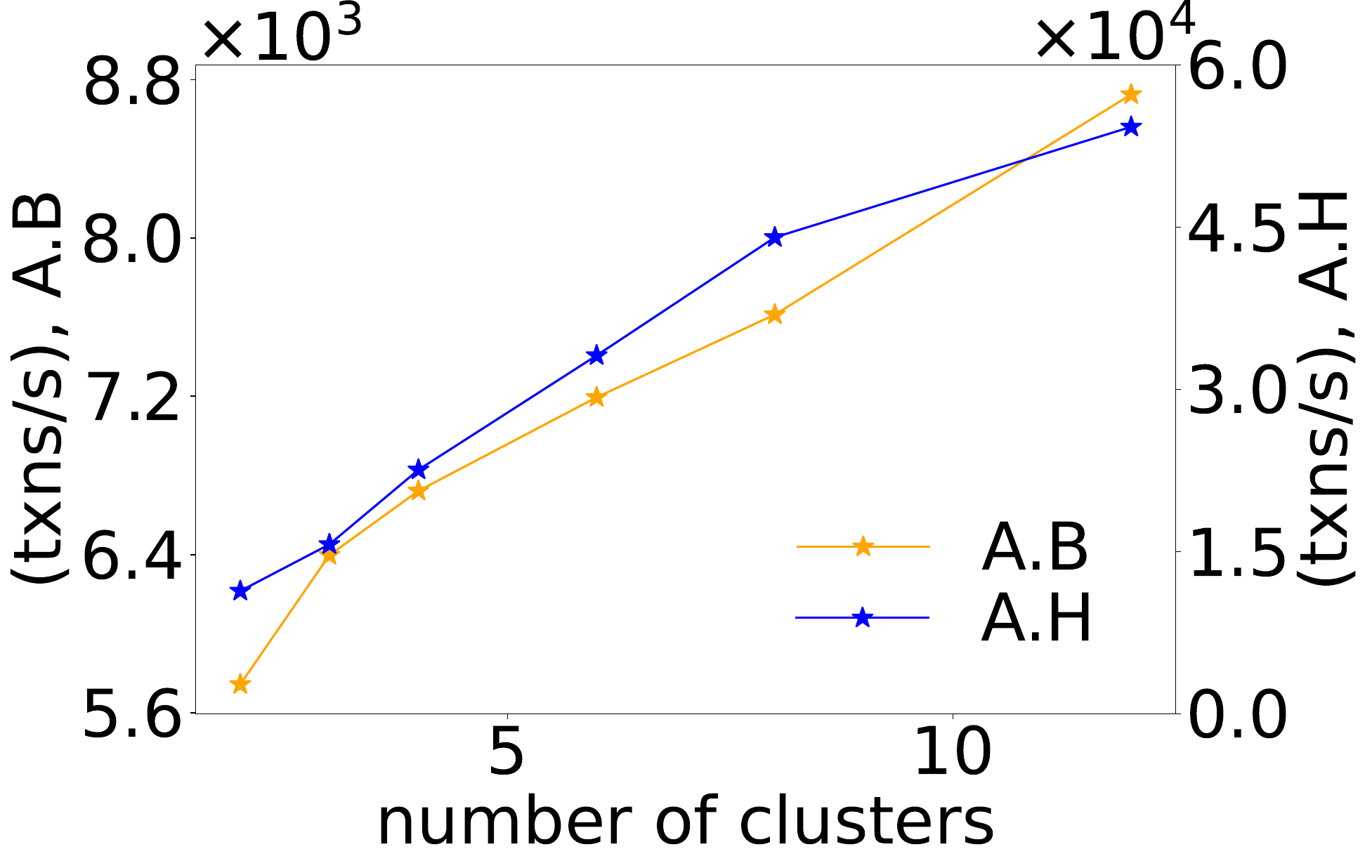} 
\includegraphics[width=0.225\linewidth]{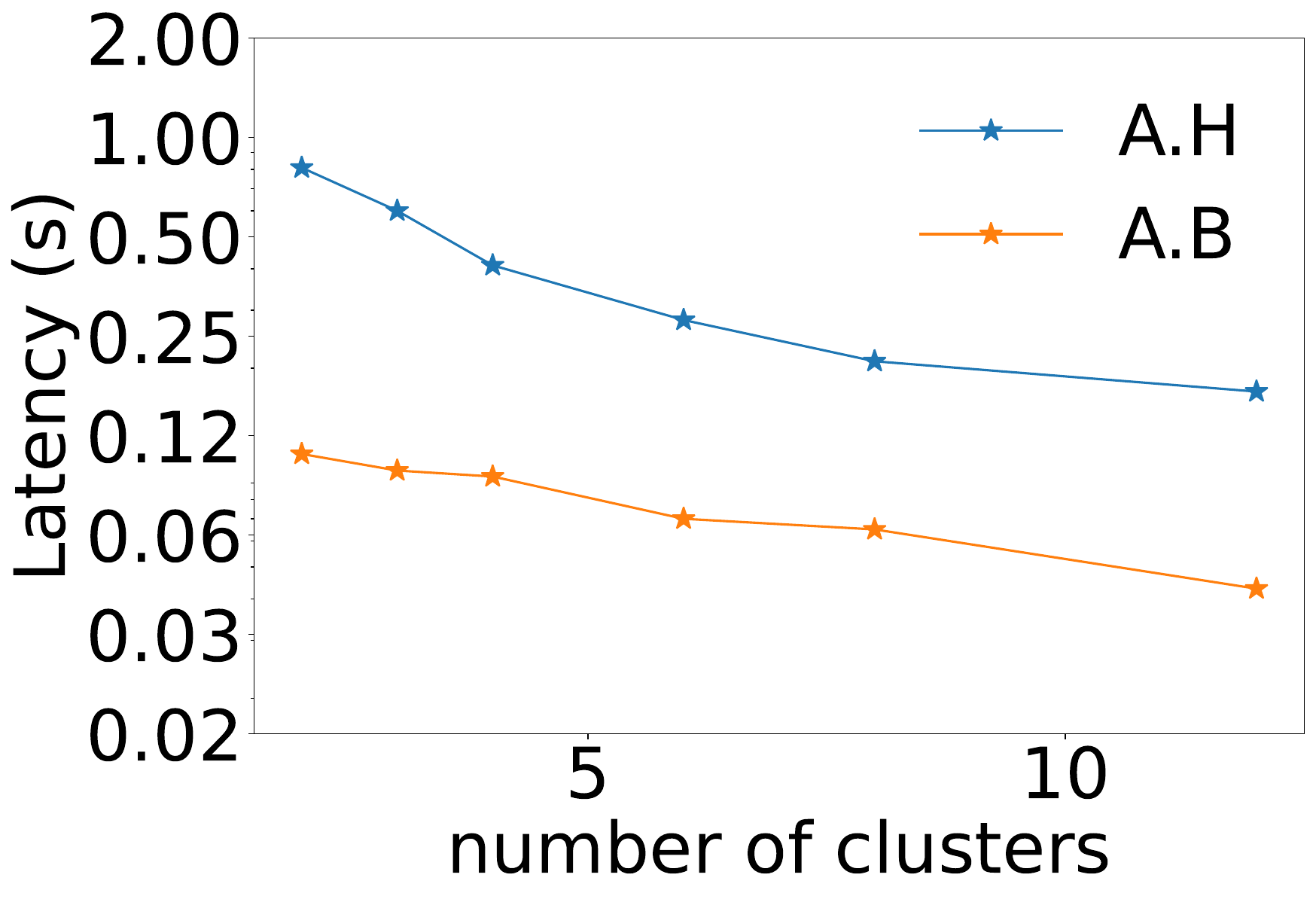} 
\includegraphics[width=0.255\linewidth]{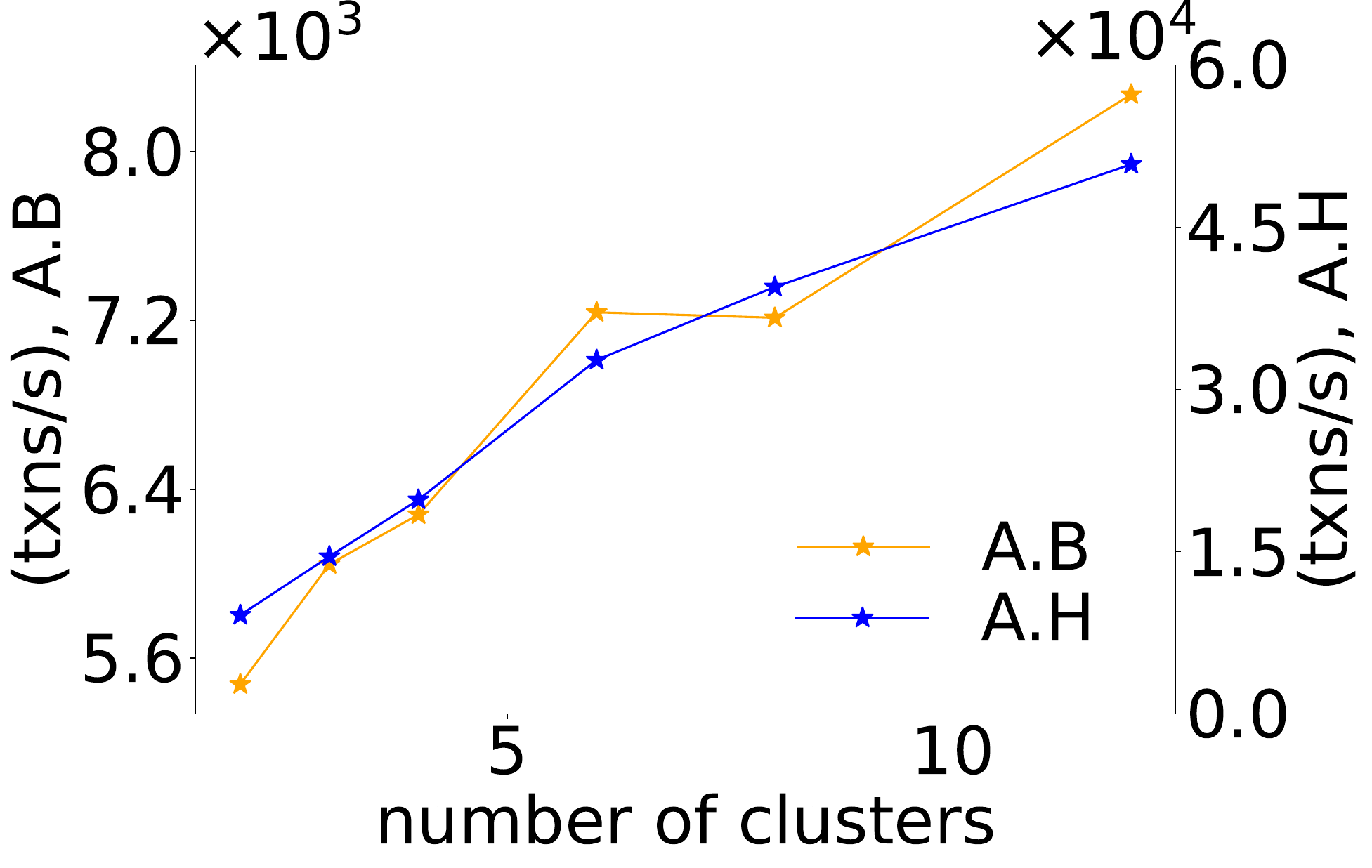}
\includegraphics[width=0.225\linewidth]{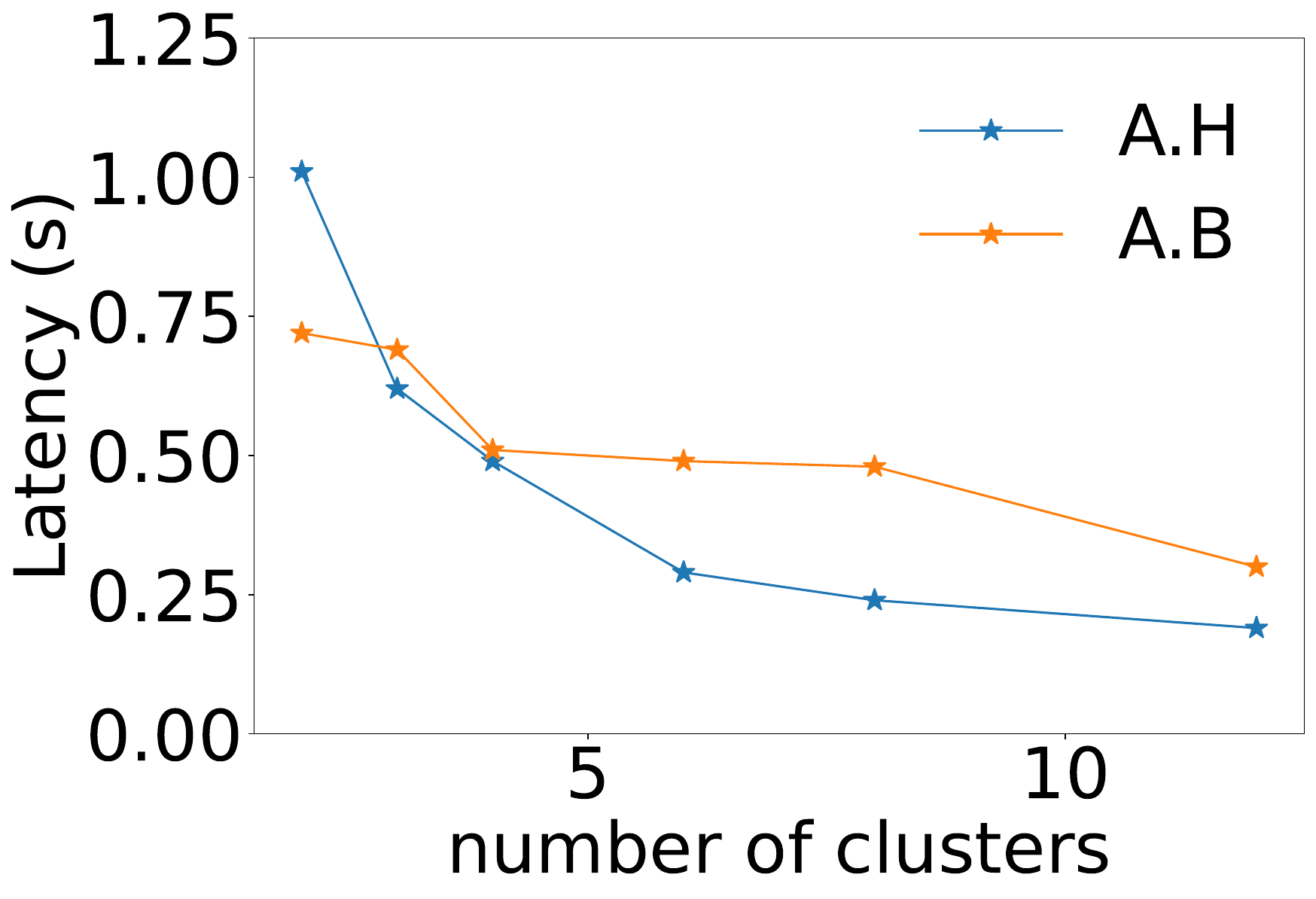}
\caption{Throughput and latency as a function of number of clusters (with 96 nodes) E0. in the same region (two left plots) and E1. across regions (two right plots). 
(In the throughput plot, the left y-axis is for \textsc{Ava-BftSmart} and the right y-axis is for \textsc{Ava-Hotstuff}).
\vspace{-7mm}
}
\label{fig:Throughput}
\label{fig:ThroughputClusterGeo}
\end{figure*}

\textbf{E0. Multi-cluster Single-region. \ }
We investigate the impact of multi-cluster deployment in one Google region on throughput and latency. 
We keep the total number of nodes constant (96),
and divide them to different number of clusters (2, 3, 4, 6, 8, 10, 12).
\autoref{fig:Throughput} reports the effect of the number of clusters on throughput and latency. (In the throughput plot, the left y-axis is for \textsc{Ava-BftSmart} and the right y-axis is for \textsc{Ava-Hotstuff}).
\textit{Assessment. \ }
We observe that
as the number of clusters increases, the throughput of both \textsc{Ava-Hotstuff} and \textsc{Ava-BftSmart} increases.
\textsc{Ava-Hotstuff} exhibits higher throughput than \textsc{Ava-BftSmart}.
We observe that as the number of clusters increase, the latency decreases for both \textsc{Ava-Hotstuff}  and \textsc{Ava-BftSmart}.
\textsc{Ava-BftSmart} exhibits lower latency than \textsc{Ava-Hotstuff}.
As the number of clusters increase,
each cluster has fewer nodes, and local replication is more efficient,
and further, clusters execute the divided workload in parallel.
Thus, the throughput and latency of local replication is improved
that in turn improves the end-to-end throughput and latency.
The two \textsc{Hamava} implementations 
outperform non-clustered replication 
for both throughput and latency.

\textbf{E1. Multi-cluster Multi-region. \ }
We study the impact of deploying clusters in multiple Google regions on the throughput and latency
in \autoref{fig:ThroughputClusterGeo}.
We equally split 96 nodes into different number of clusters (2, 3, 4, 6, 8, 12),
and host them on 3 regions.
A cluster is completely hosted on a single region.
For example, for the 4 clusters setup,
we divide the 96 nodes into 4 clusters of 24 nodes where the first region hosts two clusters, and the second and third regions host one cluster each.
\textit{Assessment. \ }
Similar to the previous experiment,
as the number of clusters increase,
the throughput increases and the latency decreases for both systems.
Similarly, since 
the number of nodes per cluster decreases, and
the workload is divided between clusters,
the throughput and latency are improved.
However, overall performance is lower than the previous experiment
since inter-cluster communication across regions is slower than within one region.
We observe that with multiple regions,
\textsc{Ava-Hotstuff} and \textsc{Ava-Bftsmart} clustered replication still outperform non-clustered replication 
for both throughput and latency.

\textbf{E2. Latency Breakdown. \ }
In \autoref{fig:LatencyBreakdown},
we show the latency breakdown for processing  transactions.
We report the average latency for read and write transactions.
Read transaction have lower latency than writer transactions
since the former can be immediately processed but the latter go through three stages.
We experiment with
3 clusters each containing 4 nodes
in three setups where
clusters span one (Asia), two (EU and Asia), and three (EU, Asia, US) regions.
With one region, the bottleneck is the local ordering as it involves 4 rounds of messages.
On the other hand, the inter-cluster broadcast that involves one round of messages is relatively a smaller part.
With two regions, the latency is dominated by the inter-cluster broadcast when messages have to travel across regions.
With three regions, the inter-cluster broadcast is still the dominating part, and is further increased. 
As shown in \autoref{tab:LatencyRegions},
the round-trip time for EU and Asia is about 134, but when US is added, it is about 214.
Thus, it is crucial to minimize cross-region messaging as our clustered protocol does.

\textbf{E3. Heterogeneity in Clusters. \ }
We investigate the impact of heterogeneity on throughput and latency
for \textsc{Ava-Hotstuff}
in \autoref{fig:optClusters} and \ref{fig:optClusters2},
and 
for \textsc{Ava-BftSmart}
in \ref{fig:optClusters_bft} and \ref{fig:optClusters_bft2}.
Consider 9 nodes in Asia (ap-south-1) and 5 nodes in EU (eu-central-1) regions.
We consider a scale factor $s$ of these numbers varying from $1$ to $5$.
For example, with scale factor $2$, we have $2 \times 9 = 18$ nodes in Asia, and $2 \times 5 = 10$ nodes in EU.
For each scale, we consider 3 setups:
(1) Equal sized clusters. $C_1$: $7$ in Asia. $C_2$: $2$ in Asia and $5$ in EU.
(2) Partition based on region. $C_1$: $9$ in Asia. $C_2$: $5$ in EU.
(3) Partition based on region, and within region. $C_1$: $5$ in Asia. $C_2$: $4$ in Asia. $C_3$: $5$ in EU.
In contrast to previous works, \textsc{Hamava} supports the heterogeneous setups 2 and 3.
\textit{Assessment. \ }
For both \textsc{Ava-Hotstuff} and \textsc{Ava-BftSmart},
setup 2 
exhibits higher throughput and lower latency than setup 1
especially at higher scales.
The setup 2 exploits heterogeneity to
host all members of each cluster in the same region.
Therefore, it deceases the the cost of local replication.
Similarly, 
setup 3
exhibits higher throughput and lower latency than setup 2 at higher scales.
The setup 3 splits a cluster into two smaller clusters in the same region.
Therefore, it further deceases the the cost of local replication.
Further, 
the general trend is that 
throughput and latency 
are better
at lower scale factors,
since 
they have lower cost of local replication.

\textbf{E4. Failures. \ }
We investigate the impact of failures by measuring the performance for 2 clusters with 10 nodes per cluster ($f_1 = f_2 =3$) under three failure scenarios:

\textit{(1) Up to $f$ non-leader failures.}
In \autoref{fig:MultipleNonLeaderFailure},
we test the resiliency of both \textsc{Hamava} systems 
by failing up to $f$ non-leader nodes in each cluster.
The vertical lines show the failure time.
The system 
tolerates the failures, and remains functional.
After the recovery,
the throughput can slightly increase
since local replication is more efficient with fewer number of nodes.

\textit{(2) Leader Failure.}
In \autoref{fig:SingleLeaderFailure},
we fail the leader of a cluster.
After a short window, 
the leader is properly changed, and
the throughput is recovered to the same level.
The timeout for leader change can be adjusted according to the local network latency.
This experiment set it to 20 second; 
thus, the window to complete the leader change
is slightly more than 20 seconds.

\textit{(3) Byzantine Leader and Remote Leader Change.}
We inject Byzantine behavior into leaders to trigger remote leader change.
We make the leader replica complete the first stage within its cluster as a correct leader, but avoid sending inter-cluster broadcast messages. 
As we can see in \autoref{fig:SingleLeaderFailureRVC},
after a short period,
the leader is properly changed, and the throughput comes back up.
The 20 seconds period is the adjustable timeout for multi-cluster message-passing.

\captionsetup[sub]{font=footnotesize,labelfont={bf,sf}}
\begin{figure*}
\vspace{-7mm}

\captionsetup[subfigure]{justification=centering}
\begin{subfigure}{.27\textwidth}
\includegraphics[width=\linewidth]{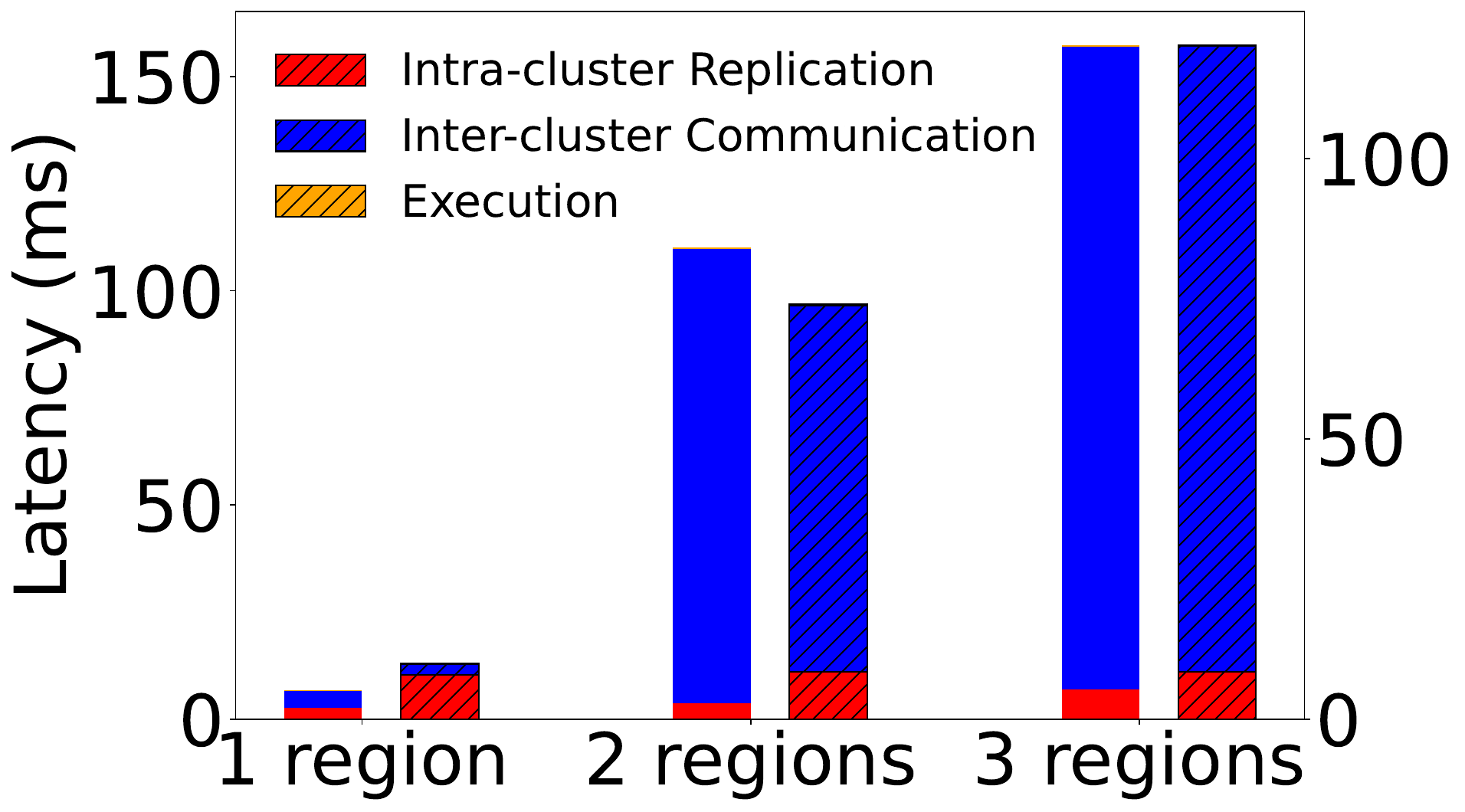} 
\caption{E2}
\label{fig:LatencyBreakdown}
\end{subfigure}
\begin{subfigure}{.225\textwidth}
\includegraphics[width=\linewidth]{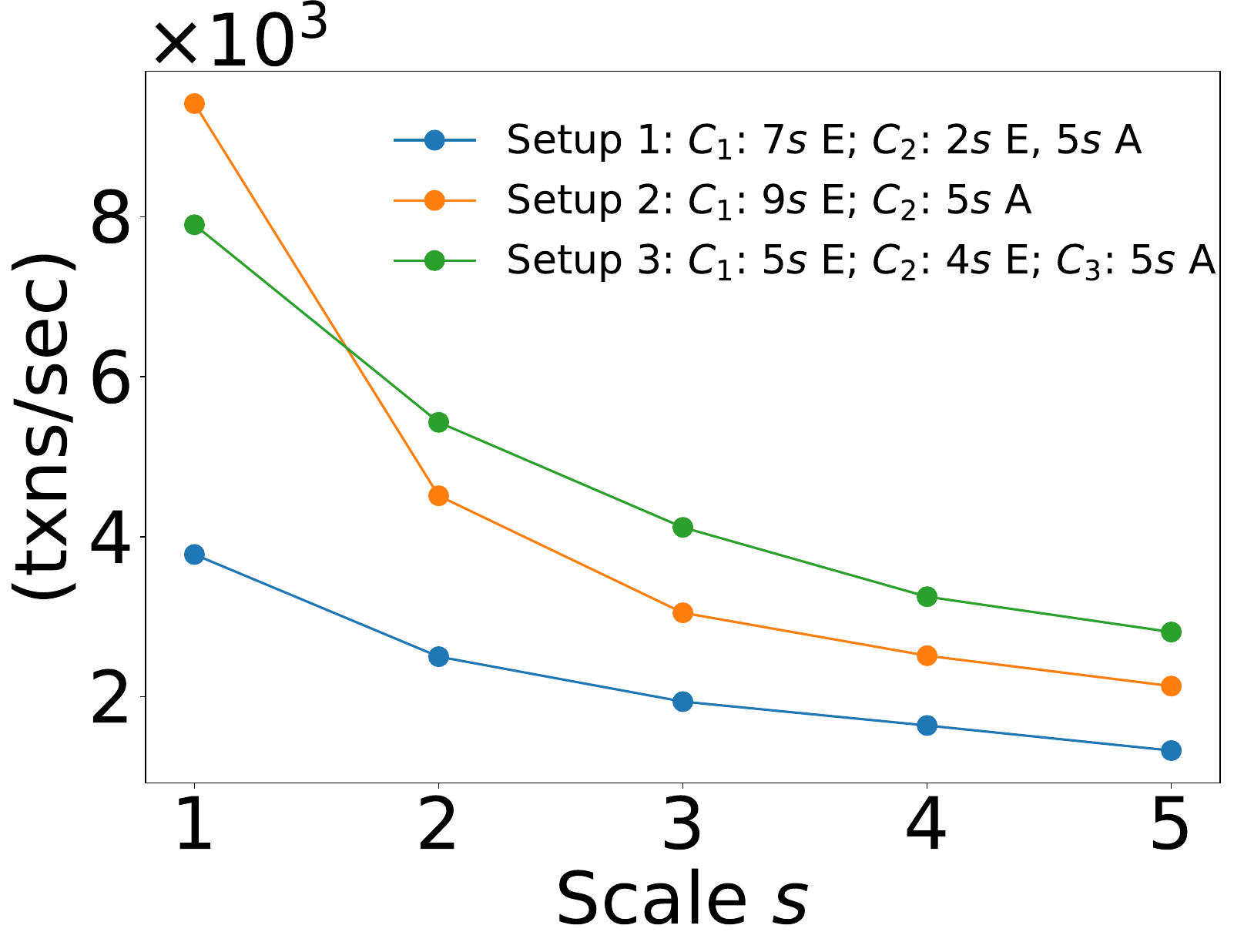} 
\caption{E3. \textsc{Ava-Hotstuff}}
\label{fig:optClusters}
\end{subfigure}
\begin{subfigure}{.245\textwidth}
\includegraphics[width=\linewidth]{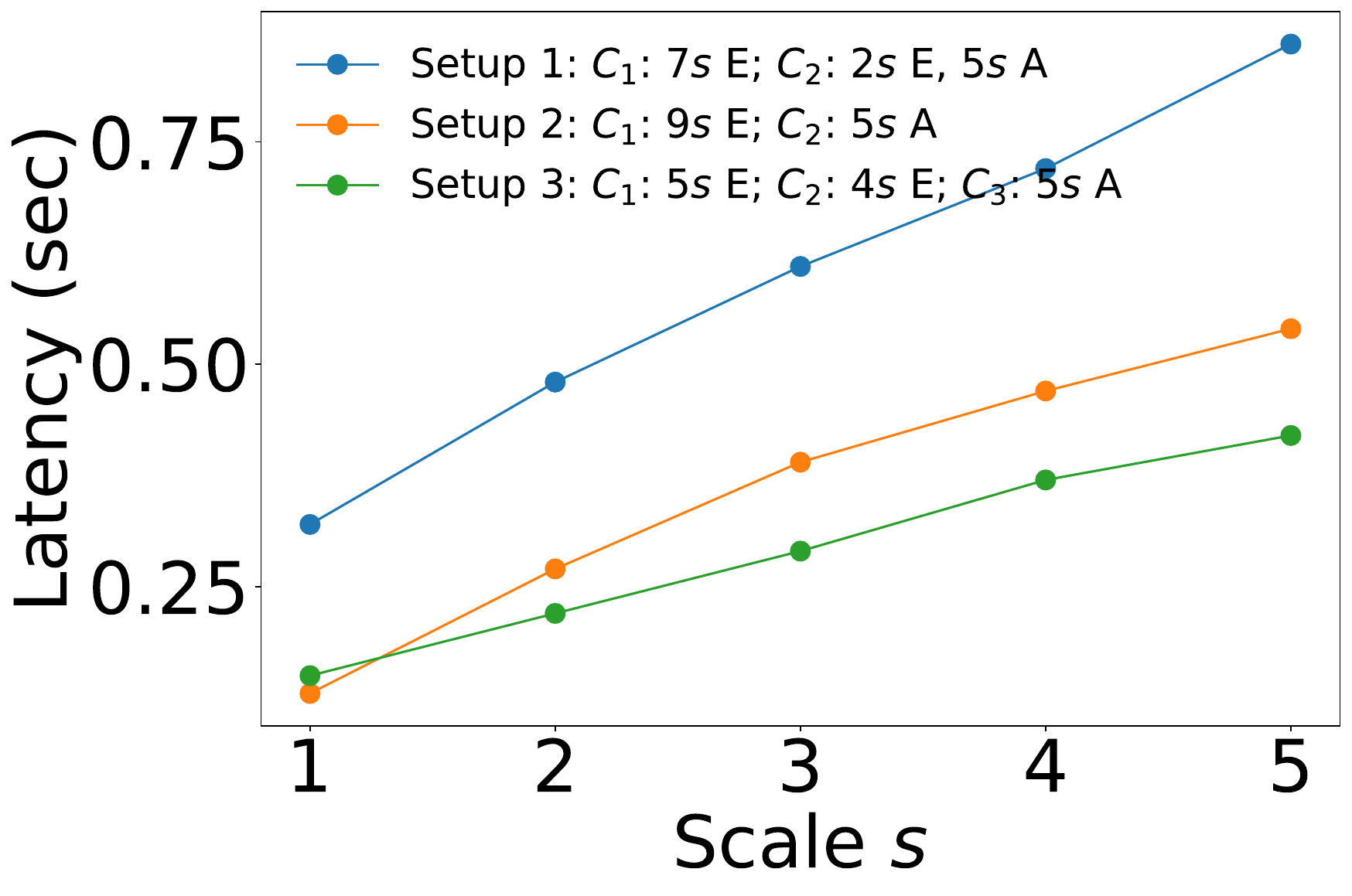} 
\caption{E3. \textsc{Ava-Hotstuff}}
\label{fig:optClusters2}
\end{subfigure}
\begin{subfigure}{.225\textwidth}
\includegraphics[width=\linewidth]{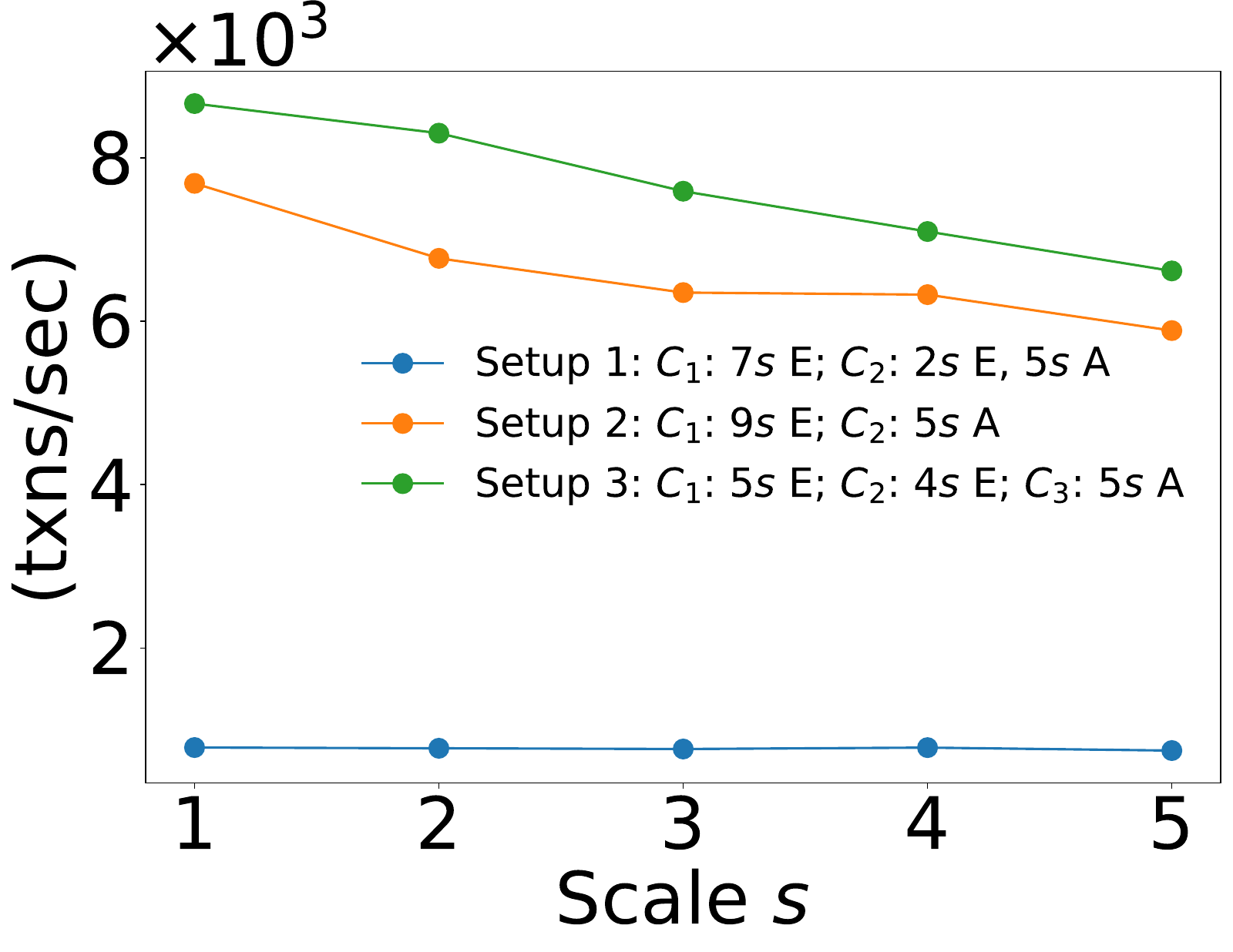} 
\caption{E3. \textsc{Ava-BftSmart}}
\label{fig:optClusters_bft}
\end{subfigure}

\begin{subfigure}{.24\textwidth}
\includegraphics[width=\linewidth]{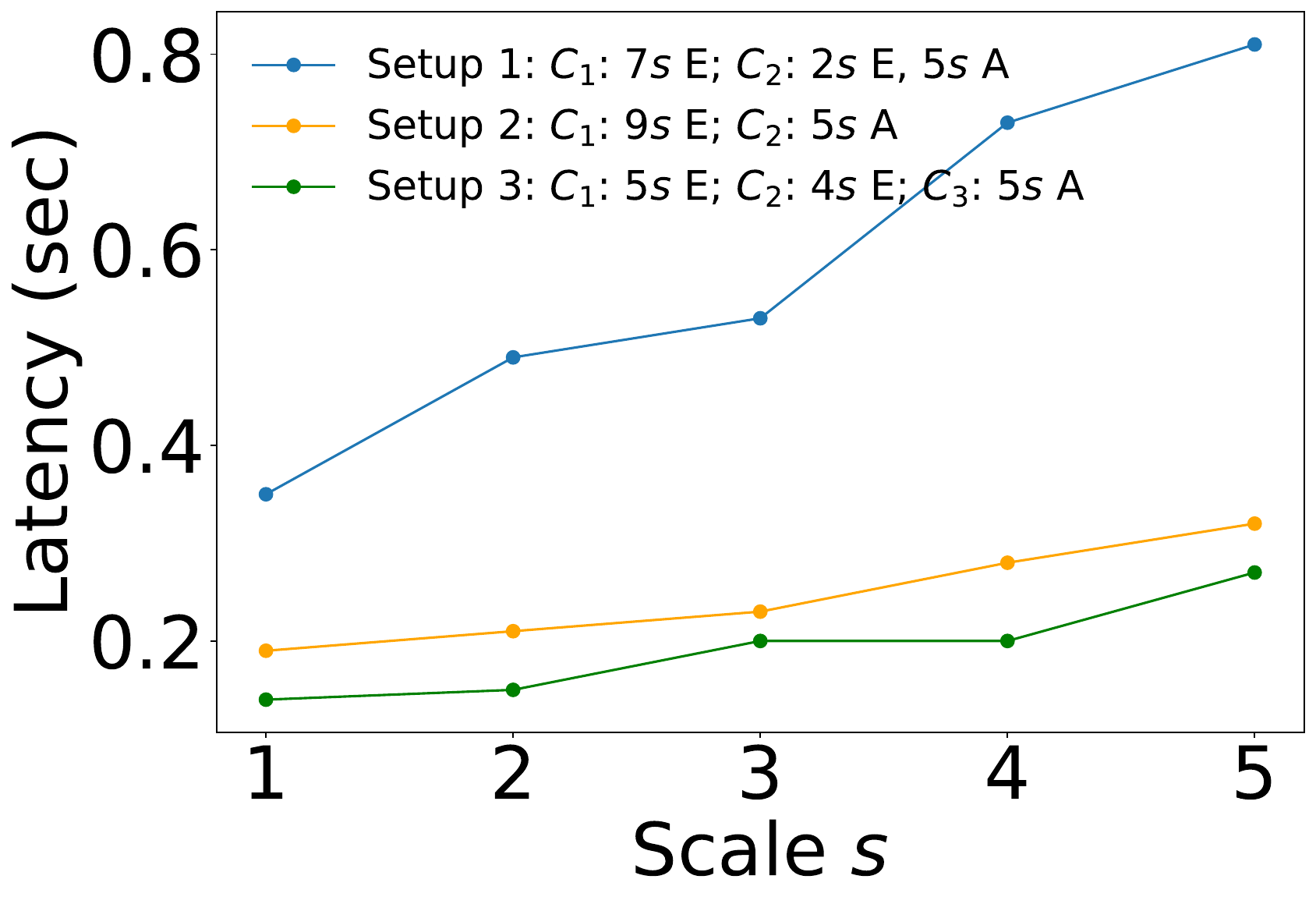} 
\caption{E3. \textsc{Ava-BftSmart}}
\label{fig:optClusters_bft2}
\end{subfigure}
\begin{subfigure}{.24\textwidth}
\includegraphics[width=\linewidth]{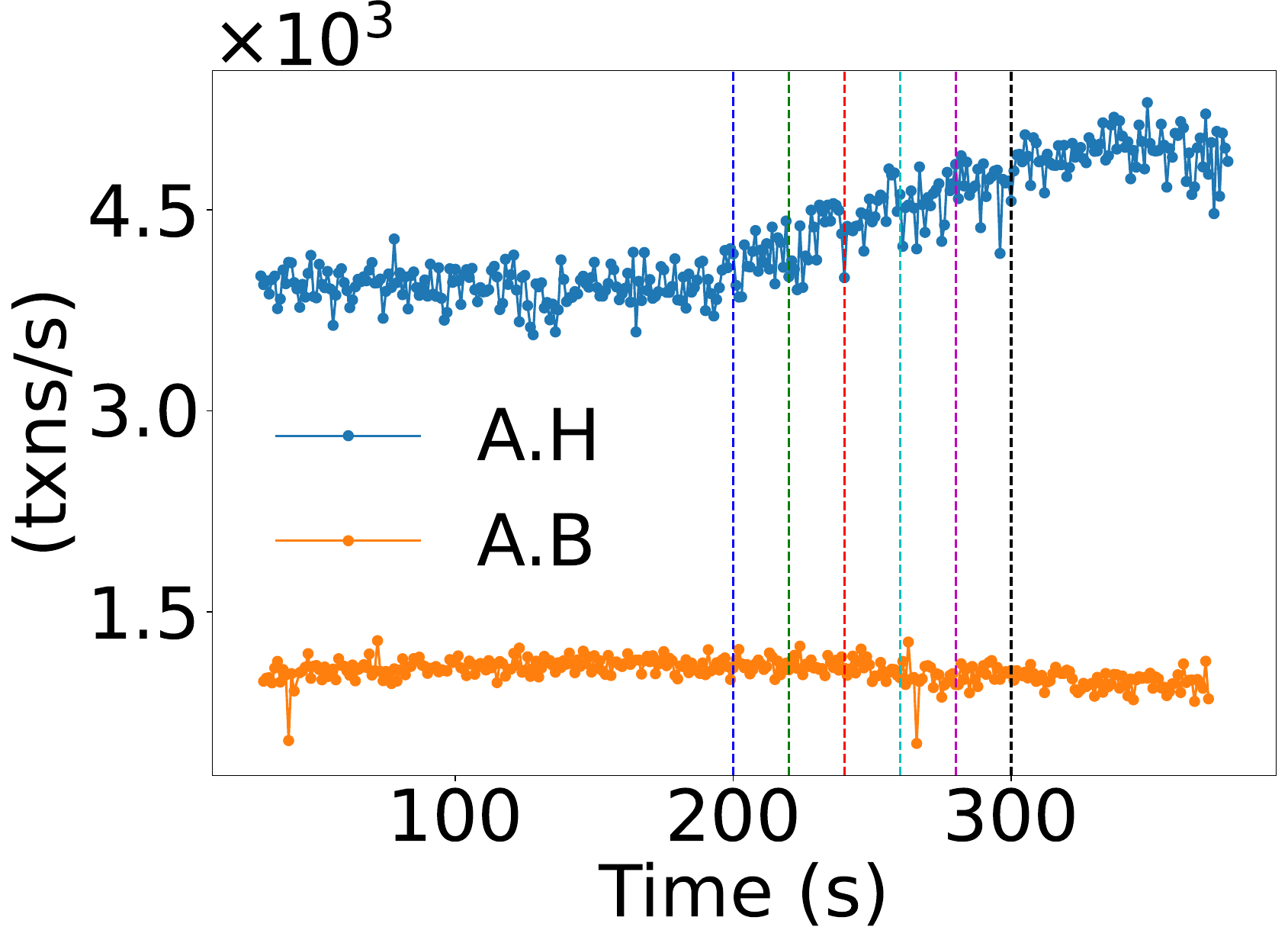}
\caption{E4.1}
\label{fig:MultipleNonLeaderFailure}
\end{subfigure}
\begin{subfigure}{.24\textwidth}
\includegraphics[width=\linewidth]{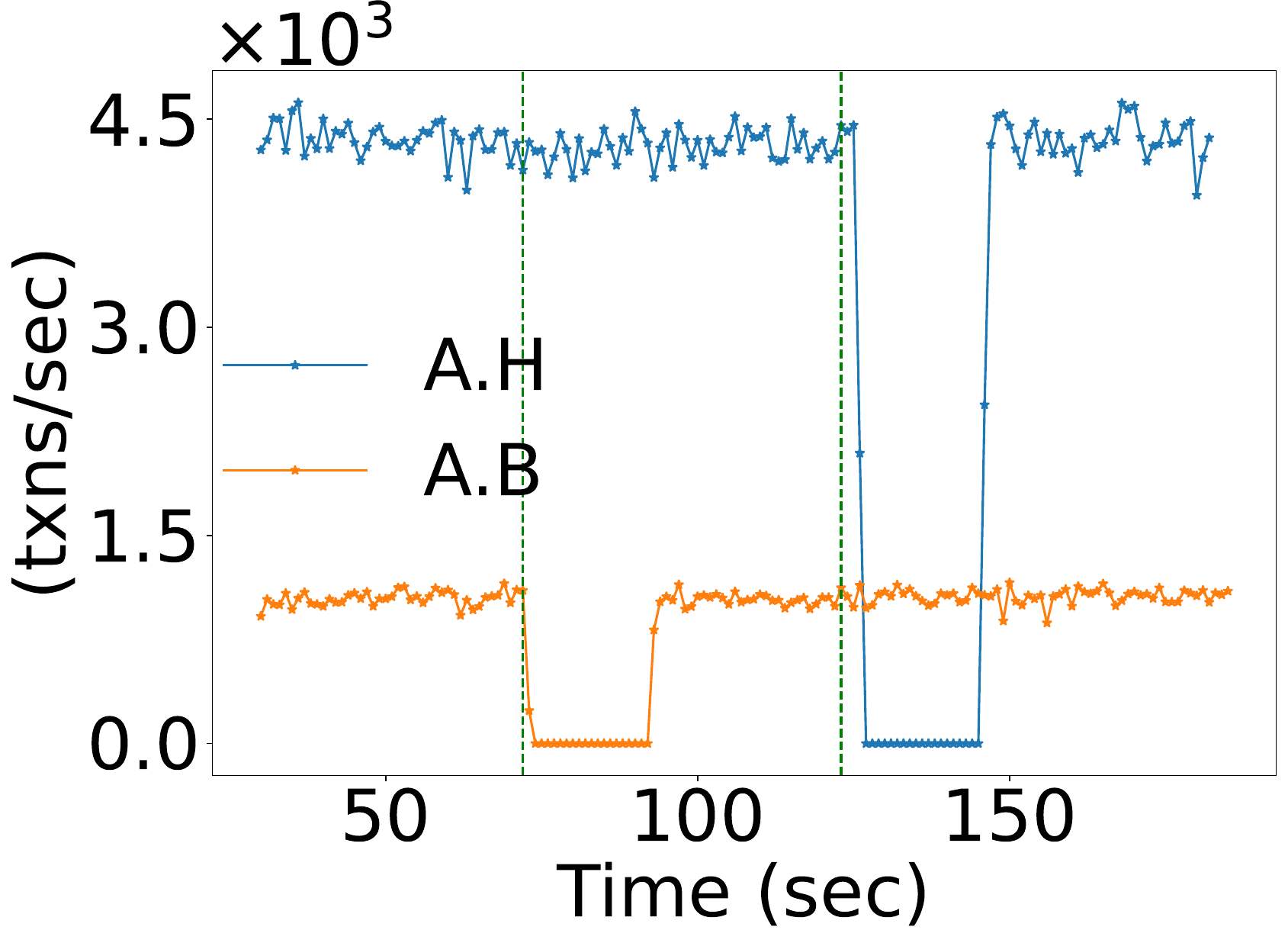} 
\caption{E4.2}
\label{fig:SingleLeaderFailure}
\end{subfigure}
\begin{subfigure}{.24\textwidth}
\includegraphics[width=\linewidth]{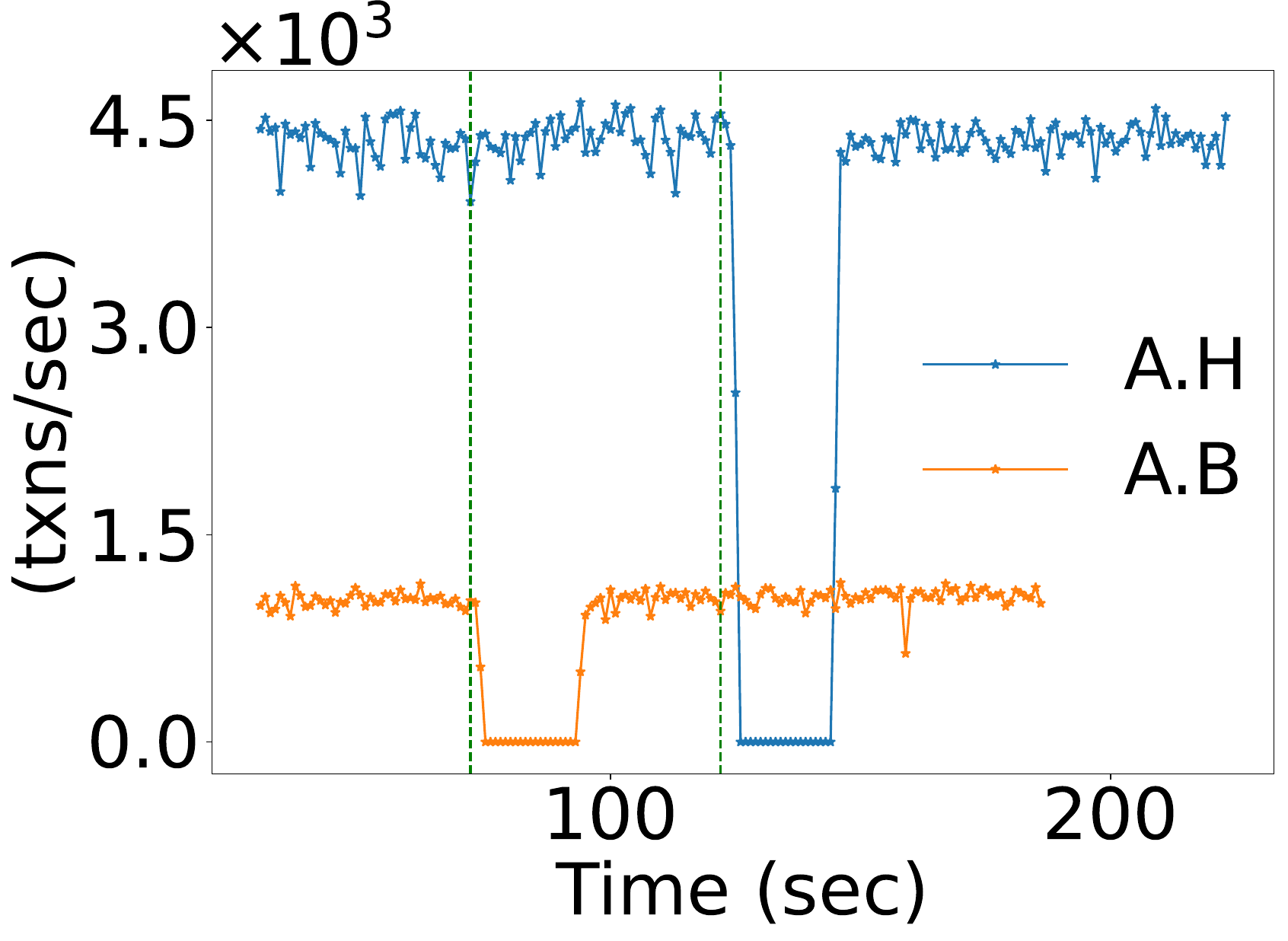} 
\caption{E4.3}
\label{fig:SingleLeaderFailureRVC}
\end{subfigure}

\caption{
E2. Latency breakdown for 
\textsc{Ava-BftSmart} (left)
and
\textsc{Ava-Hotstuff} (right and shaded) in (a).
1 region: Asia, 2 regions: EU and Asia, and 3 regions: EU, Asia, US.
E3. Impact of heterogeneity on throughput and latency for \textsc{Ava-Hotstuff} in (b) and (c) and for \textsc{Ava-BftSmart} in (d) and (e).
E4. 
1. The Impact of multiple non-leader failure on throughput in (f).
2. The Impact of leader failure on throughput in (g).
3. The impact of remote leader change on throughput in (h).
}
\end{figure*}

\begin{figure}
\vspace{-4mm}  
\begin{subfigure}{.235\textwidth}
\includegraphics[width=\linewidth]{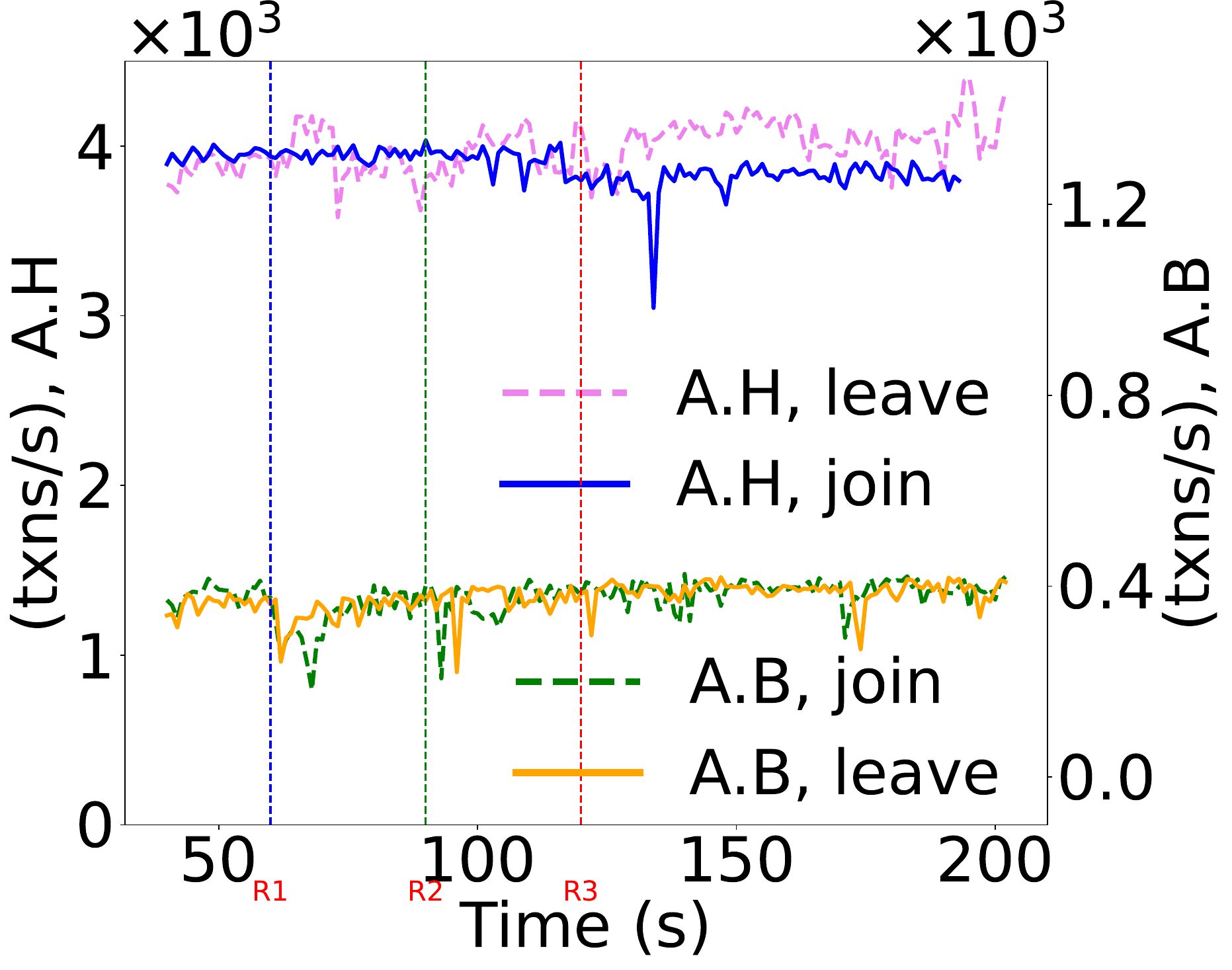} 
\caption{E5.1}
\label{fig:FrequentReconfigurationImpact_j1}
\end{subfigure}
\begin{subfigure}{.245\textwidth}
\includegraphics[width=\linewidth]{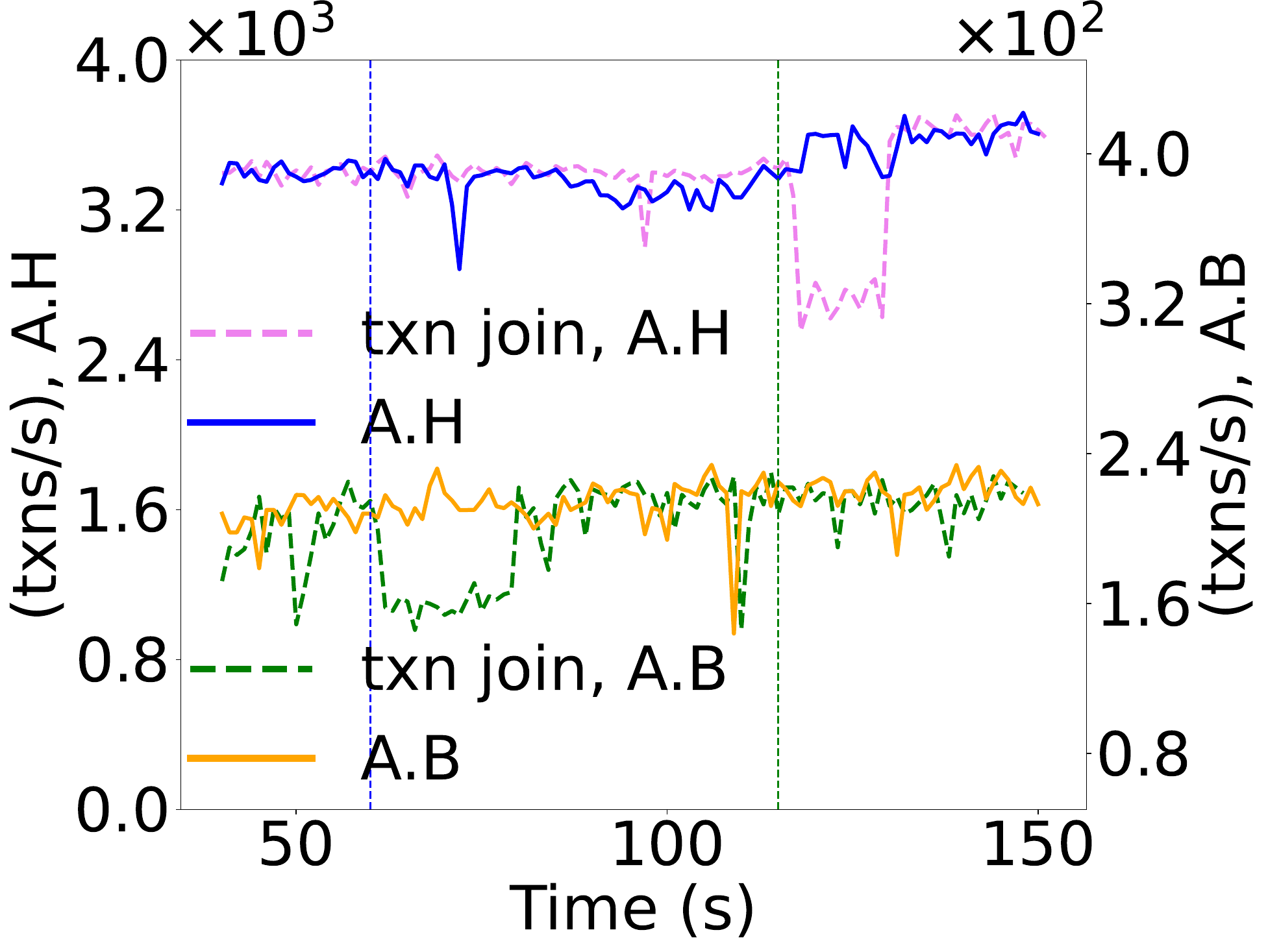} 
\caption{E5.2}
\label{fig:reconfig_stress}
\end{subfigure}
\caption{
E5. 
1.
The impact of multiple reconfigurations on throughput in (a).
2.
The impact of the parallel workflows on on throughput in (b).
(The left y-axis is for \textsc{Ava-Hotstuff} and the right y-axis is for \textsc{Ava-BftSmart}.)
}  
\vspace{-6mm}  
\end{figure}

\textbf{E5. Reconfiguration Requests. \ }
We investigate the impact of reconfiguration on the performance in two experiments. 
\textit{(1)} 
In a system of two clusters with 7 nodes each,
we issue three join and three leave requests to each cluster at the vertical lines.
The requesting nodes can properly join and leave the clusters.
\autoref{fig:FrequentReconfigurationImpact_j1} 
presents the throughput of the whole system during the reconfigurations.
We observe that the throughput slightly decreases as nodes have to communicate with the requesting node in addition to processing transactions.
Further, for joins, the throughput has a slightly decreasing trend
since the local ordering stage is less efficient in larger clusters. 
However, for leaves, the throughput stays steady.
\textit{(2)}
As we described in \autoref{sec:overview},
the protocol takes reconfigurations off the critical path that orders transactions,
and processes them in a parallel workflow.
In this experiment, we compare the parallel workflows with a single workflow that processes reconfigurations in the same sequence as transactions.
We setup connections between replicas of two clusters with 10 and 8 nodes, respectively, and 3 clients:
one client for each cluster that issues write-only transactions, one dedicated client that issues join and leave requests.
A node is repeatedly made to join and leave the system.
\autoref{fig:reconfig_stress} shows the throughput of both \textsc{Ava-Hotstuff} and \textsc{Ava-BftSmart} and their single workflow versions.
The configuration starts for BftSmart and Hotstuff at 60 sec and 115 sec respectively.
We find that the parallel workflows outperform the single workflow in both systems.
In a single workflow, the reconfigurations take slots from transactions and are processed in sequence.
In contrast, \textsc{Hamava} processes them in parallel with transactions, and collects them as a set rather than ordering them individually.

\textbf{E6. Comparison with GeoBFT. \ }
\textcolor{black}{
\autoref{fig:comparegeobft} compares the throughput and latency of \textsc{Ava-Hotstuff} with GeoBFT \cite{geobftcode}.
We experiment with both 
(1) clusters in the same region (\autoref{fig:comparegeobft1}), 
and 
(2) clusters in multiple regions (\autoref{fig:comparegeobft2}).
The number of replicas (48) and clients (24) are fixed across all data points. 
In this experiment, similar to \textsc{Ava-Hotstuff}, we run GeoBFT on 2-core machines.
Compared to \textsc{Ava-Hotstuff},
GeoBFT has lower latency in fewer clusters and similar latency in more clusters. Its throughput surpasses \textsc{Ava-Hotstuff} in fewer clusters, while \textsc{Ava-Hotstuff} slightly outperforms when increasing the number of clusters. Similar results are observed for the case where the clusters are located in multiple regions, as seen in figures \ref{fig:comparegeobft2}. But what is important to highlight is that, unlike GeoBFT, \textsc{Ava-Hotstuff} is able to support cluster reconfiguration. 
}

 

\begin{figure}
\vspace{-6mm}
\begin{subfigure}{.22\textwidth}
\includegraphics[width=\linewidth]{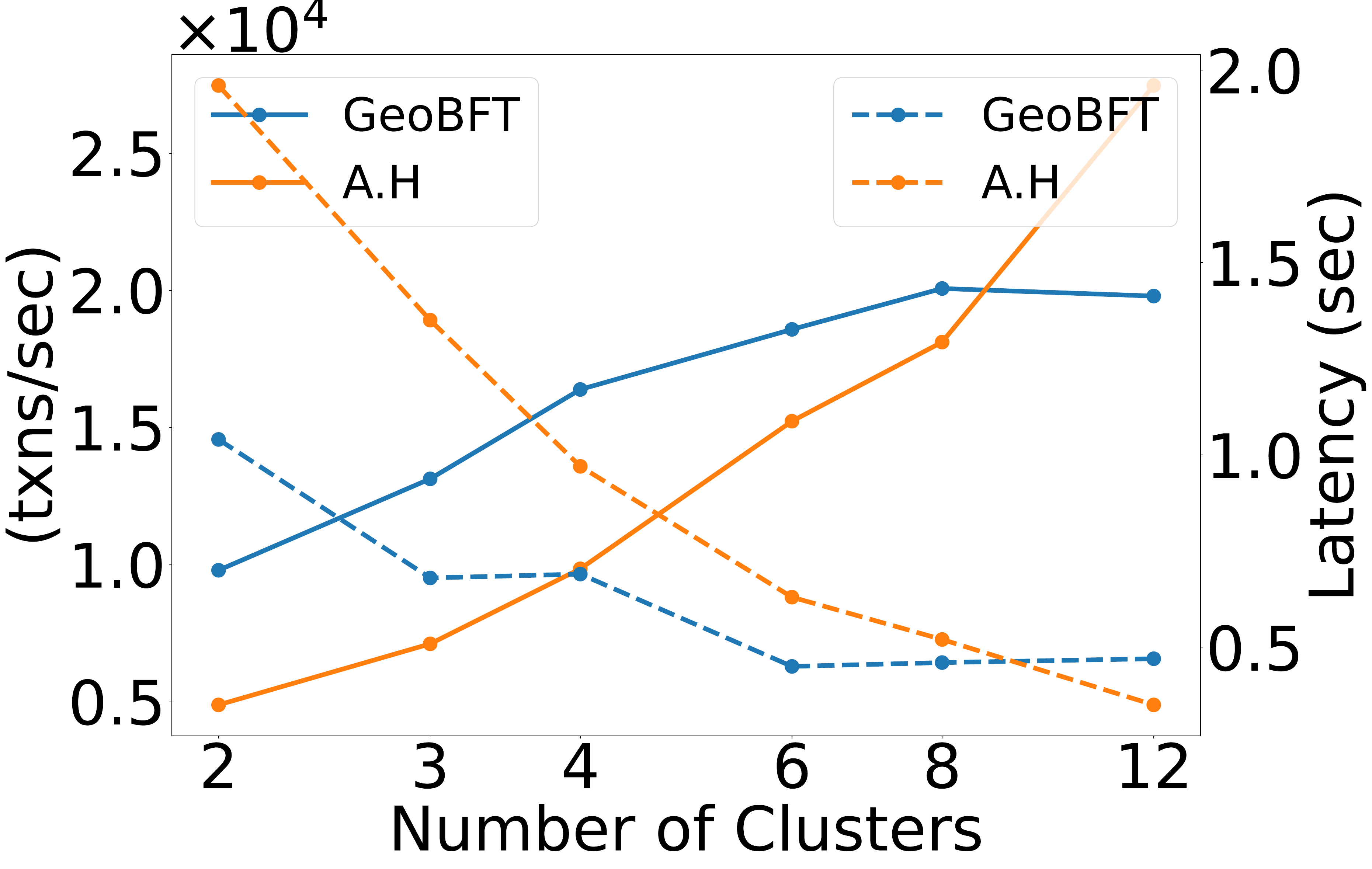} 
\caption{E6.1. Throughput and Latency}
\label{fig:comparegeobft1}
\end{subfigure}%
\begin{subfigure}{.22\textwidth}
\includegraphics[width=\linewidth]{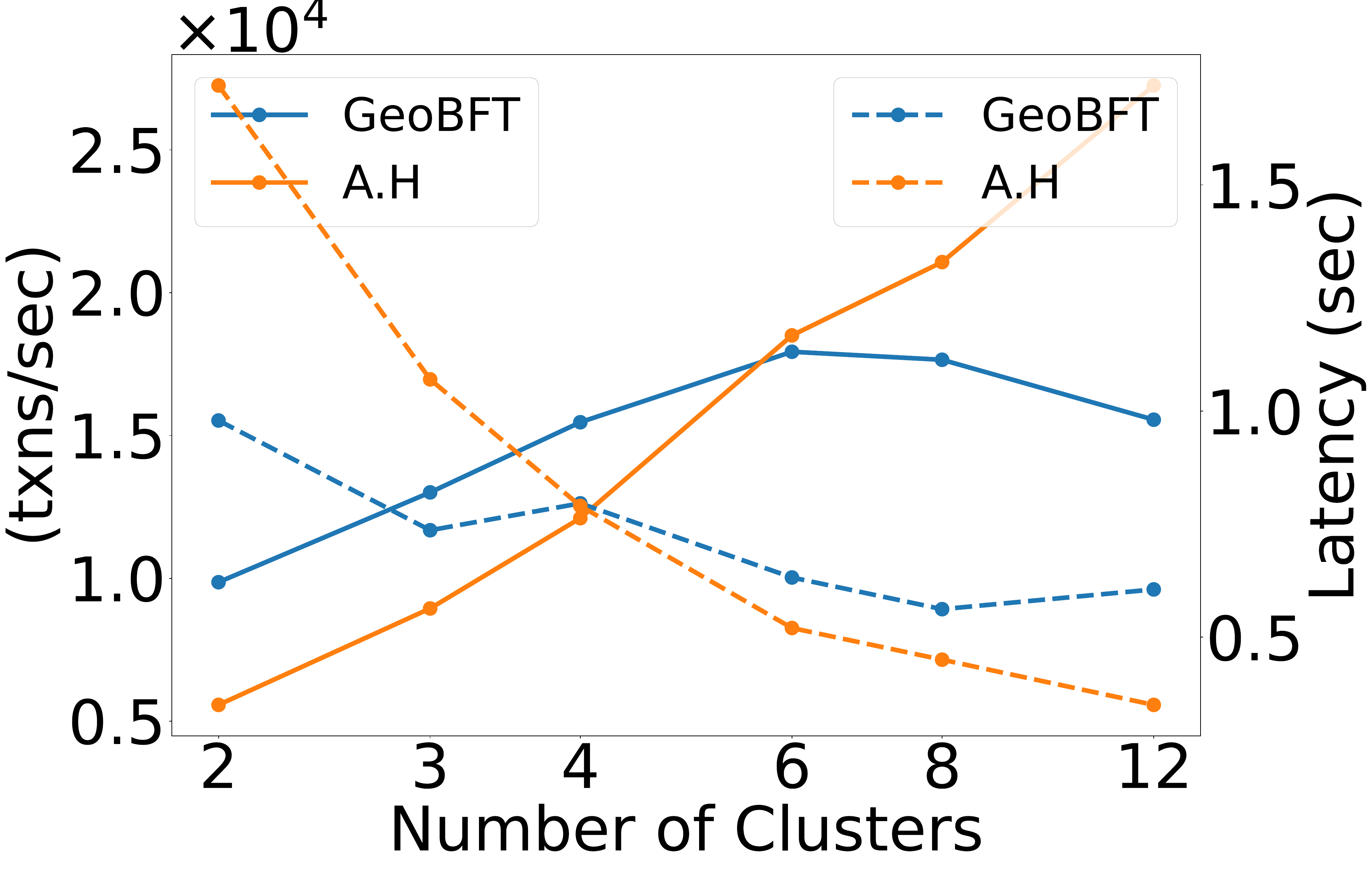} 
\caption{E6.2. Throughput and Latency}
\label{fig:comparegeobft2}
\end{subfigure}%
\caption{
\textcolor{black}{
E7. \textsc{Ava-Hotstuff} vs GeoBFT.
1. Clusters located in the same region:
(a) throughput and latency.
2. Clusters located in multiple regions:
(b) throughput and  latency.
}  
}
\label{fig:comparegeobft}
\vspace{-6mm}
\end{figure}

\textbf{E7. Impact of Reconfiguration Frequency on System Performance. \ }
\textcolor{black}{
\autoref{fig:reconfimp}, we illustrate the effect of the frequency of reconfigurations on 
the throughput and latency of \textsc{Ava-Hotstuff} and \textsc{Ava-Bftsmart}.
We experiment on a system with two clusters, where each has 10 replicas, and each has a client that issues transaction requests and another client that issues an increased number of reconfiguration requests.
Reconfigurations begin at 
80 seconds. 
We experiment with two frequencies:
(1) once every 20 seconds, and
(2) continuously, without any delay between reconfiguration requests.
With \textsc{Ava-Hotstuff}, as expected, increasing the reconfiguration frequency does impact the performance, but in the worst case, the system stabilizes with a throughput drop of at most 12\% and a latency increase of up to 15\%. While with \textsc{ava-bftsmart}, we observe a lower impact of reconfiguration frequency, resulting in a throughput drop of less than 10 percent and a latency increase of approximately 12 percent.}

\textbf{E8. Impact of network latency on performance during reconfiguration. \ }
\textcolor{black}{
\autoref{fig:ReconfigLatImpact} shows the impact of network latency on 
the throughput and latency of \textsc{Ava-Hotstuff} and \textsc{Ava-Bftsmart} during reconfiguration. 
We experiment on a system with two clusters where each has 10 replicas, and each has a client that issues reconfiguration requests.
We fixed one cluster at us-west1-b, and experimented with several locations for the second cluster: us-east5-c, asia-northeast1-b, europe-west3-c, and asia-south1-c, with latencies of 52ms, 91ms, 142ms, and 219ms to the first cluster, respectively. 
As the network latency increases, the inter-cluster communication (in phase 2) dominates the performance, and the impact of reconfigurations (in phase 1) diminishes. For \textsc{ava-hotstuff}, in the steady state with continuous reconfigurations, we observe that when we increase the network latency by more than $4$ times from 52 to 219 ms, the throughput decreases by up to $30\%$, and the latency increases by up to $40\%$. While with \textsc{ava-bftsmart}, we observe a throughput decrease of around $65\%$ and a latency increase of over $200\%$ as network latency increases from 52 to 219ms.
}

\vspace{-1mm}

\begin{figure}
\vspace{-8mm}
\begin{subfigure}{.235\textwidth}
\includegraphics[width=\linewidth]{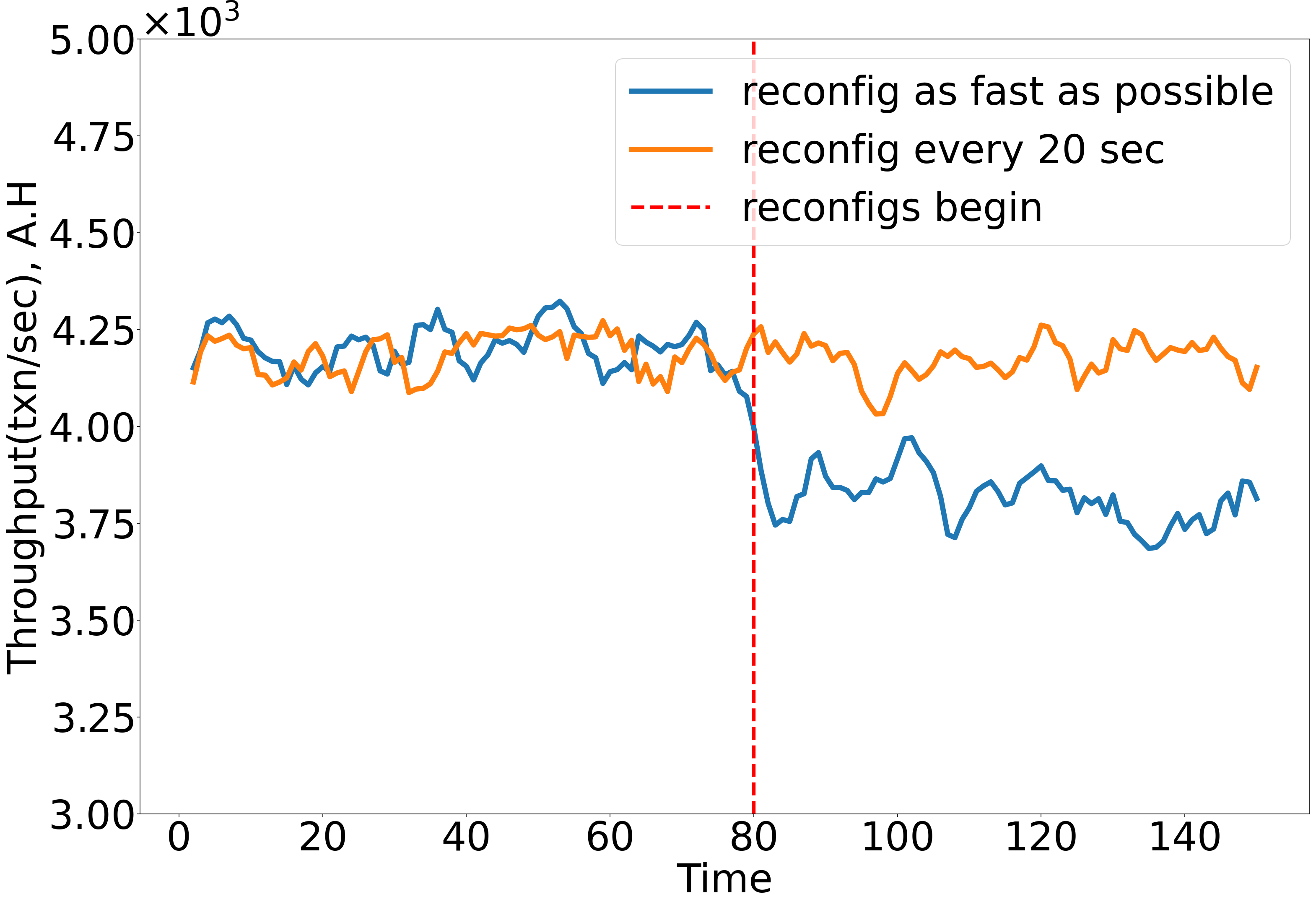} 
\caption{E7.1 Throughput}
\label{fig:ReconfigFreqImpactThput}
\end{subfigure}%
\begin{subfigure}{.235\textwidth}
\includegraphics[width=\linewidth]{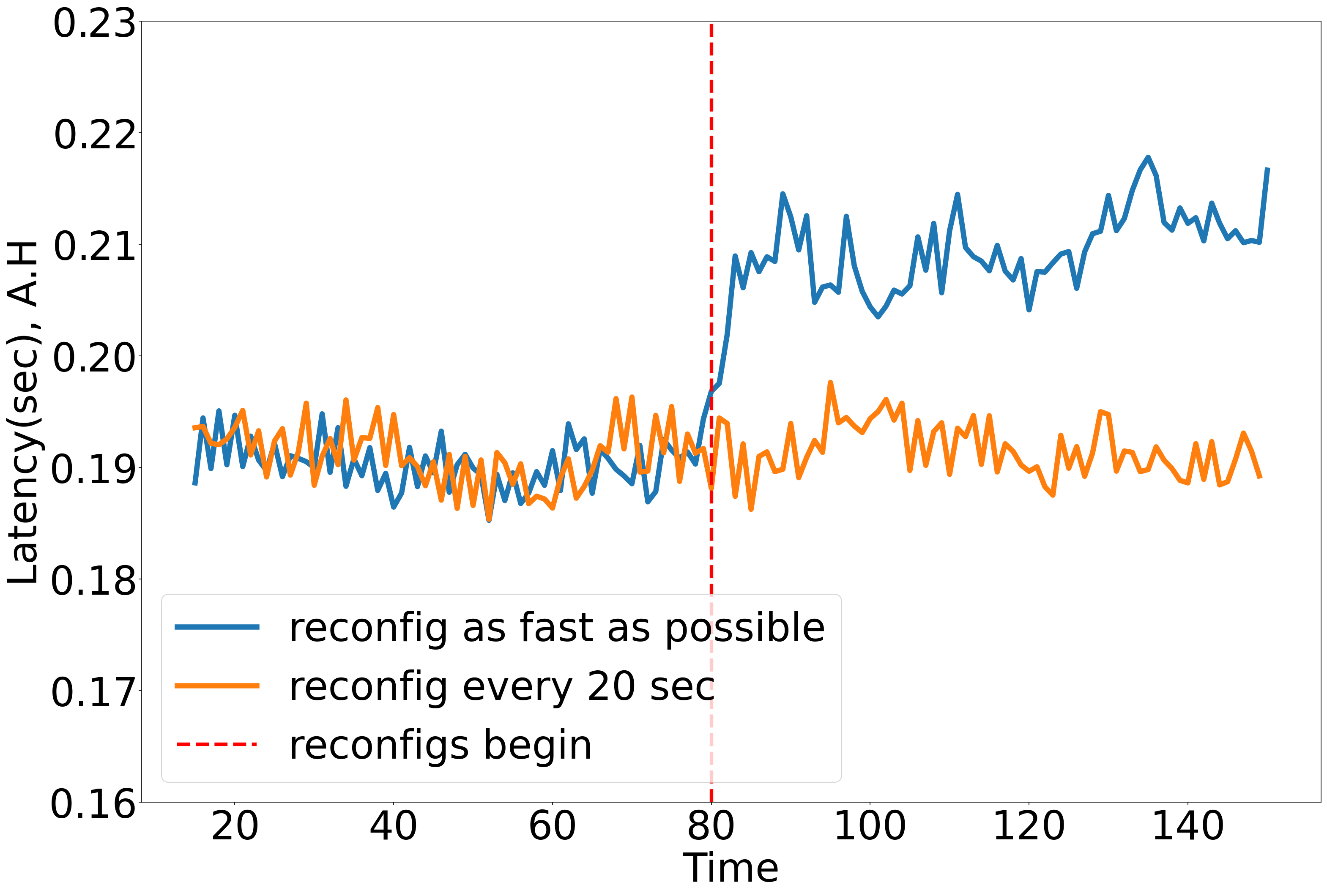} 
\caption{E7.2 Latency}
\label{fig:ReconfigFreqImpactLat_bft}
\end{subfigure}%

\begin{subfigure}{.235\textwidth}
\includegraphics[width=\linewidth]{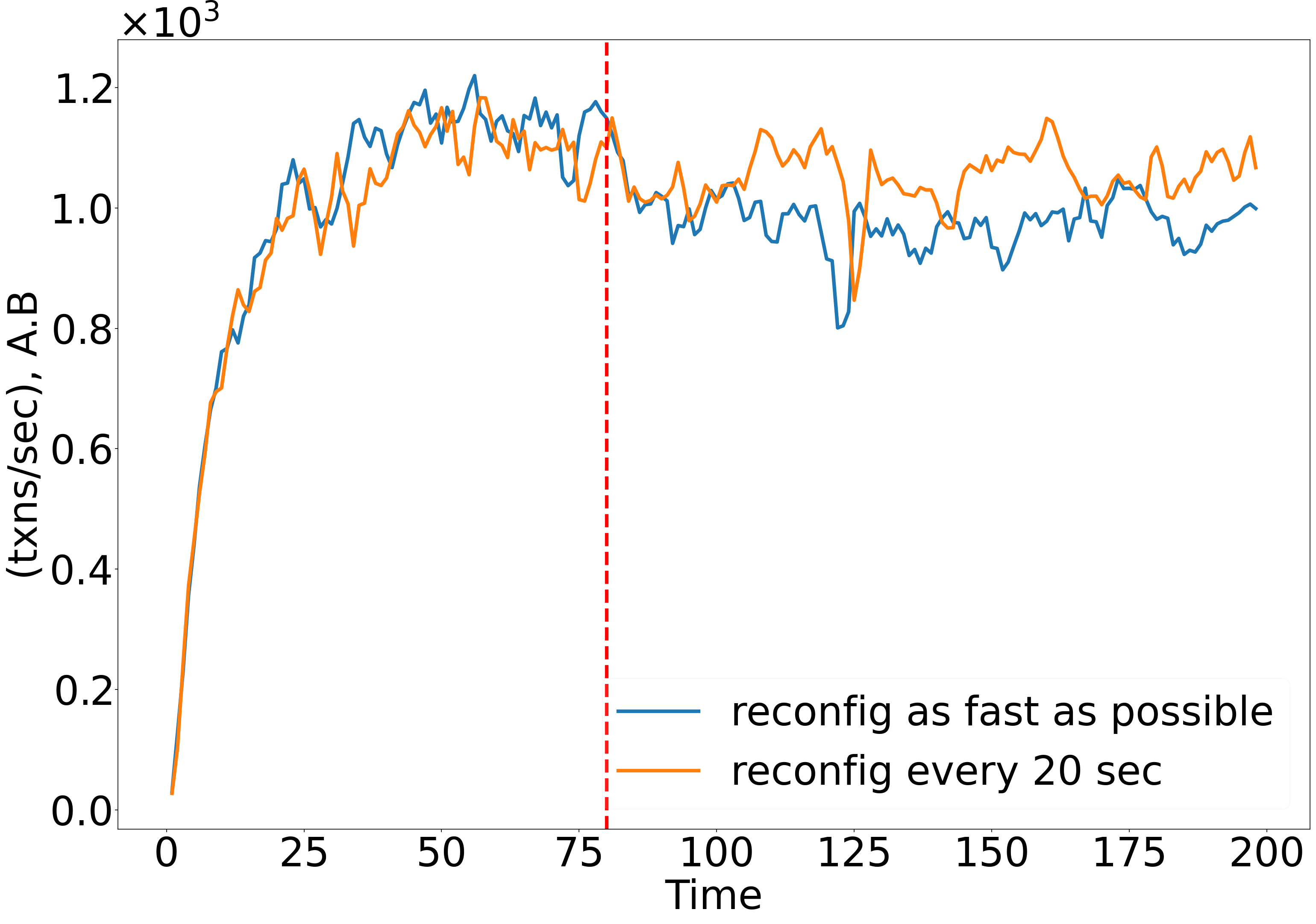} 
\caption{E7.3 Throughput}
\label{fig:ReconfigFreqImpactThput2}
\end{subfigure}%
\begin{subfigure}{.235\textwidth}
\includegraphics[width=\linewidth]{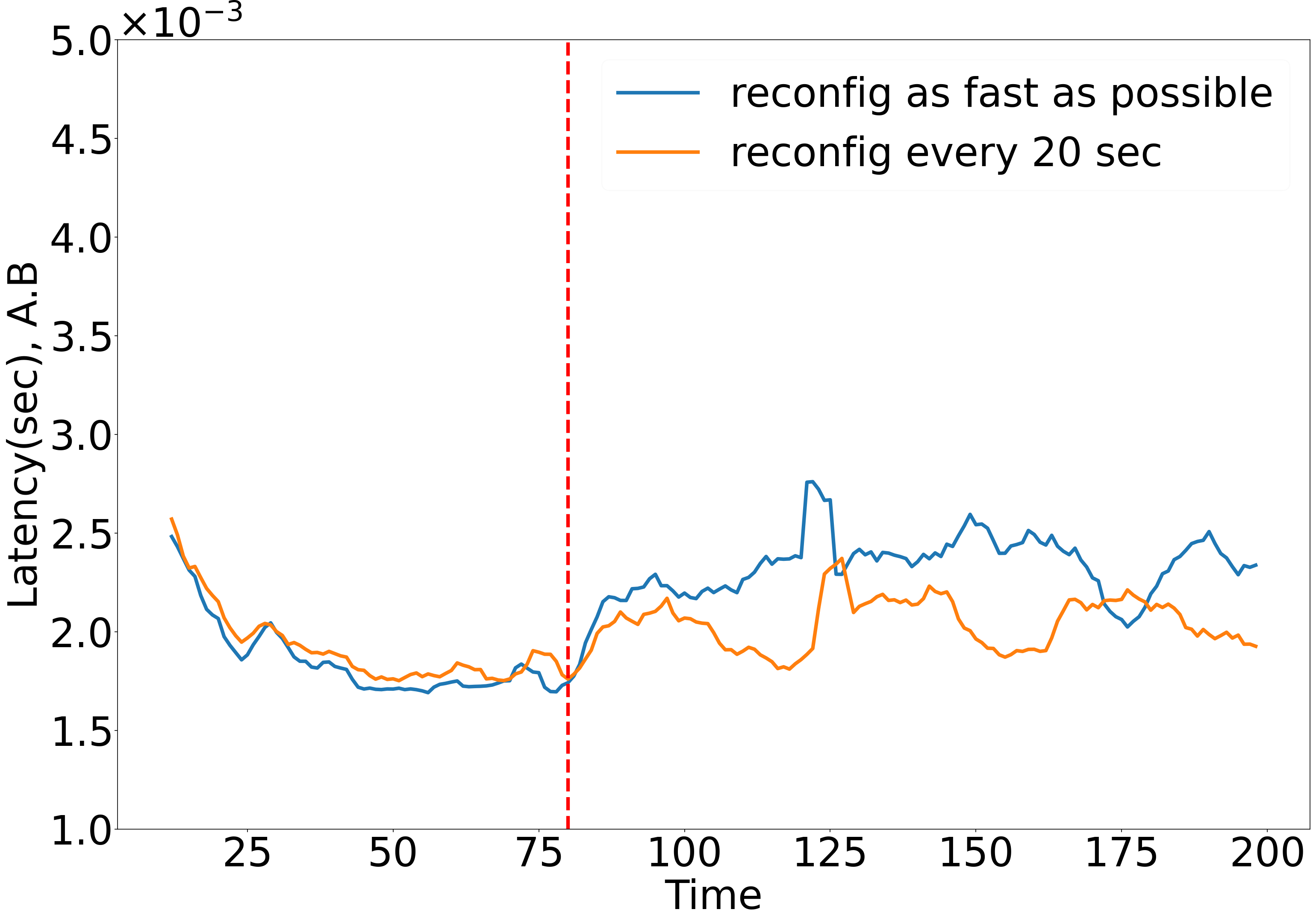} 
\caption{E7.4 Latency}
\label{fig:ReconfigFreqImpactLat_bft2} 
\end{subfigure}%
\caption{
\textcolor{black}{
E7. 
Impact of reconfiguration frequency on 
(a,c) throughput,
and
(b,d) latency.
}  }
\label{fig:reconfimp}
\vspace{-6mm}  
\end{figure}

\section{Related Work}

\textbf{Clustered Replication. \ }
%
%
Compared to non-clustered replication, 
clustered replication has fewer number of replicas in each cluster;
therefore, it exhibits improved performance and scalability.
Steward \cite{amir2008steward} implements a replication protocol where replicas are partitioned into multiple sites.
A leader site is responsible for driving an inter-site coordination protocol similar to Paxos \cite{lamport2001paxos}, which may become the bottleneck.
The inspiring work
GeoBFT \cite{gupta13resilientdb} alleviates the need for a leader site, and enables higher throughput by letting clusters process their own transactions, and then propagate them.
However, it does not support reconfiguration.
Our clustered replication protocol supports heterogeneity and reconfiguration across clusters which allows more flexible and efficient setups.

Another line of work is sharding-based consensus \cite{amiri2021sharper,hellings2021byshard,dang2019towards}.
Elastico \cite{luu2016secure} presents a sharding-based consensus protocol for permissionless blockchains.
OmniLedger \cite{kokoris2018omniledger} and 
RapidChain \cite{zamani2018rapidchain} 
support reconfiguration for sharding-based consensus.
OmniLedger and RapidChain are linearly scalable; but, they suffer from replay attacks in cross-shard commit protocols \cite{sonnino2020replay}.
In contrast, our protocol provides full replication and avoids complications 
of cross-shard synchronization.

\begin{figure}
\vspace{-8mm}
\begin{subfigure}{.235\textwidth}
\includegraphics[width=\linewidth]{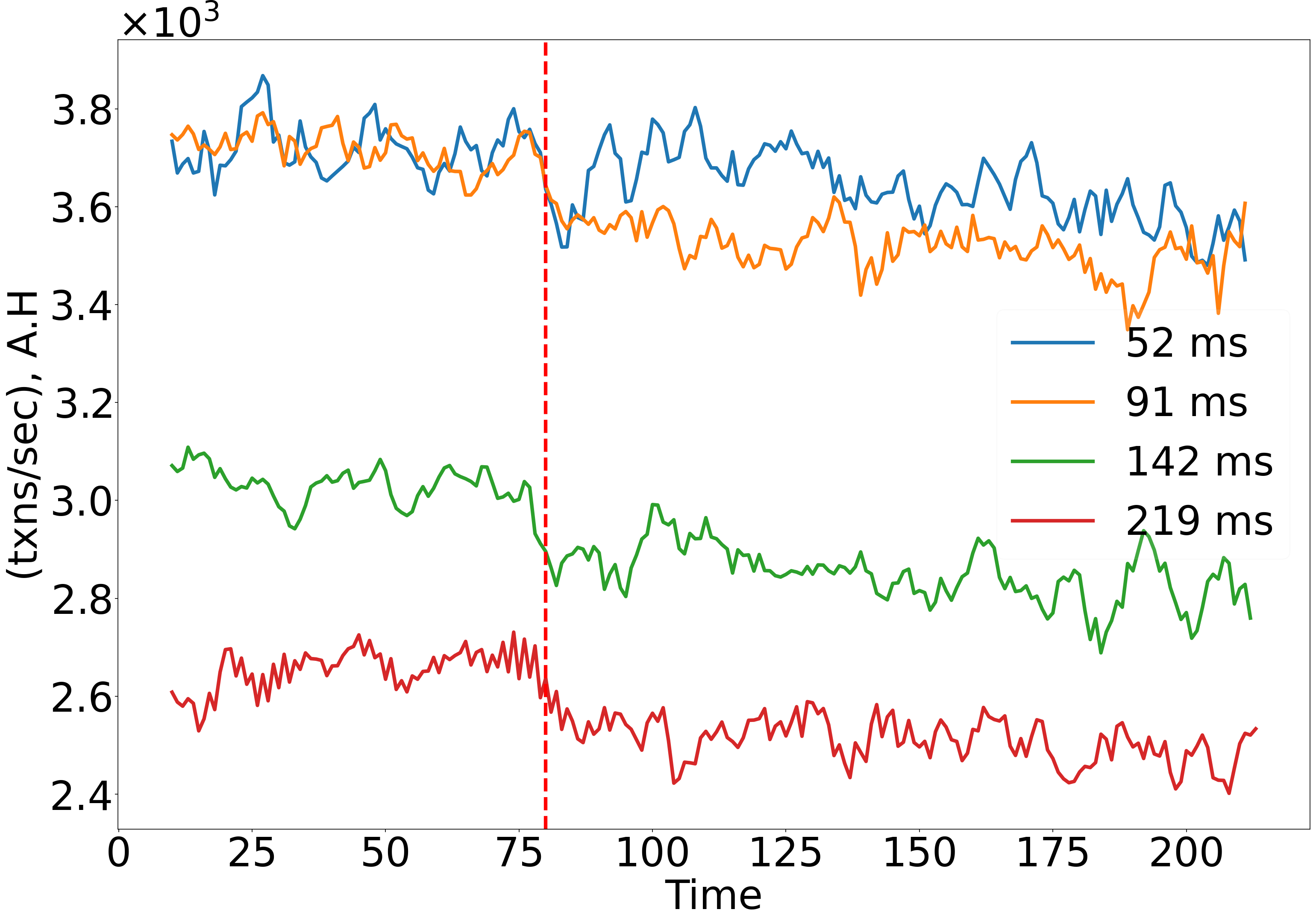} 

\caption{E8.1 Throughput}
\label{fig:ReconfigLatImpactThput}
\end{subfigure}%
\begin{subfigure}{.235\textwidth}
\includegraphics[width=\linewidth]{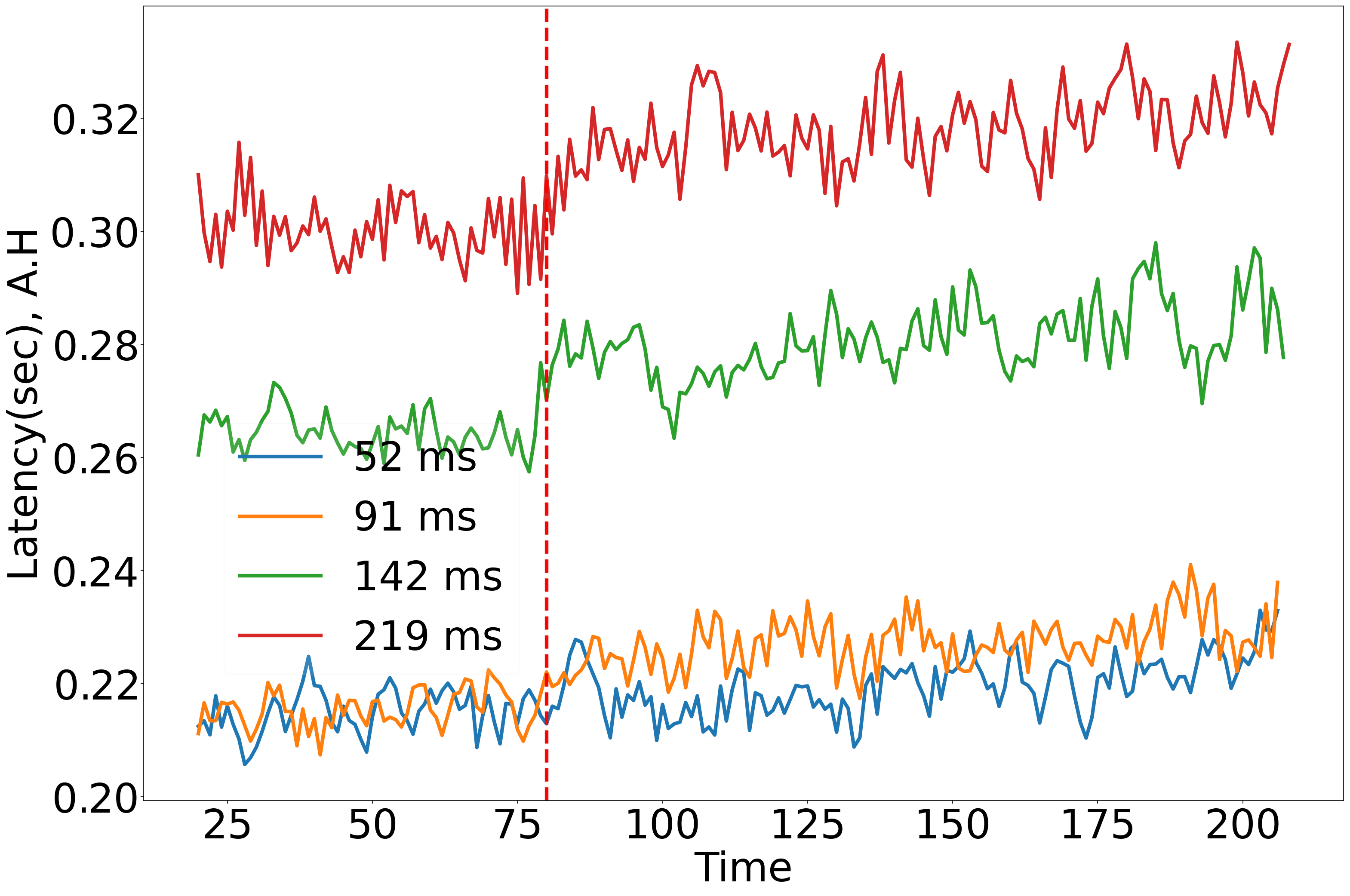} 
\caption{E8.2 Latency}
\label{fig:ReconfigLatImpactLat}
\end{subfigure}%

\begin{subfigure}{.235\textwidth}
\includegraphics[width=\linewidth]{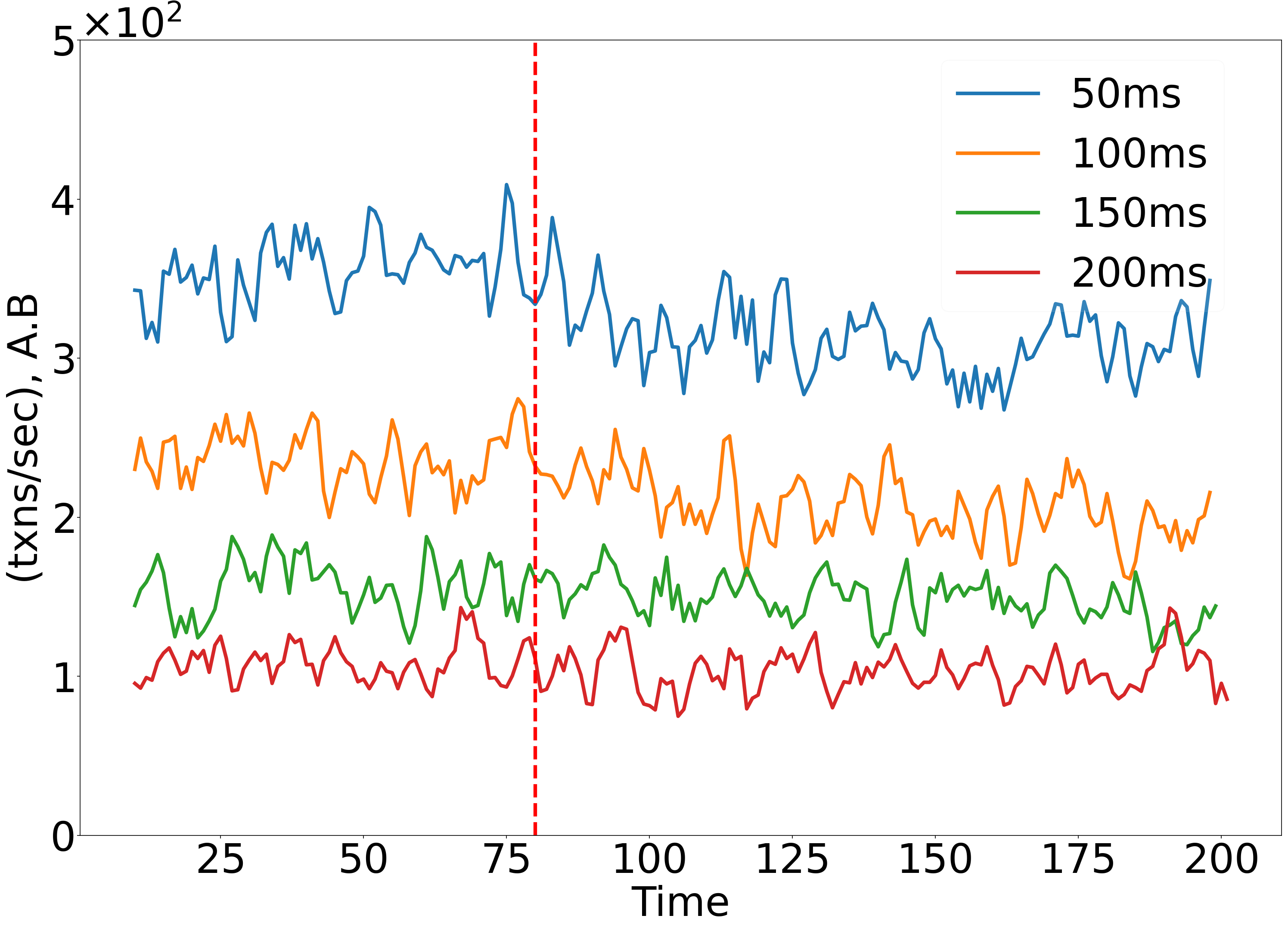} 

\caption{E8.3 Throughput}
\label{fig:ReconfigLatImpactThput_bft}
\end{subfigure}%
\begin{subfigure}{.235\textwidth}
\includegraphics[width=\linewidth]{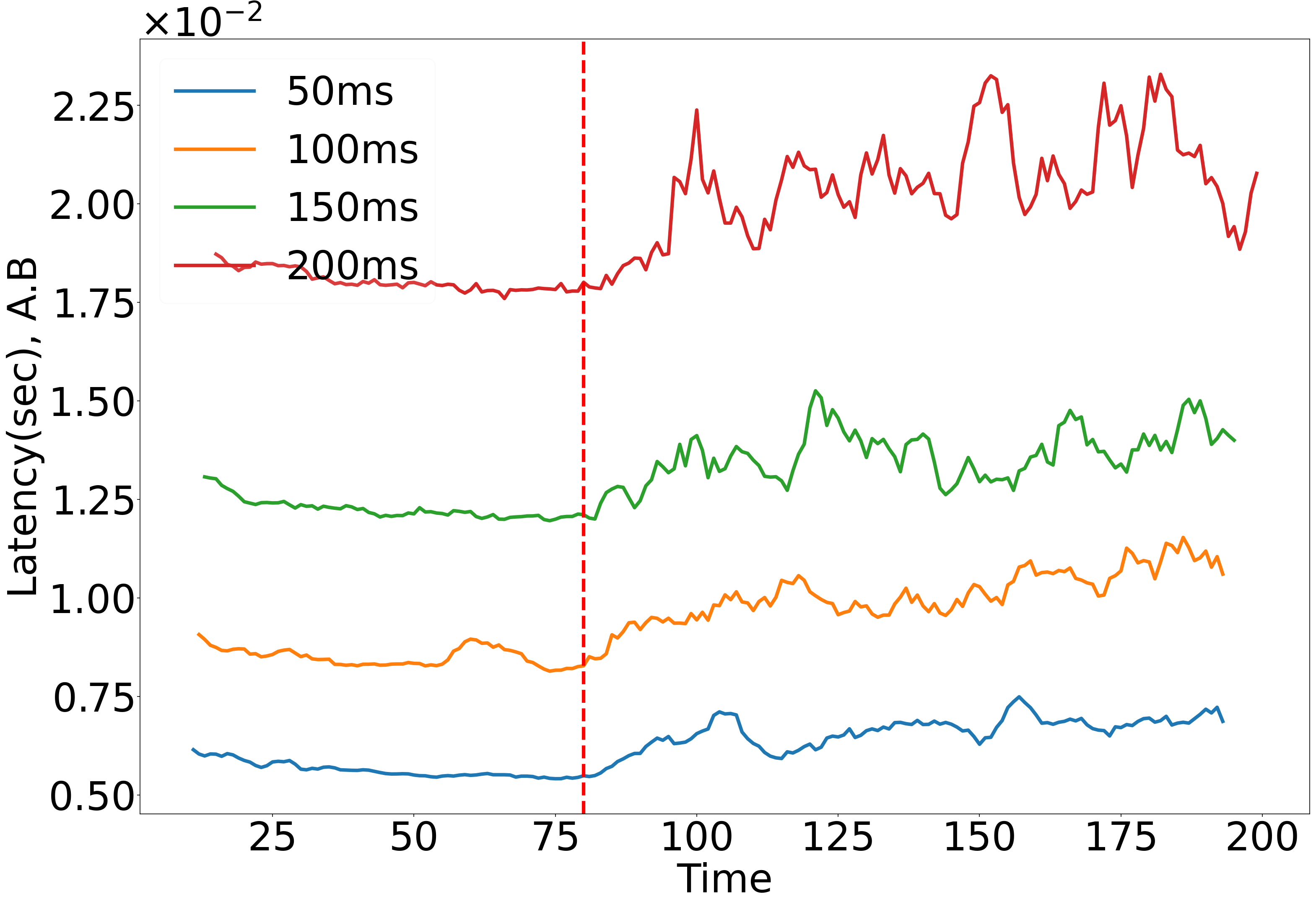} 
\caption{E8.4 Latency}
\label{fig:ReconfigLatImpactLat_bft}
\end{subfigure}%
\caption{
\textcolor{black}{
E8. 
The impact of network latency on 
(a,c) throughput 
and
(b,d) latency
during reconfiguration.
}  }
\label{fig:ReconfigLatImpact}
\vspace{-7mm}  
\end{figure}

\textbf{Group and Open Membership. \ }
A group membership service maintains the set of active replicas by installing new views.
Since accurate membership is as strong as consensus \cite{chandra1996impossibility,chockler2001group},
classical 
\cite{reiter1996secure, lamport2010reconfiguring, lamport1998part, guerraoui2010next, rodrigues2010automatic,bessani2020byzantine,garcia2019lazarus,van2012byzantine}
group membership and reconfiguration protocols
use consensus to reach an agreement on membership and adjust quorums accordingly.
SmartMerge \cite{jehl2015smartmerge}
provides replication, and uses a commutative, associative and 
idempotent merge function on reconfiguration requests to avoid consensus.
It ensures that all replicas eventually perform the merge of all the reconfiguration requests.
Dyno \cite{duan2022foundations} provides 
replication and
group membership in the primary partition model.
Similar to \cite{lorch2006smart}, it uses an instance of consensus to order reconfiguration requests.
In contrast to SmartMerge and Dyno, 
\textsc{Hamava} presents reconfiguration for clustered replication systems without relying on a single 
instance of consensus to safely apply reconfiguration request, a bottleneck that is further amplified 
in a wide area network. Furthermore, SmartMerge reliance on weaker eventual consistency is insufficient.

\vspace{-1mm}

\section{Conclusion}
We presented heterogeneous and reconfigurable clustered replication that 
adapts to different cluster sizes and supports dynamic membership efficiently.
We formalized and proved the safety and liveness properties of our 
reconfiguration protocol, while empirically demonstrating its effectiveness.

\bibliography{Refs}

\clearpage
\section*{Appendix}

\section{Protocol Stages}
\label{sec:protocol-phases-app}

In the overview \autoref{sec:overview} and \autoref{fig:modules-overview},
we explained the structure and the three stages of the protocol.
We presented two sub-protocols in \autoref{sec:inter-cluster-comm} and \autoref{sec:reconfiguration}.
In this section, we elaborate the other sub-protocols.

\textbf{Stage 1:} Stage 1 has two parallel parts.
We saw the reconfiguration part in \autoref{sec:reconfiguration}. 
We now consider local ordering and leader change.

\textit{Local ordering. \ }
We consider local replication for transactions.

In order to process a transaction $t$,
clients can issue a request $\process(t)$ 
at any process of any cluster
(at \autoref{algTOB:local_handler}),
and will later receive a $\return (t, v)$ response.
Each cluster uses 
a total-order broadcast instance $\tob$
to propagate transactions 
to its processes
in a uniform order.
In addition to $\broadcast$ requests and $\deliver$ responses,
the total-order broadcast abstraction can
accept $\newLeader(p, ts)$ requests to install the new leader $p$ with the timestamp $ts$,
and
can issue $\complain(p)$ responses to complain about the leader $p$.
The protocol is parametric for the total-order broadcast.
The total-order broadcast abstracts
the classical non-clustered
Byzantine  replication protocols.
If a complaint is received from $\tob$,
(at \autoref{algTOB:complain}),
it is forwarded to 
the leader election module $\lem$
(in \autoref{alg:leader-change}).

Upon receiving a $\process(t)$ request
(at \autoref{algTOB:local_handler}),
the process uses the $\tob$ 
to broadcast the transaction in its own cluster
(at \autoref{algTOB:issue_TOB}).
Each process stores the operations $\operations_{j}$ that it receives from each cluster $C_{j}$.
Upon delivery of a transaction $t$ from the $\tob$
(at \autoref{algTOB:TOB_received}), 
the process appends $t$ to $\operations_\i$ received from this cluster $C_\i$
(at \autoref{algTOB:append-trans}).
Each process keeps the number $\i$ of the current cluster $C_\i$ that it is a member of.
(We use the index $\i$ only for the current cluster, and the index $j$ for other clusters.)
The $\tob$ delivers a transaction $t$ with a commit certificate $ \sigma $
that is the set of signatures of the quorum that committed $t$.
Each process keeps the set of certificates $\certs$ for the transactions committed in its local cluster
(at \autoref{algTOB:add-cert}).
In the next stage, 
the leader sends the transactions together with their certificates to other clusters.
The certificates prevent Byzantine leaders from sending forged transactions.

In parallel to 
receiving $\process(t)$ requests and ordering transactions $t$,
processes can receive and propagate 
$\join$ and $\leave$
reconfiguration requests.
We will describe the reconfiguration protocol
in the next subsection.
In each round,
the processes of each cluster should agree on the reconfigurations 
before the end of the intra-cluster replication stage (phase 1).
The reconfigurations are then propagated to other clusters 
in the inter-cluster broadcast stage (phase 2).
In stage 1,
a process collects the set of reconfiguration requests $\recs$.
It then calls the function $\sendRecs$ 
(at \autoref{algTOB:send-tentative})
to send the set of reconfigurations it has collected to the leader
who aggregates and uniformly replicates them.
In \autoref{sec:reconfiguration}, 
we presented the Byzantine Reliable Dissemination component
that collects and sends reconfigurations to the leader.
A process calls this function
towards the end of stage 1,
\ie,
when a large fraction $ \alpha $ of the transaction batch is already ordered
(at \autoref{algTOB:send-tentative-cond}).
This leaves ample time in the beginning of stage 1 
to accept reconfiguration requests,
and also leaves enough time at the end of stage 1
to reach agreement for the reconfigurations.
Finally, at the end of stage 1,
when 
$\operations_\i$ contains both the batch of transactions and the reconfigurations
(\autoref{algTOB:batch-size}),
if the current process (denoted as $\self$) is the leader
(\autoref{algTOB:self-leader}),
it calls the function $\interBroadcast$
(at \autoref{algTOB:sharing_request})
to start the inter-cluster broadcast stage (phase 2).

\begin{algorithm}
   \small
    \caption{Local Ordering} 
    \label{alg:phase1}
    
	\DontPrintSemicolon
	\SetKwBlock{When}{when received}{end}
	\SetKwBlock{Upon}{upon}{end}
	$\request : \process (t)$ \;
	$\response : \return (t, v)$ \;

	\Uses : \;
	\ \ \ $\tob : \TOB$ \;
	\ \ \ \ \ \ $\request : \broadcast (t), \ \newLeader (p, ts)$ \;
	\ \ \ \ \ \ $\response : \deliver (p, t), \ \complain (p)$ \;

   \Vars : \;
   \ \ \ $r$ \AComment{The current round}   
   \ \ \ $\i$ \AComment{The number of the current cluster}
   \ \ \ $\self$ \AComment{The current process}   
   \ \ \ $\leader : P  \leftarrow  p_{0}^\i$ 
   \AComment{The leader of current cluster $C_\i$}
   \ \ \ $ts  \leftarrow  0$ \AComment{Timestamp for $\leader$}
   \ \ \ $\operations_{j}  \leftarrow   \emptyset $  \AComment{Operations from each cluster $C_{j}$}
   \ \ \ $\certs$
   \AComment{Certificates for $\operations_\i$ of $C_\i$}

   \vspace{0.7ex}    
   \Upon($\request \ \process \mbox{$(t)$}$ \label{algTOB:local_handler}) { 
      $\tob \ \request \ \broadcast(t)$ \label{algTOB:issue_TOB} \;
	}

   \vspace{0.7ex}	
	\Upon(\mbox{$\tob \ \response \ \deliver(p, t^ \sigma )$} \label{algTOB:TOB_received}) {
	    append $\Transaction(p, t)$ to $\operations_\i$\label{algTOB:append-trans} \;
	    add $ \sigma $ to $\certs$ \label{algTOB:add-cert}\;
        
      \If{$ |\operations_\i| = \batchSize  \times   \alpha $\label{algTOB:send-tentative-cond}}{
         $\call \ \sendRecs()$ \label{algTOB:send-tentative} \;
      }
      \ElseIf{$ |\operations_\i| = \batchSize + 1$\label{algTOB:batch-size}}{
         \ACommentL{$\batchSize$ transactions + $1$ reconfiguration set}
         \If{$\self = \leader$\label{algTOB:self-leader}}{
            $\interBroadcast(r, \operations_\i, \certs)$ \label{algTOB:sharing_request}\;
         }
      }
   }
   
   \Upon(\mbox{$\tob \ \response \ \complain(p)$} \label{algTOB:complain}) {   
        $\call \ \complain(p)$ \label{algTOB:complain-ld}\;
   }

\end{algorithm}

\textit{Leader Change. \ }
A leader orchestrates both 
the ordering of transactions in the total-order broadcast,
and the delivery of the reconfigurations.
However, a leader may be Byzantine, and may not properly lead the cluster.
Therefore, 
as presented in \autoref{alg:leader-change},
the protocol monitors and changes leaders.
As we described,
the total-order broadcast $\tob$ 
(\autoref{alg:phase1} at \autoref{algTOB:complain-ld})
and
the Byzantine reliable dissemination $\brd$
(\autoref{alg:reconf1} at \autoref{algTOB:brd-call-complain-fun})
complain
when the delivery of transactions or reconfigurations is not timely.
The complains 
are sent to the leader election module $\lem$
(at \autoref{algJ:fun_complain}-\ref{algLC:complain}).

The protocol uses the classical leader election module $\lem$.
The implementation of this module is presented in \autoref{alg:leader-election}.
Once a quorum of processes send complain requests to $\lem$,
it eventually issues a response $\newLeader(p, ts)$
at all correct processes 
to elect a new leader $p$ with the timestamp $ts$.
Further,
if the current process sends a $\nextLeader$ request to the module,
it issues a response $\newLeader$ at the current process.
This module guarantees that
the leader for each timestamp is uniform across processes,
the timestamps are monotonically increasing,
and
eventually a correct leader is elected.

When a process receives a $\newLeader(p, ts)$ response
(at \autoref{algLC:new_leader}),
it records the new leader and timestamp
(at \autoref{algLC:record-leader-ts}),
and 
forwards the new leader event to the total-order broadcast $\tob$ 
and Byzantine reliable dissemination $\brd$ modules
as well
(at \autoref{algLC:tob-issue-newleader}-\ref{algLC:brc-newleader}).
Further, 
the previous leader might have failed to communicate the operations of the previous round to other clusters.
As we will describe next, clusters wait for the operations of each other in each round;
therefore, a remote cluster can fall behind by at most one round.
Thus,
the new leader sends operations of the previous in addition to the current round
(at \autoref{algLC:self-leader}-\ref{algLC:broadcast-pre}).

\begin{algorithm}
   \small
   \caption{Leader Change}
   \label{alg:leader-change}

   \DontPrintSemicolon
   \SetKwBlock{When}{when received}{end}
   \SetKwBlock{Upon}{upon}{end}    
   \SetKwBlock{Function}{function}{end}
   \SetKwBlock{Foreach}{foreach}{end}
   \SetKwBlock{Match}{match}{end}
   \SetKwBlock{Case}{case}{end}

   \Uses: \ \;
   \ \ \ $\lem : \mathsf{LeaderElection}$ \;
   \ \ \ \ \ \ $\request : \complain (p), \ \nextLeader$ \;
   \ \ \ \ \ \ $\response : \newLeader (p, ts)$ \;

   \Vars: \ \;
   \ \ \ $\preOps$, $\preCerts$
   \AComment{ops and certs of the previous round}

   \Function($\complain$\mbox{$(p)$} \label{algJ:fun_complain}){
      $\lem \ \request \ \complain(p)$ \label{algLC:complain}\;      
   }

   \vspace{0.7ex}    
	\Upon(\mbox{$\lem \ \response \ \newLeader(p, \ts')$} \label{algLC:new_leader}){
	    $ \langle \leader, \ts  \rangle  \leftarrow   \langle p, \ts' \rangle $ \label{algLC:record-leader-ts} \;
	    $\tob \ \request \ \newLeader(\leader, \ts)$ \label{algLC:tob-issue-newleader} \;
	    $\brd \ \request \ \newLeader(\leader, \ts)$ \label{algLC:brc-newleader} \;
       reset $\timer_\i$ \;	    

      \If{$\leader = \self$\label{algLC:self-leader}}{
         \If{$|\operations_\i| = \batchSize$}{
            $\call \ \interBroadcast(r, \operations_\i, \certs)$
         }
         \If{$r > 1$}{
            $\call \ \interBroadcast(r-1,$ $\preOps,$ $\preCerts)$ \label{algLC:broadcast-pre}
         }
      }

   }
   
\end{algorithm}

\begin{algorithm}
    \small
    \caption{Leader Election}
    \label{alg:leader-election}
	\small
   \DontPrintSemicolon
   \SetKwBlock{When}{when received}{end}
   \SetKwBlock{Upon}{upon}{end}
   \SetKwBlock{Function}{function}{end}

   \Implements $:$ \ Leader Election \;
   \ \ \ $\request : \complain (p)$ \;
   \ \ \ $\response : \newLeader (p, ts)$ \;
   \ \ \ $\request : \nextLeader$ \;

   \Uses:	\ \;	
   \ \ \ $\abeb : \mathsf{Authenticated Best Effort Broadcast}$\;

   \Vars $:$   \;
   \ \ \ $\ts \leftarrow 1$ \;
   \ \ \ $C  \leftarrow   \emptyset $ \AComment{Set of complaining processes}
   \ \ \ $c  \leftarrow  \false$ \AComment{Complained} 

   \Upon($\request \ \complain❪ p ❫$\label{algV:l1}){
      \If{$ \neg  c$}{
         $\call \ \mathit{send\text{-}complain} ❪ ❫$
      }
   }
   
   \Function($\mathit{send\text{-}complain} ❪ ❫$ \label{algV:bb}){
      $c  \leftarrow  \true$ \;
      $\abeb \ \request \ \broadcast(\Complaint( \ts ))$ \label{algV:aa} \;
   }
   
   \Upon($\abeb \ \response \ \deliver ❪p, \Complaint ❪ \ts' ❫❫$ where $\ts = \ts'$\label{algV:l0}){    
      $C  \leftarrow  C ∪ \{ p \}$ \label{algV:l2}\;
    
      \If{$|C|  \geq  f + 1 ∧  \neg c$\label{algV:l3}}{
         $\call \ \mathit{send\text{-}complain} ❪ ❫$
      }

      \If{$|C|  \geq  2  \times  f(i) + 1$\label{algV:l5}}{
         $\call \ \mathit{change} ❪ ❫$
      }
   }
    
   \Function($\mathit{change} ❪ ❫$ \label{algV:l9}){
      $\ts  \leftarrow  \ts + 1$ \;
      $C  \leftarrow   \emptyset $ \;
      $c  \leftarrow  \false$ \;            
      $\response \ \newLeader(p_{\ts \ \mathsf{mod} \ N}, \ts)$ \;
      \AComment{Choose leaders in a round robin order.}
      \AComment{$N$ is the number of processes.}
   }

   \Upon($\request \ \nextLeader$\label{algV:l10}){
      $\call \ \mathit{change} ❪ ❫$
   }
   
\end{algorithm}

\textbf{Stage 2:}
We already considered this stage in \autoref{sec:inter-cluster-comm}.

\textbf{Stage 3: Execution. \ }
At the end of the inter-cluster communication stage,
a process receives the batches of operations from each other cluster.
It then calls the $\execute$ function
(\autoref{alg:inter-cluster} at \autoref{algI:execution_request})
that performs the last stage: execution
(at \autoref{alg:order-exec}).
Processes uniformly $\mathit{order}$ the batches of operations:
first, they process the transactions, and then the reconfigurations,
and further,
use a predefined order of clusters 
to order transactions
(at \autoref{algE:order}).
Then, they process each operation:
they apply each transaction and reconfiguration
(at \autoref{algE:match-op}-\ref{algE:newline}).
If a transaction has been issued by the current process,
a $\mathit{return}$ response 
is issued
(at \autoref{algE:retv}).
Finally, in order to prepare for the next round,
the timers and variables are reset
and the round number is incremented
(at \autoref{algE:save-pre}-\ref{algE:reset_vars}).

\textit{Application of Reconfigurations. \ }
The function $\reconfigure$ is called
for each set of reconfigurations $\rc$ from a cluster $j$
(at \autoref{algJ:install_handler}).
First, the process adds joining processes, and removes leaving processes from 
the set of processes $C_j$ of cluster $j$
(at \autoref{algJ:install_join} and \ref{algJ:install_leave}).
Then the function $\mathit{kickstart}$ is called on the reconfigurations of the local cluster $irc$
(at \autoref{algE:callkickstart}).
The function $\mathit{kickstart}$
(at \autoref{algJ:install_handler})
processes all the joins before the leave reconfigurations.
We keep this specific order since leaving processes may still need to send additional messages for the new processes. If they leave first, then the new processes will not be able to collect enough states to start the execution.
If the leave is for the current process,
it issues a $\left$ response
(at \autoref{algJ:install_leave-resp}).
To kick-start a new process $p$,
the members of its local cluster send
a $\CurrState$ message to $p$
(at \autoref{algJ:share_state}).
The message contains the local $\st$,
the current round number $r$, 
and the cluster members $C$.
Further, 
the process
resets its $\echoed$, $\readied$, $\delivered$, and $\valid$ variables.
When a correct process receives $\CurrState$ messages with 
the same state $s'$, 
cluster members $C'$,
and 
round $r'$
from a quorum 
(at \autoref{algJ:reply_received}),
the process sets 
its $\st$, cluster $C$, and round $r$ to the received values.
It then issues a $\joined$ response
(at \autoref{algJ:install_join-resp}).
After an addition or a removal,
the process further updates the failure threshold $f_{j}$ for the cluster $j$
to less than one-third of the new cluster size
(at \autoref{algJ:update_other_var}).

\begin{algorithm}
    \small
    \caption{Stage 3: Execution}
    \label{alg:order-exec}
    
	\DontPrintSemicolon
	\SetKwBlock{When}{when received}{end}
	\SetKwBlock{Upon}{upon}{end}
	\SetKwBlock{Foreach}{foreach}{end}
   \SetKwBlock{Function}{function}{end}
   \SetKwBlock{Match}{match}{end}
   \SetKwBlock{Case}{case}{end}

   \Vars : \;
   \ \ \ $\st$ \AComment{Process state}

   \Function($\execute ❪ \operations ❫$ \label{algE:execution_handler}){
      \Foreach($\operations_j  \in  \mathit{order} ❪\operations❫$\label{algE:order}){
         \Foreach($o  \in  \operations_j$){
            \Match($o$\label{algE:match-op}){
               \Case($\Transaction ❪ p, t ❫  \  \Rightarrow $){
                  $ \langle \st, v \rangle   \leftarrow  t(\st)$ \label{algE:apply-trans} \;
                  \lIf{$p = \self$}{$\response \ \mathit{return}(t, v)$ \label{algE:retv}}
               }
               
               \Case($\Reconfig ❪ \rc ❫  \  \Rightarrow $){
                  $\call \ \reconfigure(j, \rc)$\label{algE:apply-reconf}  \; 
                  \If{$j = \i$}{
                     $\mathit{irc}  \leftarrow  \rc$ \label{algE:newline}
                  }                  
               }
            }
         }
      }
      $\call \ \mathit{kickstart(irc)}$ \label{algE:callkickstart} \;
            
      $\preOps  \leftarrow \operations_j; \ \ \preCerts  \leftarrow  \certs$ \label{algE:save-pre} \;
      \Foreach(cluster $C_{j}$){
         reset $\timer_{j}$ \; 
         $\operations_{j}  \leftarrow  \emptyset ; \ \ \certs  \leftarrow  \emptyset $\;
         $\cn_{j}  \leftarrow  \rcn_{j}  \leftarrow  0$ \;
      }
      $r  \leftarrow  r + 1$ \label{algE:reset_vars} \;
   }

   \Function($\reconfigure$\mbox{$(j, \rc)$} \label{algJ:install_handler}){
   \ACommentL{Function $\reconfigure $ is called in Stage 3.}
      \Foreach($o  \in  \rc$){
         \Match($o$){
            \Case($\join❪p❫  \,  \Rightarrow $\label{algJ:case1}){
               $C_{j}  \leftarrow  C_{j} ∪ \{p\}$ \label{algJ:install_join} \;
            }
            \Case($\leave❪p❫  \,  \Rightarrow $\label{algJ:case2}){
               $C_{j}  \leftarrow  C_{j} ∖ \{p\}$ \label{algJ:install_leave} \; 
            }
            $f_{j} = \lfloor (|C_{j}| -1) / 3 \rfloor$ \label{algJ:update_other_var}\;
         }
      }
   }

   \Function($\mathit{kickstart}❪rc❫$ \label{algJ:install_handler2}){
      \Foreach($o  \in  \rc$ \mbox{\ \ \ \ $ \triangleright $ \ First joins and then leaves.}){
         \Match($o$){
            \Case($\join❪p❫  \,  \Rightarrow $\label{algJ:casep1}){
               $\apl \ \request \ \send(p, \CurrState(\st, C, r))$ \label{algJ:share_state} \;
            }
            \Case($\leave❪p❫  \,  \Rightarrow $\label{algJ:casep2}){
               \lIf{$p = \self$\label{algJ:install_leave-resp}}{$\response \ \left$}
            }
         }
      }
      
         $\recs  \leftarrow  \recs ∖ \{ rc \}$ \;
         $\echoed  \leftarrow \readied  \leftarrow  \delivered  \leftarrow  \false$ \;
         $\valid  \leftarrow  \bot$ \;
   }

   \vspace{0.7ex}  
	\Upon($\apl \ \response \ \deliver ❪\overline{p}, \CurrState(s',$ $C',$ $r' ❫❫$ where $|\{\overline{p}\}|  \geq  2  \times  f_\i + 1$ \label{algJ:reply_received}){
	    $\mathit{state}  \leftarrow  s';
	    \ \
	    r  \leftarrow  r'; 
	    \ \ 
	    C  \leftarrow  C'$
	    \;
	    $\response \ \joined$ \label{algJ:install_join-resp} \;
	}

\end{algorithm}

\section{Correctness}
\label{sec:correctness-app}

We now state the correctness properties of the sub-protocols and then the end-to-end protocol. The proofs are available in the extended report~\cite{reconfig}.

\noindent \textbf{Remote Leader Change.}
\vspace{-1mm}

\begin{lemma}[Eventual Succession]
\label{lem:validity1}
Let $ops$ be the locally replicated operations of a cluster $C$ in a round.
Either $ops$ are eventually delivered to all correct processes of every other cluster in that round,
or correct processes in $C$ eventually adopt a new leader.

\end{lemma}

\begin{lemma}[Eventual Agreement]
\label{lem:leader-agreement}
All correct processes in the same cluster eventually adopt the same leader.

\end{lemma}
 
\begin{lemma}[Overthrow Resistance]
\label{lem:remote-leader-integrity}
A correct process does not adopt a new leader unless 
at least one correct process complains about the previous leader.

\end{lemma}

\noindent \textbf{Inter-cluster Broadcast.}
\vspace{-1mm}
 
\begin{lemma}[Termination]
	\label{lem:inter-termination-main}
	\label{lem:inter-termination}
In every round, every correct process eventually receives operations from each other cluster.

\end{lemma}

\begin{lemma}[Agreement]
	\label{lem:inter-agreement-main}
	\label{lem:inter-agreement}
In every round, the operations that every pair of correct processes receive from a cluster are the same.

\end{lemma}

\noindent \textbf{Byzantine Reliable Dissemination.}
\vspace{-1mm}
%

\begin{lemma}[Integrity]
\label{lem:brd_integrity}
Every delivered set contains at least a quorum of messages from distinct processes.
\end{lemma}

\begin{lemma}[Termination]
\label{lem:brd_termination}
If all correct processes broadcast messages then every correct process eventually delivers a set of messages.
\end{lemma}

\begin{lemma}[Uniformity]
\label{lem:brd_uniformity}
No correct pair of processes deliver different sets of messages.
\end{lemma}

\begin{lemma}[No duplication]
\label{lem:brd_no_duplication}
Every correct process delivers at most one set of messages.
\end{lemma}

\begin{lemma}[Validity]
\label{lem:brd_validity}
If a correct process delivers a set of messages containing $m$ from a correct sender $p$, then $m$ was broadcast by $p$.

\end{lemma}

\noindent \textbf{Reconfiguration.}
\vspace{-1mm}

\begin{lemma}[Completeness]
\label{lem:reconfig-completeness-main}
\label{lem:reconfig-completeness}
If a correct process $p$ requests to join (or leave) cluster $i$, 
then every correct process will eventually have a configuration $C$ such that
$p \in C$ (or $p \notin C$).

\end{lemma}

\begin{lemma}[Accuracy]
\label{lem:reconfig-accuracy-main}
Consider a correct process $p$ that has a configuration $C$ in a round, and then another configuration $C'$ in 
a later round.
If a correct process $p \in C_i' \setminus C_i$, then $p$ requested to join the cluster $i$.
Similarly, if 
a correct process $p \in C_i \setminus C_i'$, then $p$ requested to leave the cluster $i$.

\end{lemma}

\begin{lemma}[Uniformity]
\label{lem:reconfig-agreement-main}
In every round, 
the configurations that every pair of correct processes
execute
are the same.

\end{lemma}

\noindent \textbf{Reconfigurable Clustered Replication.}

\noindent
\textbf{Theorem \ref{lem:general-validity-main}} (Validity). \
\textit{Every operation that a correct process requests
is eventually executed by a correct process.
}

\noindent
\textbf{Theorem \ref{lem:general-termination-main}} (Agreement). \
\textit{If a correct process executes an operation in a round
then every correct process executes that operation in the same round.
}

\noindent
\textbf{Theorem \ref{lem:general-total-order-main}} (Total order). \
\textit{For every pair of operations $o$ and $o'$,
if a correct process executes only $o$, or executes $o$ before $o'$,
then every correct process executes $o'$ only after $o$.
}

%
%
%
%

\section{Proofs}
\label{sec:proofs}
 
 \subsection{Remote Leader Change}

\textbf{Lemma \ref{lem:validity1}} (Eventual Succession). \ 
\textit{Let $ops$ be the locally replicated operations of a cluster $C$ in a round.
Either $ops$ are eventually delivered to all correct processes of every other cluster in that round,
or correct processes in $C$ eventually adopt a new leader.
}

 \begin{proof}
 	Let $C_2$ be any other cluster in the system except $C$. There are two cases regarding the delivery of $m$ in cluster $C_2$.
 	
 	In the first case, at least one correct process $p$ in $C_2$ delivers $m$. 
 	Then, it uses $\rb$ to broadcasts $m$ to all members of the local cluster at \autoref{alg:inter-cluster}, \autoref{algI:local_send}. 
 	By validity of reliable broadcast, all the correct processes in $C_2$ deliver $m$.
 	
 	In the second case, none of the correct processes in $C_2$ delivers $m$. 
 	We prove that processes in $C_2$ will invoke a remote leader change for $C$ and finally correct processes in $C$ adopt a new leader.
 	If none of the correct processes of $C_2$ delivers $m$, 
 	then their timers will eventually be triggered at \autoref{alg:het-remote-leader-change}, \autoref{algV:timer_trigger} and all of the correct processes broadcast $\LComplaint$ at \autoref{algV:Drvc_send}.
 	Thus, the signatures of all of them are stored in $\cs_{1}$ variable at \autoref{algV:add_complain}.
 	Since there are at least $2  \times  f_2 + 1$ correct processes in cluster $C_2$,
 	all the correct processes eventually receive enough $\LComplaint$ messages,
 	and $\cs_{1}$ will be large enough.
 	Thus, $f_2 + 1$ processes in $C_{2}$ send $\RComplaint$ messages,
 	and each send it to $f + 1$ distinct processes in $C$ at \autoref{algV:Rvc_send}.
 	Thus, at least one correct process in $C$ eventually delivers the $\RComplaint$ message at \autoref{algV:Rvc_deliver} and verifies the validity of 
 	the accompanying signatures $\Sigma$.
 	Then, it broadcasts the $\Complaint$ message locally at \autoref{algV:rb_Rvc}.
 	By validity of $\abeb$, all the correct processes in $C$ deliver the complain at \autoref{algV:Rvc_rb_deliver}, 
 	and request the leader election module to move to the next leader at \autoref{algV:local_complain}.
 	Thus, the leader election module will eventually choose a new leader.
 	Thus, all the correct processes in $C$
 	will eventually adopt a new leader at \autoref{alg:leader-change} \autoref{algLC:new_leader}-\ref{algLC:record-leader-ts}.
 	
 \end{proof}

\begin{lemma}[Local Complaint Synchronization]
	\label{lem:local_complaint_counter}
	If a correct process in cluster $C_i$ installs $cn_j = k$, then all the correct processes in $C_i$ eventually install $cn_j = k$.
\end{lemma}

\begin{proof}
	We prove this lemma by induction. 
	
	For $cn_j = 0$, all the correct processes assign $cn_j$ to be the same value $0$ at initialization.
	
	The induction hypothesis is that if a correct process installs $cn_j = k$, then all the correct processes eventually install $cn_j =k$. 
	
	We prove that if a correct process in cluster $C_i$ installs $cn_j = k + 1$, then all the correct processes in $C_i$ eventually install $cn_j = k + 1$
	
	A correct process $p$ increments $cn_j$ to $k + 1$ at \autoref{algV:reset1} after verifying $2f_i + 1$ $\LComplaint$ messages has been delivered for the same $cn_j = k$ at \autoref{algV:2f_complain}.
	Thus at least $f_i + 1$ correct processes have broadcast $\LComplaint$ messages for $cn_j = k$.
	By the validity of $\abeb$, all the correct processes eventually delivers at least $f_i + 1$ consistent $\LComplaint$ messages at \autoref{algV:f_complain} and verify the complaint counter:
	by induction hypothesis and $p$ installed $cn_j = k$, all the correct processes eventually install $cn_j = k$.
	Then correct processes amplify the complain by broadcasting $\LComplaint$ messages for $cn_j = k$ at \autoref{algV:amplify_Drvc}.
	There are at least $2f_i + 1$ correct processes in cluster $C_i$.
	By the validity of $\abeb$, eventually at least $2f_i + 1$ $\LComplaint$ messages are delivered to all correct processes at \autoref{algV:2f_complain} and they increment $cn_j$ to $k + 1$ at \autoref{algV:reset1}.
	Therefore, all the correct processes install $cn_j = k + 1$.
	
	We conclude the induction proof: for $k \geq 0$, if a correct process install $cn_j = k$, then all the correct processes install $cn_j = k$. 
\end{proof}

\begin{lemma}[Remote Complaint Synchronization]
	\label{lem:remote_complaint_counter}
	If a correct process in cluster $C_i$ installs $rcn_j = k$, then all the correct processes in $C_i$ eventually install $rcn_j = k$.
\end{lemma}

\begin{proof}
	We prove this lemma by induction. 
	
	For $rcn_j = 0$, all the correct processes assign $rcn_j$ to be the same value $0$ at initialization.
	
	The induction hypothesis is that if a correct process in cluster $C_i$ installs $rcn_j = k$, then all the correct processes eventually install $rcn_j =k$. 
	
	We prove that if a correct process $p$ in cluster $C_i$ installs $rcn_j = k + 1$, then all the correct processes in $C_i$ eventually install $rcn_j = k + 1$
	
	A correct process $p$ increments $rcn_j$ to $k + 1$ at \autoref{algV:rcn-inc} after verifying $2f_j + 1$ $\LComplaint$ messages was in $\Sigma$ for the same $rcn_j = k$ at \autoref{algV:2f_complain}.
	Thus by \autoref{lem:local_complaint_counter} and $\Sigma$ verifies that a correct process in $C_j$ installed $cn_j = k + 1$, all the correct processes in $C_j$ eventually install $cn_j = k + 1$ at \autoref{algV:reset1}.
	There are at most $f_j$ Byzantine processes in cluster $C_j$ and $\S$ contains $f_j + 1$ processes, therefore at least one correct process in $\S$ sends $\RComplaint(k, j, \Sigma, r)$ messages to $f_i + 1$ processes in $C_i$.
	By the validity of $\apl$, at least one correct process in $C_i$ receives the $\RComplaint$ message at \autoref{algV:Rvc_deliver}  and broadcasts $\Complaint(k, j, \Sigma)$ message at \autoref{algV:rb_Rvc}.
	By the validity of $\abeb$, all the correct processes in $C_i$ eventually delivers $\Complaint$ messages at \autoref{algV:Rvc_rb_deliver} and verify the complaint counter:
	by induction hypothesis and $p$ installed $rcn_j = k$, all the correct processes eventually install $rcn_j = k$.
	They increment the remote complaint counter $rcn_j$ to $k + 1$ \autoref{algV:rcn-inc}.
	Therefore, all the correct processes install $rcn_j = k + 1$.
	
	We conclude the induction proof: for $k \geq 0$, if a correct process install $rcn_j = k$, then all the correct processes install $rcn_j = k$. 
\end{proof}

\textbf{Lemma \ref{lem:leader-agreement}} (Eventual Agreement). \
All correct processes in the same cluster eventually adopt the same leader.


\begin{proof}
	We prove this lemma in three steps. 
	Firstly, we prove if a correct process in $C_i$ issue response $\newLeader$ for $\ts$, then eventually all correct process in $C_i$ issue response $\newLeader$ for $\ts$.
	Secondly, we prove that eventually all the correct process stop changing leader and stay in the same timestamp.
	Finally, since the leader is deterministically chosen according to the timestamp and cluster membership, we prove that eventually all the correct process eventually adopt the same leader.
	
	For the first statement, $\lem$ issue response for two type of requests: $\complain$ and $\nextLeader$.
	For $\complain$ request, we directly use the eventual agreement property of underlying module.
	For $\nextLeader$ request, a correct process $p$ in cluster $C_i$ requests a $\nextLeader$ at \autoref{algV:local_complain}.
	Let us assume that $p$ installs $rcn_j = n$ before the $\nextLeader$ request at \autoref{algV:rcn-inc}.
	By \autoref{lem:remote_complaint_counter}, all the correct processes in $C_i$ eventually install $rcn_j = n$.
	By assumption, this request is apart from the previous remote leader change events and $\Delta - \timer_i > \epsilon$.
	Then all the correct process request the $\nextLeader$ for the same $\Complaint$ message.
	Therefore the $\ts$ at all correct processes are eventually the same.
	
	For the second statement, correct processes eventually wait long enough for a correct leader to complete inter-broadcast stage: the timer for remote leader change increases exponentially and eventually, all the messages are delivered within a bounded delay after GST. 
	When all $\Complaint$ messages have been received, all the correct processes in the same cluster don't issue new complains and by \autoref{lem:remote-leader-integrity}, they stay in the same $\ts$.
	
	For the third statement, by \autoref{lem:reconfig-agreement-main}, all the correct processes in the same cluster maintain a consistent group membership for each round. Then all of them deterministically choose the same process as leader based on group member and timestamp.
\end{proof}

\textbf{Lemma \ref{lem:remote-leader-integrity}} (Overthrow resistance). \
A correct process does not adopt a new leader unless 
at least one correct process complains about the previous leader.


\begin{proof}
	The correct process requests the leader election module to adopt the next leader at \autoref{alg:het-remote-leader-change}, \autoref{algV:local_complain}.
	This request is after receiving a $\Complaint$ message at \autoref{algV:Rvc_rb_deliver} with the following checks: 
	(1) the expected next complaint counter $\rcn_j$ is equal to the received complain number $c$, and 
	(2) the signatures $\Sigma$ include at least $2  \times  f_j + 1$ signatures from $C_j$.
	The first check prevents replay attacks; thus, no complaints about previous leaders can be reused.
	Therefore, all the signatures in $\Sigma$ are complaints for the current leader.
	The second one implies that
	a correct process in $C_j$ sent the $\RComplaint$ message after receiving $2  \times  f_j + 1$ $\LComplaint$ messages at \autoref{algV:2f_complain}.
	Thus, at least $f_j + 1$ correct processes sent $\LComplaint$ messages.
	A correct process sends a $\LComplaint$ message at two places: 
	(1) the timer triggers at \autoref{algV:timer_trigger}; 
	(2) the process amplifies the received complaints at \autoref{algV:f_complain}.
	The first case reached the conclusion.
	In the second case, a correct process only amplifies after receiving $f_j + 1$ $\LComplaint$ messages.
	Thus, at least one correct process sent a $\LComplaint$ message with the same two cases as above.
	This second case is the inductive case, and the first case is the base case.
	Since the number of processes is finite, by induction,
	this case is reduced to the first case in a finite number of steps.
\end{proof}

\subsection{Inter-cluster Broadcast}

\textbf{Lemma \ref{lem:inter-termination-main}} (Inter Broadcast Termination). \
In every round, every correct process eventually receives operations from each other cluster.


\begin{proof}
	We prove the termination property for inter-cluster broadcast with the help of \autoref{lem:validity1}.
	A leader of cluster $i$ should send $\Inter$ message to $f_j + 1$ processes in cluster $j$ for all $i \neq j$ at \autoref{algI:inter_send}. By the validity of remote leader change, either this $\Inter$ message was delivered to all correct processes in cluster $j$ or all the correct processes in cluster $i$ change a leader.
	In the first case we conclude the proof.
	In the second case, eventually the correct processes in cluster $i$ adopt a correct leader.
	The correct leader sends $\Inter$ messages to $f_j + 1$ processes in cluster $j$.
	By the validity of  $\apl$, at least one correct process $p$ in cluster $j$ delivers the $\Inter$ message at \autoref{algI:inter_received}.
	Then $p$ broadcasts the received content in $\Local$ message at \autoref{algI:local_send}.
	By the validity of $\abeb$, all the correct processes in cluster $j$ eventually deliver the $\Local$ message at \autoref{algI:local_handler}.
	We generalize the same reasoning for all the other cluster and conclude the proof.
\end{proof}

\textbf{Lemma \ref{lem:inter-agreement-main}} (Inter Broadcast Agreement). \
In every round, the operations that every pair of correct processes receive from a cluster are the same.


\begin{proof}
	Let process $p$ receives $\Local(r, j, ops, \Sigma)$ and $p'$ receives $\Local(r, j, ops', \Sigma')$.
	Correct processes only delivery valid $\Local$ messages, which means $\Sigma$ attests $\ops$ and $\Sigma'$ attests $\ops'$.
	Then $\Sigma$ and $\Sigma'$ both contains $2f+1$ commit signatures for each operation in $\ops$ and $\ops'$.
	By the agreement property of TOB in the first stage and $|\ops| = |\ops'|$, $\ops$ and $\ops'$ contains the same set of operations.
	By the total order property of the TOB, operations in $\ops$ and $\ops'$ have the same order.
	Thus, $\ops = \ops'$.
\end{proof}

\subsection{Byzantine Reliable Dissemination}
\label{sec:proofs-brd}

\begin{lemma}[Integrity]
	The delivered set contains at least a quorum of messages from distinct processes.
\end{lemma}

\begin{proof}
	 A set of messages is delivered at \autoref{algJ:issue_deliver} 
	 which is after the delivery of $2f_i + 1$ of $\ready$ messages (at \autoref{algJ:2f_ready_delivered}).
	At least $f_i + 1$ correct processes sent $\ready$ messages since there are only $f_i$ Byzantine processes in a cluster $i$.
	A correct process only sends $\ready$ message when it receives $2f_i + 1$ $\echo$ messages or $f_{i} + 1$ $\ready$ messages.
	Then by induction, at least $2f_i + 1$ $\echo$ messages were received by a correct process.
	Then at least $f_{i} + 1$ correct processes sent $\echo$ messages.
	A correct process only sends $\echo$ messages when it verifies $M$ is valid (at \autoref{algJ:proposal_delivered}).
	A $M$ is valid if and only if $\Sigma$ includes $2f_i + 1$ distinct signatures and $M$ is the union of all the $m$ sets in those messages;
	Or $M$ is adopted from the $\valid$ and $\Sigma$ contains $2f_i + 1$ $\echo$ or $f_i + 1$ $\ready$ messages.
	In the first case, the delivered $M$ contains at least a quorum of $m$.
	In the second case, by induction $M$ was in $2f_i + 1$ of $\echo$ messages
	and the correct processes who sent the $\echo$ message verify that $M$ originally was a union of $2f_i + 1$ $m$. 
\end{proof}

\begin{lemma}[Termination]
	If all correct processes broadcast messages then every correct process eventually delivers a set of messages.
\end{lemma}

\begin{proof}
	We consider two cases based on whether there is a correct process delivered a set of messages.
	
	Case 1: If there is a correct process that delivers, then eventually all the correct processes deliver.
	A correct process delivers $M$ after receiving $2f_i + 1$ $\ready$ message at \autoref{algJ:2f_ready_delivered}.
	Then at least $f_i + 1$ correct processes broadcast the $\ready$ message at \autoref{algJ:send_ready}.
	By the validity of $\abeb$, eventually all the correct processes deliver $f_i + 1$ $\ready$ message at \autoref{algJ:f_ready_delivered} and broadcast the same message at \autoref{algJ:amplify_ready}.
	Eventually, all the correct processes deliver $2f_i + 1$ $\ready$ messages and issue delivery response (at \autoref{algJ:issue_deliver}).

	Otherwise, Case 2: 
	if no correct process delivers, 
	then each correct process complains about the current leader.
	Then by the eventual agreement property of the Byzantine leader election, 
	all the correct processes eventually adopt the same correct leader.
	Upon the last leader election delivered at \autoref{algJ:brc-new-leader}, all the correct processes send $\Valid$ or $\mym$ to the correct leader at \autoref{algTOB:newleadercollect_send} or \autoref{algJ:send_collect-p}.
	Since the set of correct processes is a quorum, then the correct leader either delivers a quorum of $\mym$ messages at \autoref{algJ:leader_converge} or a $\Valid$ message at \autoref{algTOB:newleadercollect_received}.
	Then we have two cases, either there is a valid $\valid$ or not.
	In the first case, the correct leader adopts $M$ from $\valid$.
	In the second case, the correct leader composes a new set of reconfiguration requests.
	Both cases can be verified and accepted by correct processes at \autoref{algJ:proposal_delivered}.
	Then all the correct processes send $\echo$ message at \autoref{algJ:send_echo} and eventually $2f_i + 1$ $\echo$ messages are delivered to all the correct processes.
	Then all the correct processes send $\ready$ message at \autoref{algJ:send_ready} and eventually $2f_i + 1$ $\ready$ message are delivered to all correct processes.
	Then all the correct processes issue delivery response at \autoref{algJ:issue_deliver} and we conclude the proof.
\end{proof}

\begin{lemma}[Uniformity]
	No correct processes deliver different set of messages
\end{lemma}

\begin{proof}
	There are two cases regarding the delivery of messages for $p_1$ and $p_2$: either they deliver messages with the same $\ts$ or different $\ts$.
	
	In the first case, since any pair of quorums has a correct process in the intersection, if $p_1$ delivers $M_1$ and $p_2$ delivers $M_2$, $M_1 = M_2$.
	Otherwise, the correct process sends different $\ready$ messages for the same round and $\ts$, which is not permitted by the protocol (at \autoref{algJ:send_ready}, \autoref{algJ:amplify_ready}). 
	
	In the second case, let us assume that $p_1$ delivers first with timestamp $\ts_1$ and then $p_2$ delivers with another timestamp $\ts_2$.
	Without losing generality, let us assume that $\ts_1 < \ts_2$.
	If $p_1$ delivers $M_1$ with $\ts_1$, then $p_1$ receives at least a quorum of $\ready$ messages.
	A correct process set its $\valid$ before sending $\ready$ messages (at \autoref{algJ:update_rconfigcandidate1}, \autoref{algJ:update_rconfigcandidate2}).
	Therefore, at least $f_i + 1$ correct processes set their $\valid$ variable with $M_1$.
	For the next timestamp $\ts_1 < \ts_i  \leq  \ts_2$, it collects a quorum of $\mym$ messages or at least one $\Valid$ message.
	By assumption, cluster $i$ has $3f_i + 1$ members in total, then at most $2f_i$ processes have not set $\valid$ and can send $\mym$ message, which is not a quorum.
	Therefore, the leader for $\ts_i$ waits for the $\Valid$ message and adopts its value.
	Valid $\valid$ requires either $2f_i + 1$ $\echo$ messages or $f_i + 1$ $\ready$ messages for the same $\ts$.
	By induction, since there are only $f_i$ Byzantine processes, a correct process receives $2f_i + 1$ $\echo$ messages before sending out $\ready$ messages and triggering the amplification. 
	Since any pair of quorums has a correct process in the intersection, there is only one $M$ that can be echoed by a quorum of processes and appears in $\valid$.
	The leader for $\ts_i$ can only propose $M_1$ that will be accepted by correct processes at \autoref{algJ:proposal_delivered}.
	From $\ts_i$ to $\ts_2$, the $\valid$ can only be updated to the same $M_1$.
	Then when $p_2$ delivers $M_2$ in $\ts_2$, $M_2 = M_1$.
\end{proof}

\begin{lemma}[No duplication]
	Every correct process delivers at most one set of messages
\end{lemma}

\begin{proof}
	This lemma follows directly from the condition (at \autoref{algJ:2f_ready_delivered}) before the delivery response is issued at \autoref{algJ:issue_deliver}.
\end{proof}

\begin{lemma}[Validity]
	If a correct process delivers a set of messages containing $m$ from a correct sender $p$, then $m$ was broadcast by $p$
\end{lemma}

\begin{proof}
	If a correct process delivers a set of messages, then it receives a quorum of $\ready$ messages. 
	A ready message is send by a correct process if it receives a quorum of $\echo$ messages or $f+1$ ready messages.
	Since there are only $f$ Byzantine processes, then by induction, the first ready message sent by a correct process is because of receiving a quorum of echo messages.
	A correct process only send echo message if delivers the  $\Agg$ from the leader with valid certificate.
	A valid $\Agg$ message states that $M$ is either collected from a quorum of distinct processes through $\apl$ or adopted from the previous leader.
	For the first case, by the validity of $\apl$, if the sender of $m$ is correct, then it sends $m$ to the leader.
	For the second case, $M$ can be adopted only if it carries a certificate with a quorum of $\echo$ messages for $M$ or $f+1$ $\ready$ messages for $m$.
	By the same induction, the messages contained in $M$ is broadcast by its sender $p$ if $p$ is correct.
\end{proof}

\subsection{Reconfiguration}

\textbf{Lemma \ref{lem:reconfig-completeness-main}} (Completeness). \
If a correct process $p$ requests to join (or leave) cluster $i$, 
then every correct process will eventually have a configuration $C$ such that
$p \in C$ (or $p \notin C$).


\begin{proof}
	We prove the completeness in two steps:
	first we prove that all the reconfiguration requests will be in a prepared state which we will formally define later;
	then we prove that all the prepared reconfiguration requests will be delivered within one round.
	
	We define that a new process prepares a join request when it receives at least a quorum of replies from the existing replicas.
	Our protocol guarantees that a new process officially joins the system in the round it is prepared. 
	Similarly, we define a leaving process that prepares a leave request when its $\RequestLeave$ message has been delivered to a quorum of existing replicas.
	Our protocol guarantees that a leaving process officially leaves the system in the round it is prepared. 
	
	For the first statement, 
	when a correct process $p$ requests to join (or leave) the cluster $C_\i$, it sends out $\RequestJoin$ (or $\RequestLeave$) messages to all the existing processes at \autoref{algJ:request_states} (or at \autoref{algJ:request_leave}).
	If $p$'s request is not installed in a long time \autoref{algJ:join_timer_expires},
	it resends the $\RequestJoin$ (or $\RequestLeave$) message and doubles the timer at \autoref{algJ:repeat_join} (or \autoref{algJ:request_leave-repeat}). 
	Therefore $\RequestJoin$ (or $\RequestLeave$)  messages sent out by $p$ at \autoref{algJ:request_states} will be delivered at all the correct processes in $C_{i}$ in the first stage at \autoref{algJ:requestS_delivered} after GST.
	Upon receiving the $\RequestJoin$ and $\RequestLeave$ message at \autoref{algJ:update_tentative} and \autoref{algJ:update_tentative2},
	correct processes in the system add the reconfiguration request into their $\recs$ variable.
	Since all the correct processes in a cluster is a quorum, $p$'s reconfiguration request is eventually prepared.

	We prove the second statement in two steps.
	First, we prove that any set of installed reconfiguration requests at round $r$ includes $p$'s reconfiguration request.
	Second, we prove that eventually, all correct processes install a set of reconfiguration requests in round $r$.
	
	For the first step, 
	at the end of the local ordering stage of each round at \autoref{algJ:send_collect},
	correct processes use Byzantine reliable dissemination module to deliver 
	the reconfiguration requests $\recs$ that they have collected.
	Assume that $p$'s reconfiguration request is prepared in round $r$.
	By the integrity of BRD, the delivered set contains a quorum of messages send by distinct processes.
	Since every pair of quorums have at least one correct process in their intersection,
	at round $r$,
	there is always a correct process which sends $p$'s reconfiguration request in the BRD message and the message is included in the delivered set.
	
	For the second step, we consider the delivery of reconfiguration requests for both local and remote clusters.
	
	For the remote clusters, by \autoref{lem:inter-termination-main} all the correct processes in the remote cluster deliver $\Local$ message, which is verified to contain reconfiguration requests at \autoref{algI:local_handler}. 
	Correct processes eventually receives all the $\Local$ message at \autoref{algI:receive_all_clusters} and install reconfiguration at \autoref{algI:execution_request}.
	
	For the local cluster, 
	by the termination property of BRD, all the correct nodes in the local cluster eventually deliver a set of reconfiguration requests through BRD at \autoref{algJ:deliver-collection}.
	They insert the reconfiguration requests at \autoref{algJ:insert}.
	By \autoref{lem:inter-termination-main}, all the correct processes receive enough $\Local$ message and install the reconfiguration requests at \autoref{algJ:install_handler}.
	
	In conclusion, a set of reconfiguration requests is eventually installed at all the correct processes and we conclude the second step.
\end{proof}

\textbf{Lemma \ref{lem:reconfig-agreement-main}} (Uniformity). \
In every round, 
the configurations that every pair of correct processes
execute
are the same.


\begin{proof}
	Let us assume that two correct processes $p_1$ and $p_2$ installed new configurations.
	The correct process installs new group membership at \autoref{algJ:install_handler}, which is at the order and execution stage.
	We prove agreement for correct processes in both local and remote clusters.
	
	For the local cluster, a correct process installs a reconfiguration request from $\operations_{i}$ at \autoref{algE:apply-reconf}.
	$\operations_i$ is updated at \autoref{algJ:insert}, which is after the delivery of an instance of BRD at \autoref{algJ:deliver-collection}.
	By the uniformity property of BRD, all the correct processes deliver the same set of messages.
	Since the installation of new membership is deterministic and only dependent on the set of reconfiguration requests, we have $C = C'$.
	
	For remote cluster reconfiguration, a correct process in cluster $i$ installs the reconfiguration requests for cluster $j$ at the order and execution stage at \autoref{algJ:install_handler}.
	$\operations_{j}$ is updated after verifying the $ \sigma $ at \autoref{algI:local_handler}.
	$ \sigma $ is valid if and only if for each reconfiguration request in $T$, it contains a quorum of signatures from cluster $j$ in round $r$.
	By \autoref{lem:inter-agreement-main}, the reconfiguration requests installed at cluster $j$ are the same. 
	Therefore, we conclude $C = C'$.

\end{proof}

\textbf{Lemma \ref{lem:reconfig-accuracy-main}} (Accuracy). \
Consider a correct process $p$ that has a configuration $C$ in a round, and then another configuration $C'$ in 
a later round.
If a correct process $p \in C_i' \setminus C_i$, then $p$ requested to join the cluster $i$.
Similarly, if 
a correct process $p \in C_i \setminus C_i'$, then $p$ requested to leave the cluster $i$.


\begin{proof}
	Since we have $p_n  \in  C_2 \setminus C_1 \land r_2 > r_1$, $p_n$ is not originally a member of this cluster.
	The cluster membership is updated at \autoref{algJ:install_join}, which is after verifying each reconfiguration request is valid:
	each reconfiguration request is delivered after a quorum of $\ready$ messages.
	At \autoref{algJ:proposal_delivered}, every correct process checks the validity of $\rc$, including its signatures from $p_n$. 
	By the authenticity of $\apl$, if $p_n$ is correct, then it is the sender of the $\RequestJoin$ messages and thus requested to join.
	The same reasoning applies to $\leave$ requests.
	
\end{proof}

\subsection{Replication System}

By \autoref{lem:reconfig-agreement-main}, at the beginning of each round, all the correct processes have the same configuration.
Thus during the execution of each round, all the correct processes maintain a static membership and we prove termination and total order properties for each round.
We prove validity for eventual progress.

\textbf{Theorem \ref{lem:general-validity-main}} (Validity). \
Every operation that a correct process requests
is eventually executed by a correct process.


\begin{proof}
	Based on the validity of the underlying TOB protocol in the first stage, if a valid operation $o$ is submitted to a cluster $i$ (at \autoref{algTOB:local_handler}), then $o$ is eventually delivered at a correct process $p$ in $C_i$ at \autoref{algTOB:TOB_received} and included in $\operations_i$ (at \autoref{algTOB:append-trans}).
	By the \autoref{lem:inter-termination-main}, each correct process receives $\Local$ message from each other cluster and call $\execute$ function (at \autoref{algE:execution_handler}).
	Since $o$ was included in $\operations_i$, it is executed at \autoref{algE:apply-trans} and we conclude the proof.
\end{proof}

\textbf{Theorem \ref{lem:general-termination-main}} (Agreement). \
If a correct process executes an operation in a round
then every correct process executes that operation in the same round.


\begin{proof}
	A correct process deliver a operation in the execution stage (at \autoref{alg:order-exec}), which is stored in $\operations$.
	$\operations$ are updated in the inter-cluster stage (at \autoref{alg:inter-cluster}) for remote clusters and in the local ordering stage (at \autoref{algTOB:append-trans}) for the local cluster.
	Based on whether $o$ is an operation from the local cluster, we prove the termination in two cases.
	
	Case 1: $o$ is from the local cluster. Then we prove that $o$ will be delivered locally and remotely in round $r$.
	For the local cluster, we can directly use the termination property provided by the underlying TOB protocol: all the correct processes eventually deliver $o$.
	Since correct processes in the local cluster are waiting for a fixed number of operations to be delivered in a batch for each round, they will not move to the next round before they deliver $o$ in round $r$.
	Then we proved that $o$ will be delivered locally in round $r$.
	
	For the remote delivery of $o$, by the total order and termination property of underlying TOB protocol, $o$ will be delivered at the leader and included in the $\Local$ message.
	Then by the \autoref{lem:inter-termination-main}, all the other remote clusters will receive a $\Local$ message for each cluster, including the current one.
	 Thus, $o$ will be delivered in the $\Local$ message and inserted to $\operations$.
	 Finally, after all the $\Local$ messages are delivered, $o$ will be executed in the execution stage.
	
	Case 2: $o$ is from a remote cluster. Then we prove that $o$ will be delivered at all the other clusters. 
	
	If $o$ is from a remote cluster, then $\operations$ is only updated if $\Sigma$ is valid and the deliver is for the same round.
	$\Sigma$ is valid if it contains a quorum of commit certificate for each operation in $\ops$:
	a quorum of commit messages certify the delivery in the local order protocol.
	$\operations$ is updated when receiving a valid $\Local$ message.
	Then by the \autoref{lem:inter-agreement-main} and \autoref{lem:inter-termination-main}, $o$ will be delivered at all the other clusters through the same $\Local$ message. 
	
\end{proof}

\textbf{Theorem \ref{lem:general-total-order-main}} (Total order). \
For every pair of operations $o$ and $o'$,
if a correct process executes only $o$, or executes $o$ before $o'$,
then every correct process executes $o'$ only after $o$.


\begin{proof}
	We prove the total order property in two steps.
	
	First we prove that if a pair of processes $p_1$ and $p_2$ both execute $o$ and $o'$, then they execute $o$ and $o'$ in the same order.
	Without loss of generality, let us assume that $o$ is executed in $p_1$ before $o'$
	By \autoref{lem:general-termination-main}, all the correct processes deliver the same operations for each cluster.
	Then they combine the operations in the predefined order based on the cluster identifier.
	Within each cluster, $\ops$ have been ordered across all the correct processes by the total order property of the underlying TOB protocol.
	Thus the combined operations keeps a total order across all the operations from all the clusters for round $r$: $o$ is executed before $o'$ at all correct processes including $p_2$.
	
	In the second step, we prove by contradiction.
	Assume that process $p'$ executed operation $o'$ before operation $o$.
	The process $p$ executed only $o$, or executed $o$ before $o'$.
	In the case that it has executed only $o$ and not $o'$, then, by the \autoref{lem:general-termination-main}, it will eventually execute $o'$ after $o$. Thus, we will reach a state where $p$ and $p'$ have a different order for the two operations $o$ and $o'$, which contradicts the first statement.
\end{proof}

\end{document}